\documentclass[twoside]{article}
\usepackage[a4paper]{geometry}
\usepackage[utf8]{inputenc}
\usepackage[T1]{fontenc} 
\usepackage{RR}
\usepackage{hyperref}
\usepackage{graphicx}
\usepackage{subfig}
\usepackage{caption}
\usepackage{amsfonts,latexsym,amssymb,amscd,epsf}
\usepackage{amsthm}
\usepackage[tbtags]{amsmath}
\usepackage{mathtools}
\usepackage{hyperref}
\usepackage{algorithm}
\usepackage{algpseudocode}
\usepackage{booktabs,multirow}

\usepackage{array}
\usepackage{makecell}

\usepackage{dsfont}
\DeclareMathAlphabet{\mymathbb}{U}{BOONDOX-ds}{m}{n}
\DeclareMathAlphabet\mathbfcal{OMS}{cmsy}{b}{n}
\newcommand{\numberset}{\mathbb}
\newcommand{\ten}[1]{\mathbf{#1}}

\newcommand{\0}{\mymathbb{0}}
\newcommand{\1}{\mymathbb{1}}
\newcommand{\I}{\mathbb{I}}

\newcommand{\R}{\numberset{R}}

\newcommand{\N}{\numberset{N}}

\newcommand{\Nabla}{\boldsymbol{\nabla}}

\newcommand{\kron}{\otimes_{\textsc{K}}}
\newcommand{\coreTT}[1]{\underline{\ten{#1}}}
\newcommand{\matcoreTT}[1]{\underline{{#1}}}
\usepackage{tensor}
\newcommand{\cont}[2]{\tensor[_{#1}]{\bullet}{_{#2}}} 
\newcommand{\dotprod}[2]{{\langle #1, #2\rangle}}
\newcommand{\norm}[1]{\Vert{#1} \Vert}

\newcommand{\ignore}[1]{}
\newcommand{\etab}{\eta_b}
\newcommand{\etaAb}[2]{\eta_{#1,#2}}

\newcommand{\Db}{\Delta b}
\newcommand{\DA}{\Delta A}

\DeclareMathOperator*{\argmin}{arg min}
\DeclareMathOperator*{\argmax}{arg max}
\DeclareMathOperator{\spn}{span}

\newcommand{\hypref}[2]{\hyperref[#2]{#1 \ref*{#2}}}
\newcommand{\mc}[1]{\mathcal{#1}}
\newcommand{\mbbR}{\mathbb{R}}

\newcommand{\reportPaper}[2]{#1} 

\usepackage{etoolbox} 

\DeclareMathAlphabet{\mathdutchcal}{U}{dutchcal}{m}{n}
\SetMathAlphabet{\mathdutchcal}{bold}{U}{dutchcal}{b}{n}
\newcommand{\set}[1]{\mathdutchcal{#1}}
 \newtheorem{theorem}{Theorem}[section]
 \newtheorem{corollary}[theorem]{Corollary}
 \newtheorem{lemma}[theorem]{Lemma}
 \newtheorem{proposition}[theorem]{Proposition}

 \newtheorem{remark}[theorem]{Remark}

\usepackage[style=numeric,  backend=bibtex, firstinits, maxbibnames=50,sortcites = true, url=false,eprint=false]{biblatex}
\bibliography{references}

\RRNo{9484}
\RRdate{September 2022}
\RRauthor{
  Olivier Coulaud\thanks{Inria, Inria centre at the University of Bordeaux}
  \and
  Luc Giraud\footnotemark[1]
  \and
  Martina Iannacito\footnotemark[1]
}
\authorhead{Coulaud et al.}
\RRtitle{
Un algorithme GMRES robuste au format tensor train
}
\RRetitle{A robust GMRES algorithm in Tensor Train format}
\titlehead{GMRES in Tensor Train}
\RRresume{
Nous considérons la résolution de systèmes linéaires avec une structure de produit tensoriel en utilisant un algorithme GMRES.
  Afin de faire face à la complexité de calcul en grande dimension, à la fois en termes d'opérations en virgule flottante et d'exigences de mémoire, notre algorithme est basé sur une représentation tensorielle à faible rang, à savoir le format Tensor Train.
  Dans un cadre d'analyse d'erreur inverse, nous montrons comment l'approximation tensorielle affecte la précision de la solution calculée.
  Dans une perspective d'erreur inverse, nous étudions les situations où le problème de dimension $(d+1)$ à résoudre résulte de la concaténation d'une séquence de problèmes de dimension $d$ (comme les problèmes d'opérateurs linéaires paramétriques ou de second membres paramétriques), nous fournissons des bornes d'erreur inverse pour relier la précision de la solution calculée de dimension $(d+1)$ à la qualité numérique de la séquence de solutions de dimension $d$ qui peut être extraite de celle-ci. Cela permet de prescrire un seuil de convergence lors de la résolution du problème à $(d+1)$ dimensions qui garantit la qualité numérique des solutions à $d$ dimensions qui seront extraites de la solution calculée en $(d+1)$ dimensions une fois que le solveur aura convergé.
  Les caractéristiques mentionnées ci-dessus sont illustrées sur un ensemble d'exemples académiques de dimensions et de tailles variables.
}
\RRabstract{
We consider the solution of linear systems with tensor product structure using a GMRES algorithm.
  In order to cope with the computational complexity in large dimension both in terms of floating point operations and memory requirement, our algorithm is based on low-rank tensor representation, namely the Tensor Train format.
  In a backward error analysis framework, we show how the tensor approximation affects the accuracy of the computed solution.
  With the bacwkward perspective, we investigate the situations where the $(d+1)$-dimensional problem to be solved results from the concatenation of a sequence of $d$-dimensional problems (like parametric linear operator or parametric right-hand side problems), we provide backward error bounds to relate the accuracy of the $(d+1)$-dimensional computed solution with the numerical quality of the sequence of $d$-dimensional solutions that can be extracted form it. This enables to prescribe convergence threshold when solving the $(d+1)$-dimensional problem that ensures the numerical quality of the $d$-dimensional solutions that will be extracted from the $(d+1)$-dimensional computed solution once the solver has converged.
  The above mentioned features are illustrated on a set of academic examples of varying dimensions and sizes.
}
\RRmotcle{ GMRES, backward stabilité, format Tenseur Train
}
\RRkeyword{
GMRES, backward stability, Tensor Train format}
\RRprojet{Concace}
\RCBordeaux 

\begin{document}
\makeRR   
\let\thefootnote\relax\footnotetext{
Distributed under a \href{https://creativecommons.org/licenses/by/4.0/}{Creative Commons Attribution 4.0} International License
}
\reportPaper{
\tableofcontents
\newpage
}{}

\section{Introduction}

In many domains in sciences and engineering, the problems to be solved can naturally be  modeled mathematically as $d$-{dimensional} linear systems with tensor product structure, i.e., as
\[
	\ten{A}\ten{x} = \ten{b}
\]
where $\ten{A}$ represents a multilinear endomorphism operator on $\R^{n_1\times \dots \times n_d}$, $\ten{b}\in\R^{n_1\times \dots \times n_d}$ is the right-hand side and $\ten{x}\in\R^{n_1\times \dots \times n_d}$ is the searched solution. 
Two main approaches have emerged in the search for methods to solve high-{dimensional} linear systems.
The first  approach is based on optimization techniques mainly based on Alternating Linearised Scheme such as ALS, MALS~\cite{Fom2011}  AMEN~\cite{Dolgov2014} and DMRG~\cite{Oseledets2011DMRG}, that break the $d$-{dimensional} linear system into low dimension minimization sub-problems, getting an high {dimensional} solution through an optimization process. The second approach focuses on how to generalize to high-{dimensional} linear systems iterative methods, as Krylov subspace methods, among which there are  conjugate gradient, Generalized Minimal RESidual (GMRES) and biconjugate gradient method~\cite{tobler2012}. 

Over the years, different attempts to extend iterative methods from classical matrix linear systems to high {dimensional} ones have been made, see~\cite{Dolgov2013, KressnerTobler2011,Ballani2013}. However, solving high {dimensional} linear systems is challenging, since the number of variables grows exponentially with the number of dimensions of the problem. To tackle this phenomenon, known in the tensor linear algebra community as `curse of dimensionality', there are several compression techniques, as High Order Singular Value Decomposition~\cite{DeLathauwer2000}, Hierarchical-Tucker~\cite{Grasedyck2009} and Tensor-Train (TT)~\cite{Oseledets2011}. These compression algorithms provide an approximation at a given accuracy of a given tensor, decreasing the memory footprint, but introducing meanwhile rounding errors. 
For the solution of such linear systems, the iterative methods have to rely heavily on compression techniques, to prevent memory deficiencies. Consequently, it is fundamental to take into account the effect of tensor rounding errors due to tensor recompression, when evaluating the numerical quality of the solution obtained from an iterative  algorithm. 

In this work, with a backward error perspective we investigate  the numerical performance of the Modified Gram-Schmidt GMRES (MGS-GMRES)~\cite{saad1986} for tensor linear systems represented through the TT-formalism. In the classical matrix context, it has been shown that MGS-GMRES is backward stable~\cite{paige2006} in the IEEE arithmetic, where the unit round-off $u$ bounds both the data representation and the rounding error of all the elementary floating point operations. In~\cite{agullo2022}, the authors pointed out numerically that the MGS-GMRES backward stability holds even when the data representation introduces component-wise or norm-wise perturbations, different from the unit round-off of the finite precision arithmetic. Differently for previously proposed versions of GMRES in tensor format~\cite{Dolgov2013}, this paper investigates numerically, through many examples, the backward stability of MGS-GMRES for tensor linear systems, where the TT-formalism introduces representation errors bounded by the prescribed accuracy of the computed solution.
 In particular, we consider the situation where either the right-hand side or the multilinear operator of the  $d$-{dimensional} system depends on a parameter. The tensor structure enables us to solve simultaneously for many discrete values of the parameter,  by simply reformulating the problem in a space of dimension $(d+1)$.
We establish theoretical backward error bounds to assess the quality of the $d$-{dimensional} solution extracted from the $(d+1)$-{dimensional} solution. This enables to define the convergence threshold to be used for the solution of the problem of dimension  $(d+1)$ that ensures the numerical quality of the $d$-{dimensional} solution extracted for the individual problem once MGS-GMRES has converged. We verify the tightness of these bounds through  numerical examples. 
We also investigate the memory consumption of our TT-GMRES algorithm. In particular we observe that, as it could have been expected  the memory requirement grows with the number of iterations and with the accuracy of the tensor representation, i.e., how accurate the tensor approximation is. From the memory viewpoint, MGS-GMRES in TT-format happens to be a suitable backward stable method for solving large high-{dimensional} systems, if the number of iterations remains reasonable or if a restart approach is considered. In our work, almost all the examples in TT-format are solved with a right preconditioned MGS-GMRES, to satisfy the prescribed tolerance in a small number of iterations. While we spend some words over the quality of the preconditioner, we do not study elaborated restarting techniques.

The remainder of this paper is organized as follows. In Section~\ref{sec2} we introduce the notation. Then we focus on GMRES, presenting the algorithm in a matrix computation framework. After introducing the TT representation, the MGS-GMRES algorithm in TT-format is fully described.
Next, in Section~\ref{ssec:Ai1}, we present our approach to solve simultaneously multiple linear systems, which share a common structure. We provide some theoretical results about the quality of the solution extracted from the simultaneous system solution.
Numerical experiments are reported in Section~\ref{sec3}, where we first illustrate the main features of the solver and compare its robustness to
the previous realization of GMRES in TT-format~\cite{Dolgov2013}. Then we illustrate the tightness of the bounds derived on Section~\ref{ssec:Ai1}
when solving parameter depend problems formulated in TT format.
\reportPaper{After summarizing the main results of our work in the conclusion, we investigate further the preconditioner choice and the convergence of problems with the same operator and different right-hand sides solved simultaneously in Appendix~\ref{app:prec} and~\ref{app:eig} respectively. The conclusive Appendix~\ref{app:tt} describes in details the construction of the $(d+1)$-dimensional linear system in TT-format from systems of dimension $d$.}
{We draw some concluding remarks in the final section.}

\section{Preliminaries on GMRES and tensors\label{sec2}}
For ease of reading, we adopt the following notations for the different mathematical objects involved in the description.
Small Latin letters stand for scalars and vectors (e.g., $a$), leaving the context to clarify the object nature. 
Matrices are denoted by capital Latin letters (e.g., $A$), tensors by bold small Latin letters (e.g., $\ten{a}$), the multilinear operator between two  spaces are calligraphic bold capital letter (e.g., $\mathbfcal{A}$) and the tensors representation of linear operators by bold capital Latin letters (e.g., $\ten{A}$).
We adopt the `Matlab notation' denoting by ``$\,:\,$''  all the indices along a mode. For example given a matrix $A\in\R^{m\times n}$, then $A(:, i)$ stands for the $i$-th column of $A$. The tensor product is denoted by $\otimes$ and Kronecker product by $\kron$, while the {Euclidean} dot product by $\dotprod{\cdot}{\cdot}$ both for vectors and tensors, where it is generalized through the tensor contraction. We denote by $||\cdot||$ the Euclidean norm for vectors and the Frobenious norm for matrix and tensors.
Let $\mathbfcal{A}: \R^{n_1\times\dots \times n_d} \rightarrow \R^{n_1\times\dots\times n_d}$ be a linear operator on tensor  product of  spaces and $\ten{A}\in\R^{(n_1\times n_1)\times\dots\times (n_d\times n_d)}$ its tensor representation with respect to the canonical basis, then $||\ten{A}||_2$ is the L$2$ norm of the linear operator $\mathbfcal{A}$. If $d=2$, then we have the L$2$ norm of the matrix associated with a simpler linear operator among two linear vector spaces. 
\subsection{Preconditioned GMRES}
For the solution of a linear system using an iterative solver, it is recommended to used stopping criterion based on a backward error~\cite{drsg:95,gree:97,paige2006}.
For iterative schemes, two normwise backward errors can be considered. The iterative scheme will be stopped when the backward error will become lower than  a user prescribed threshold; that is, when the current iterate can be considered as the exact solution of a perturbed problem where the relative norm of the perturbation is lower than the threshold. 
If we denote $Ax=b$ the linear system to be solved a first  backward error on $A \in \R^{n \times n}$ and $b \in \R^n$ can be considered.
We denote $\etaAb{A}{b}(x_k)$ this normwise backward error associated with the approximate solution
$x_k$ at iteration $k$, that is defined by~\cite{Rigal1967,higham2005}
\begin{align}
\etaAb{A}{b}(x_k) &= \min_{\DA, \Db} \left\{ \tau > 0 : \norm{\DA} \leq \tau \norm{A}, \ 
     \norm{\Db} \leq \tau \norm{b}  \right . \nonumber \nonumber \\
& \left .\qquad \qquad {\rm and}\ (A + \DA)x_k = b + \Db \right\} \nonumber \\
&= \frac{\norm{A x_k - b}}{\norm{A}_2\norm{x_k} + \norm{b}}  \label{eq:BE-Ab}.
\end{align}

In some circumstances, a simpler backward error criterion based on perturbations only in the right-hand side can also be considered, that leads to the second possible choice
\begin{align}
\etab(x_k) &= \min_{\Db} \left\{ \tau > 0 : \norm{\Db} \leq \tau \norm{b} \ {\rm and}\ A x_k = b + \Db \right\} \nonumber \\
&= \frac{\norm{A x_k - b}}{\norm{b}} . \label{eq:BE-b}
\end{align}

Starting from the zero  initial guess,  GMRES~\cite{saad1986} constructs a series of approximations $x_k$ in Krylov subspaces of increasing dimension $k$ so that the residual norm of the sequence of iterates is  decreasing over these
nested spaces. More specifically:
\[x_k = \argmin_{x\in  \mc{K}_k(A, b)}\left\|b - Ax\right\|,\]
with 
\[    \mc{K}_k(A, b) = \spn\{b, Ab, \ldots, A^{k - 1}b\} \]
the $k$-dimensional Krylov subspace spanned by $A$ and $b$. 
In practice,  a matrix $V_k~=~[v_1, \ldots, v_k]~\in~\mbbR^{n\times k}$ with orthonormal columns and an upper Hessenberg matrix $\bar{H}_k\in\mbbR^{(k + 1)\times k}$ are iteratively constructed using the Arnoldi procedure such that $\spn \{ V_k \} = \mc{K}_k(A, b)$ and
\begin{equation} \nonumber 
    AV_k = V_{k + 1}\bar{H}_k, \qquad\text{with}\qquad V_{k + 1}^T  V_{k + 1} = I_{k+1}.
\end{equation}
This is often referred to as the Arnoldi relation. Consequently, $x_k = V_ky_k$ with
\[ y_k = \argmin_{y\in\mbbR^k}\left\|\beta e_1 - \bar{H}_ky\right\|,\]
where $\beta = \left\|b\right\|$ and $e_1 = (1, 0, \ldots, 0)^T\in\mbbR^{k + 1}$ so that in
exact arithmetic the following equality holds between the least square residual and the true residual
\begin{equation}\label{eq:resLS_vs_trueRes}
\Vert \tilde{r}_k \Vert = \Vert \beta e_1 - \bar{H}_ky \Vert  = \Vert b  - A x_k \Vert.
\end{equation}
In finite precision calculation, this equality no longer holds but it has been shown that the GMRES method is backward stable with respect to $\etaAb{A}{b}$~\cite{paige2006} meaning that along the iterations
$\etaAb{A}{b}(x_k)$ might go down-to ${\mathcal O}(u)$ where $u$ is the unit round-off of the floating point arithmetic used to perform the calculations.
An overview of GMRES is given in~\hypref{Algorithm}{alg:gmres}; we refer  to~\cite{saad1986, saad2003} for a
more detailed presentation.
\begin{algorithm}
\caption{$x$, \texttt{hasConverged} $=$ GMRES($A$, $b$, \texttt{m}, $\varepsilon$)}\label{alg:gmres}
    \begin{algorithmic}[1]
        \State \textbf{input:} $A$, $b$, \texttt{m}, $\varepsilon$.
        \State $r_0 = b $, $\beta = \left\|r_0\right\|$ and $v_1 = r_0/\beta$
        \For{$k = 1, \ldots, \texttt{m}$}
            \State $w = A v_k$ \label{alg:gmres:1}
            \For{$i = 1, \ldots, k$} \Comment{MGS variant}
                \State $\bar{H}_{i,k} = \dotprod{v_i}{w}$ 
                \State $w = w -\bar{H}_{i, k}v_i$
            \EndFor
            \State $\bar{H}_{k + 1, k} = \left\|w\right\|$
            \State $v_{k + 1} = w/\bar{H}_{k + 1, k}$
            \State $\displaystyle y_k = \argmin_{y\in\mbbR^k}\left\|\beta e_1 - \bar{H}_ky\right\|$
            \State $x_k = V_ky_k$ \label{alg:gmres:2}
            \If{($\etaAb{A}{b}(x_k)  < \varepsilon$)} \State \texttt{hasConverged = True}
            \State \textbf{break}
            \EndIf
        \EndFor
        \State \textbf{return:} $x=x_k$, \texttt{hasConverged}
    \end{algorithmic}
\end{algorithm}

Because the orthonormal basis $V_k$ has to be stored, a restart parameter defining the maximal dimension of the search Krylov space is used to control the memory footprint of the solver.
If the maximum dimension of the search space is reached without converging, the algorithm is restarted using the final iterate as the initial guess for a new cycle of GMRES.
Furthermore, it is often needed to consider a preconditioned to speed-up the convergence. 
Using right-preconditioned GMRES consists in considering a non singular matrix $M$, the so-called preconditioner that approximates
the inverse of $A$ in some sense. In that case,  GMRES is applied to the preconditioned system $AM t = b$.
Once the solution $t$ has been computed the solution of the original system is recovered as $x = M t$.
The right-preconditioned GMRES is sketched in Algorithm~\ref{alg:rgmres} for a restart parameter \texttt{m} and a convergence threshold $\varepsilon$ .
\begin{algorithm}
    \caption{$x$, \texttt{hasConverged} = Right-GMRES($A$, $M$, $b$, $x_0$, \texttt{m}, $\varepsilon$)  }\label{alg:rgmres}
    \begin{algorithmic}[1]
        \State \textbf{input:} $A$, $M$, $b$, \texttt{m}, $\varepsilon$.
        \State \texttt{hasConverged = False}
        \State $x = x_0$
        \While{\texttt{not(hasConverged)}}
          \State $ r = b - A x$
          \hfill\Comment{\small{Iterative refinment step with at most $m$ GMRES iterations on $AM$}}\normalsize
          \State $t_k$, \texttt{hasConverged} =  GMRES($AM$, $r$, \texttt{m}, $\varepsilon$) 
          \State $ x = x + M t_k$ 
          \hfill\Comment{\small{Update the unpreconditionned with the computed correction}}\normalsize
        \EndWhile
        \State \textbf{return:} $x$, \texttt{hasConverged}
    \end{algorithmic}
\end{algorithm}

%

 \subsection{The Tensor Train format \label{ssec2:TT}}
 Firstly,  we describe the main key elements of the Tensor Train (TT) notation for tensors and linear operators between tensor  product of  spaces. Secondly, we present the advantages in using this formalism to solve linear systems that are naturally defined in high dimension spaces.

  Let $\ten{x}$ be a $d$-order tensor in  $\R^{n_1\times \dots \times n_d}$ and  $n_k$ the {dimension} of mode $k$ for every $k~\in~\{1, \dots, d\}$. Since storing the full tensor $\ten{x}\in\R^{n_1\times \dots \times n_d}$ has a memory cost of $\mathcal{O}(n^{d})$ with $n=~\max_{i\in\{1,\dots, d\}}\{n_i\}$, different compression techniques were proposed over the years to reduce the memory consumption~\cite{DeLathauwer2000, Grasedyck2009, Oseledets2011}.  For the purpose of this work the most suitable tensor representation  is the \emph{Tensor Train}  (TT) format~\cite{Oseledets2011}. The key idea of TT is expressing a tensor of order $d$ as the contraction of $d$ tensors of order $3$. The contraction is actually the generalization to tensors of the matrix-vector product. Given $\ten{a}\in\R^{n_1\times \dots\times n_{h}\times \dots\times n_{d_1} }$ and $\ten{b}\in\R^{m_1\times \dots\times n_{h}\times \dots\times m_{d_2}}$, their \emph{tensor contraction} with respect to mode $h$, denoted $\ten{a}\bullet_h\ten{b}$, provides a new tensor \[\ten{c}~\in~\R^{n_1\times\cdots\times n_{h-1}\times n_{h+1}\times\cdots\times n_{d_1}\times m_1\times\cdots\times m_{h-1}\times m_{h+1}\times\cdots\times m_{d-1}}\] such that its $(i_1,\dots,i_{h-1},i_{h+1},\dots,i_{{d_1}},j_1\dots,j_{h-1},j_{h+1},\dots,j_{{d_2}})$ element
  is
 \begin{equation*}
  \begin{split}
		 c&=(\ten{a}\bullet_h\ten{b})(i_1,\dots,i_{h-1},i_{h+1},\dots,i_{{d_1}},j_1\dots,j_{h-1},j_{h+1},\dots,j_{{d_2}}) \\
		&=\sum_{i_h = 1}^{n_h} \ten{a}(i_1,\dots, i_{h},\dots ,i_{d_1})\ten{b}(j_1,\dots, i_{h},\dots ,j_{d_2}).
  \end{split}
  \end{equation*}
  The contraction between tensors is linearly extended to more modes. To shorten the notation we omit the bullet symbol and the mode indices when the modes to contract will be clear from the context. 
  
  The contraction applies also to the tensor representation of tensor linear operators for computing the operator powers. Let $\mathbfcal{A}:\R^{n_1\times \dots \times n_d}\rightarrow \R^{n_1\times \dots \times n_d}$ be a linear operator and let the tensor $\ten{A}\in\R^{(n_1\times n_1)\times\dots\times (n_d\times n_d)}$ be its representation with respect to the canonical basis of $\R^{n_1\times \dots \times n_d}$. Then the tensor representation with respect to this canonical basis of $\mathbfcal{A}^2$ is $\ten{B}\in\R^{(n_1\times n_1)\times\dots\times (n_d\times n_d)}$, whose element $b = \ten{B}(i_1, j_1, \dots, i_d, j_d)$ is
  \begin{equation*}
	\begin{split}
  		b &= \bigl(\ten{A}\cont{h_\textsc{L}}{h_\textsc{R}} \ten{A}\bigr)(i_1, j_1, \dots, i_d, j_d)\\
  		  &= \sum_{k_1, \dots, k_d = 1}^{n_1, \dots , n_d} \ten{A}(i_1, k_1,\dots, i_d, k_d)\ten{A}(k_1, j_1,\dots, k_d, j_d)
  	\end{split}
  \end{equation*}
  with $h_\textsc{L} = \{2,4,\dots, 2d\}$ and $h_\textsc{R} = \{1,3,\dots, 2d-1\}$. From this, we recursively obtain the tensor associated with $\mathbfcal{A}^h$ for $h\in\N$.

  The Tensor Train expression of $\ten{x}\in\R^{n_1\times \dots \times n_d}$ is
  \[
  	\ten{x} = \coreTT{x}_1\coreTT{x}_2\cdots\coreTT{x}_d ,
  \]
  where $\coreTT{x}_k\in\R^{r_{k-1}\times n_k\times r_k}$ is called $k$-th \emph{TT-core} for $k\in\{1,\dots, d\}$, with $r_0 = r_d = 1$. Notice that $\coreTT{x}_1\in\R^{r_0\times n_1\times r_1}$ and $\coreTT{x}_d\in\R^{r_{d-1}\times n_d\times r_d}$ reduce essentially to matrices, but for the notation consistency we represent them as tensor. {The $k$-th TT-core of a tensor are denoted by the same bold letter underlined with a subscript $k$.} The value $r_k$ is called \emph{$k$-th TT-rank}. Thanks to the TT-formalism, the $(i_1,\dots, i_d)$-th element of $\ten{x}$ writes 
  \[
  \ten{x}(i_1,\dots, i_d) = \sum_{j_0,\dots,j_{d} = 1 }^{r_0,\dots, r_{d}}\coreTT{x}_1(j_0, i_1, j_1)\coreTT{x}_2(j_1, i_2, j_2)\dots\coreTT{x}_{d-1}(j_{d-2}, i_{d-1}, j_{d-1})\coreTT{x}_d(j_{d-1}, i_d, j_d).
	\]
  
  Given an index $i_k$, we denote the $i_k$-th matrix slice of $\coreTT{x}_k$ with respect to mode $2$ by $\matcoreTT{X}_k(i_k)$, i.e., $\matcoreTT{X}_k(i_k) = \coreTT{x}_k(:, i_k, :)$. Then each element of the TT-tensor $\ten{x}$ can be expressed as the product of $d$ matrices, i.e.,
  \[
  \ten{x}(i_1,\dots, i_d) = \matcoreTT{X}_1(i_1)\cdots\matcoreTT{X}_{d}(i_d)
  \]
  with $\matcoreTT{X}_k(i_k)\in\R^{r_{k-1}\times r_k}$ for every $i_k\in\{1,\dots, n_k\}$ and $k\in\{2, \dots, d-1\}$, while $\matcoreTT{X}_{1}(i_1)\in\R^{1\times r_1}$ and $\matcoreTT{X}_d(i_d)\in\R^{r_{d-1}\times 1}$. Remark that $\matcoreTT{X}_{1}(i_1)$ and $\matcoreTT{X}_d(i_d)$ are actually vectors, but as before to have an homogeneous notation they write as matrices with a single row or column.
 
  Storing a tensor in TT-format requires $\mathcal{O}(dnr^2)$ units of memory with $n = \max_{i\in\{1,\dots, d\}}\{n_i\}$ and $r = \max_{i\in\{1,\dots, d\}}\{r_i\}$. In this case the  memory footprint growths linearly with the tensor order and quadratically with the maximal TT-rank. Conseuqently, knowing the maximal TT-rank is usually sufficient to get an idea of
  the TT-compression benefit. However, to be more accurate, we introduce the compression ratio measure.
  Let $\ten{x}\in\R^{n_1\times\dots\times n_d}$ be a tensor in TT-format, then the compression ratio is the ratio between the storage cost of a in TT-format over the dense format storage cost, i.e.,
	\[	\dfrac{\sum_{i = 1}^{d}r_{i-1}n_ir_i}{\prod_{j=1}^{d}n_j}\]
  where $r_i$ is the $i$-th TT-rank of $\ten{x}$.
  As highlighted from the compression ratio, to have a significant benefit in the use of this formalism, the TT-ranks $r_i$ have to stay bounded and small. A first possible drawback of the TT-format appears with the addition of two TT-tensors. Indeed given two TT-tensors $\ten{x}$ and $\ten{y}$ with $k$-th TT-rank $r_k$ and $s_k$ respectively, then the $k$-th TT-rank of $\ten{x} + \ten{y}$ is equal to $r_k+s_k$, see~\cite{Gel2017TheTF}. So if $\ten{x}$ is a TT-tensor with $k$-th TT-rank $r_k$, then tensor $\ten{z} = 2\ten{x}$ has $k$-th TT-rank equal to $r_k$ if we simply multiply the first TT-core by $2$, but if it is computed as a sum of twice $\ten{x}$ then its TT-ranks double. In the following part, we discuss a studied solution to address this issue.
  \par The TT-formalism enables us to express in a compact way also linear operators between tensor product of spaces. Let $\mathbfcal{A}:\R^{n_1\times \dots\times n_d}\rightarrow \R^{n_1\times \dots n_d}$ be a linear operator between tensor  product of  spaces, fixed the canonical basis for $\R^{n_1\times \dots\times n_d}$, we associate with $\mathbfcal{A}$ the tensor $\ten{A}\in\R^{(n_1\times n_1)\times\dots \times (n_d\times n_d)}$ in the standard way. Henceforth a tensor associated with a linear operator over tensor  product of  spaces will be called a \emph{tensor operator}. The TT-representation of tensor operator $\ten{A}\in\R^{(n_1\times n_1)\times\dots \times (n_d\times n_d)}$, usually called~\emph{TT-matrix},
  is
  \[
  \ten{A} = \coreTT{a}_1\cdots\coreTT{a}_d,
  \]
where $\coreTT{a}_k\in\R^{r_{k-1}\times n_k\times n_k \times r_k}$, is its $k$-th \emph{TT-core}, with $r_0 = r_d = 1$. So its element $a= \ten{A}(i_1,j_1,\dots, i_d, j_d)$ is expressed in TT-format as
  \begin{equation} \nonumber
  \begin{split}
  a &= \sum_{h_0,\dots,h_{d} = 1 }^{r_0,\dots, r_{d}}\coreTT{a}_1(h_0, i_1, j_1, h_1)\coreTT{a}_2(h_1, i_2, j_2, h_2)\cdots\coreTT{a}_{d-1}(h_{d-2}, i_{d-1}, j_{d-1}, h_{d-1})\coreTT{a}_d(h_{d-1}, i_d, j_d, h_d).
  \end{split}
  \end{equation}
  
  Let $\matcoreTT{A}_k(i_k,j_k)\in\R^{r_{k-1}\times r_k}$ be the $(i_k, j_k)$-th slice with respect to mode $(2,3)$ of $\coreTT{a}_k$ for every $i_k, j_k\in\{1,\dots, n_k\}$ and $k\in\{1,\dots, d\}$. Then the last equation is equivalently expressed as 
  \[
  \ten{A}(i_1,j_1,\dots, i_d, j_d) = \matcoreTT{A}_1(i_1, j_1)\cdots\, \matcoreTT{A}_d(i_d, j_d).
  \]
  As before we estimate the storage cost as $\text{O}(dnmr^d)$ where $n = \max_{i\in\{1,\dots, d\}}\{n_i\}$, $m = \max_{i\in\{1,\dots, d\}}\{m_i\}$ and $r = \max_{i\in\{1,\dots, d\}}\{r_i\}$. However the $k$-th TT-rank of the contraction of a TT-operator and a TT-vector is equal to the product of the $k$-th TT-rank of the two contracted objects, see~\cite{Gel2017TheTF}. For example given the TT-operator $\ten{A} \in \R^{n_1\times m_1\times\dots \times n_d\times m_d}$ and TT-tensor $\ten{x}\in\R^{n_1\times \dots \times n_d}$ with $k$-th TT-rank $r_k$ and $s_k$ respectively, their contraction $\ten{b}= \ten{A}\ten{x}$ is a TT-tensor with $k$-th TT-rank equal to $r_ks_k$. 

The TT-rank growth is a crucial point in the implementation of algorithms using TT-tensors: it may lead to run out of memory and prevent the calculation to complete. To address this issue, a rounding algorithm to reduce the TT-rank was proposed in~\cite{Oseledets2011}. Given a TT-tensor $\ten{x}$ and a relative accuracy $\delta$, the TT-rounding algorithm provides a TT-tensor $\ten{\tilde{x}}$ that is at a relative distance $\delta$ from  $\ten{x}$,  i.e.,  $|| \ten{x}-\ten{\tilde{x}}|| \le \delta ||\ten{x}||$. Given a TT-tensor $\ten{x}\in\R^{n_1\times\dots \times n_d}$ and setting $r = \max_{i\in\{1,\dots, d\}}\{r_i\}$ and $n = \max_{i\in\{1,\dots, d\}}\{n_i\}$, the computational cost, in terms of floating point operations, of a TT-rounding over $\ten{x}$ is $\mathcal{O}(dnr^3)$, as stated in~\cite{Oseledets2011}.


\subsection{Preconditioned GMRES in Tensor Train  format}
\normalcolor
	Assume $\ten{A}\in\R^{(n_1\times n_1) \times \dots\times (n_d\times n_d)}$ to be a tensor operator and $\ten{b}\in\R^{n_1\times \dots \times n_d}$ a tensor, then the general tensor linear system is
	\begin{equation}
	\label{eqTT:3}
		\ten{A}\ten{x} = \ten{b}
	\end{equation}
	with 
	$\ten{x}\in\R^{n_1\times\dots\times n_d}$. Notice that setting $d=2$ we have the standard  linear system from classical matrix computation. A possible way for solving~\eqref{eqTT:3} is using a tensor-extended version of GMRES. Since all the operations appearing in this iterative solver are feasible with the TT-formalism, we assume that all the objects are expressed in TT-format. A main drawback in this approach is due to the repetition of sums and contractions in the different loops, which leads to the TT-rank growth and a possible memory over-consumption. Therefore introducing compression steps in TT-GMRES is essential but a particular attention should be paid to the choice of the rounding parameter to ensure that the prescribed GMRES tolerance $\varepsilon$ can be reached. Our TT-GMRES algorithm is fully presented in Algorithm \ref{alg:TT-gmres}.

	\begin{algorithm}
	\caption{$\ten{x}$, \texttt{hasConverged} = TT-GMRES($\ten{A}$, $\ten{b}$, \texttt{m}, $\varepsilon$, $\delta$)}\label{alg:TT-gmres}
	\begin{algorithmic}[1]
		\State \textbf{input:} $\ten{A}$, $\ten{b}$, \texttt{m}, $\varepsilon$, $\delta$.
		\State $\ten{r_0} = \ten{b}$, $\beta = \left\|\ten{r_0}\right\|$ and $\ten{v}_1 = (1/\beta)\ten{r}_0$\label{alg3:rd0}
		\For{$k = 1, \ldots, \texttt{maxit}$}
		\State $\ten{w} = \texttt{TT-round}(\ten{A}\ten{v}_k, \delta)$ \label{alg3:rd1}
                \Comment{MGS variant}
		\For{$i = 1, \ldots, k$}
		\State $\bar{H}_{i,k} = \langle\ten{v}_i\,, \ten{w}\rangle$  
		\State $\ten{w} = \ten{w} -\bar{H}_{i, k}\ten{v}_i$
		\EndFor
		\State $\ten{w} = \texttt{TT-round}(\ten{w}, \delta)$\label{alg3:rd2}
		\State $\bar{H}_{k + 1, k} = \left\|\ten{w}\right\|$
		\State $\ten{v}_{k + 1} = (1/\bar{H}_{k + 1, k})\ten{w}$
		\State $y_k = \argmin_{y\in\mbbR^k}\left\|\beta e_1 - \bar{H}_ky\right\|$
		\State $\ten{x}_k = \texttt{TT-round}\bigl(\sum_{j = 1}^{k+1} y_k(j)\ten{v}_j, \delta\bigr)$\label{alg3:rd3} 
		\If{($\etaAb{\ten{A}}{\ten{b}}(\ten{x}_k)  < \varepsilon$)}
		\State \texttt{hasConverged = True}
		\State \textbf{break}
		\EndIf
		\EndFor
		\State \textbf{return:} $\ten{x}=\ten{x}_k$, \texttt{hasConverged}
	\end{algorithmic}
	\end{algorithm}

\begin{algorithm}
    \caption{$\ten{x}$, \texttt{hasConverged} = Right-GMRES($\ten{A}$, $\ten{M}$, $\ten{b}$, $\ten{x_0}$, \texttt{m}, $\varepsilon$, $\delta$)  }\label{alg:TT-rgmres}
    \begin{algorithmic}[1]
        \State \textbf{input:} $\ten{A}$, $\ten{M}$, $\ten{b}$, \texttt{m}, $\varepsilon$, $\delta$.
        \State \texttt{hasConverged = False}
        \State $\ten{x} = \ten{x_0}$
        \While{\texttt{not(hasConverged)}}
          \State $ \ten{r} = \texttt{TT-round}( \ten{b} - \ten{A} \ten{x}, \delta)$
          \hfill\Comment{\small{Iterative refinment step with at most $m$ GMRES iterations on $\ten{AM}$}}\normalsize
          \State $\ten{t}_k$, \texttt{hasConverged} =  GMRES($\ten{AM}$, $\ten{r}$, \texttt{m}, $\varepsilon$, $\delta$) 
          \State $\ten{x} = \texttt{TT-round}(\ten{x} + \ten{M t}_k, \delta)$ 
          \hfill\Comment{\small{Update the unpreconditionned with the computed correction}}\normalsize
        \EndWhile
        \State \textbf{return:} $\ten{x}$, \texttt{hasConverged}
    \end{algorithmic}
\end{algorithm}

	In Algorithm~\ref{alg:TT-gmres} and~\ref{alg:TT-rgmres} there is an additional input parameter $\delta$, i.e., the rounding accuracy. The TT-rounding algorithm at accuracy $\delta$ is applied to the result of the contraction between $\ten{A}$ and the last Krylov basis vector computed in Line~\ref{alg3:rd1}, to the new Krylov basis vector after orthogonalization in Line~\ref{alg3:rd2} and to the updated iterative solution, Line~\ref{alg3:rd3}.  The purpose is to balance with the rounding the rank growth due to the tensor contraction or sum that occurred in the immediate previous step. {As it will be observed in the numerical experiments of Section~\ref{sec3}, the rounding accuracy $\delta$ has to be chosen lower or equal than the GMRES target accuracy $\varepsilon$.

	
     
\section{Solution of parametric problems in Tensor Train format} \label{ssec:Ai1}
        In this section, we investigate the situation where either the tensor representation of the linear operator or the right-hand side has a mode related to a parameter that is discretized.
        In the case of the parametric linear operator,  we are interested into the numerical quality of the computed solutions when we solve for all the parameters at once compared to the
        solution computed when the parametric systems 
         are treated independently.  
        In the case of the right-hand sides depending on a parameter, we investigate the links between the search space of TT-GMRES enabling the solution of  all the right-hand sides at once
        and the spaces built by the GMRES solver on each right-hand side considered independently. In this subsection, tensor slices play a key role, as consequence we introduce a specific notation. Given a tensor $\ten{a}\in\R^{n_1\times \cdots \times n_d}$ in TT-format with TT-cores
		$\coreTT{a}_k\in\R^{r_{k-1}\times n_k\times r_k}$, $\ten{a}^{[k, i_k]}$ denotes the $i_k$-th slice with respect to mode $k$. 
        Since henceforth we will take slice only with respect to the first mode, instead of writing $\ten{a}^{[1, i_1]}$ for the $i_1$-th slice on the first mode we will simply write $\ten{a}^{[i_1]}$. Similarly $\ten{A}^{[i_1]}$ denotes the $(i_1, i_1)$-th slice with respect to mode $(1,2)$ of a tensor operator $\ten{A}\in\R^{(n_1\times n_1)\times \cdots \times (n_d\times n_d)}$.

        \subsection{Parameter dependent linear operators}
	This subsection focuses on a specific type of parametric tensor operators expressed as $\ten{A}_\alpha~=~\ten{B}_0+\alpha\ten{B}_1$ with $\alpha\in\R$ and $\ten{B}_0, \ten{B}_1$ two tensor operators of $\R^{(n_1\times n_1)\times\dots\times (n_d\times n_d)}$. Assuming that $\alpha$ takes $p$ different real values in the interval $[a,b]$, we define $p$ linear systems of the form 
	\begin{equation}\label{eq:Param Ax=b}
	\ten{A}_{\ell}\ten{y}_\ell = \ten{b}_\ell
	\end{equation}
	where $\ten{A}_{\ell} = \ten{B}_0 + \alpha_\ell\ten{B}_1$, $\ten{b}_\ell\in\R^{n_1\times\dots\times n_d}$, $\alpha_\ell\in[a,b]$ for every $\ell\in\{1,\dots, p\}$.  
	At this level, it is possible to choose between classically solving each system independently or solving them simultaneously in a higher dimensional space defining the so-called  ``all-in-one'' system. This latter system writes
	\begin{equation}\label{eqn:allInOne}
	\ten{A}\ten{x}=\ten{b}
	\end{equation}
	where  $\ten{A}\in\R^{(p\times p)\times (n_1\times n_1)\times \dots \times(n_d\times n_d)}$ such that
	\begin{equation}
	\label{eqAi1:op}
		\ten{A}( h, \ell, i_1, j_1, \dots, i_d, j_d) =
		\begin{cases}
		\ten{A}_{\ell}(i_1, j_1, \dots, i_d, j_d)\quad&\text{if}\quad h=\ell , \\
		\qquad 0\quad&\text{if}\quad h\ne\ell ,
		\end{cases}
	\end{equation}	
and the right-hand side is $\ten{b}\in\R^{p\times n_1\times\dots  \times n_d}$ defined as 
	\begin{equation}
	\label{eqAi1:rhs}
		\ten{b}(\ell,i_1, \dots, i_d) = \ten{b}_\ell(i_1, \dots, i_d)
	\end{equation}
	 for $i_k, j_k\in\{1,\dots, n_k\}$, $k\in\{1,\dots, d\}$ and $\ell, h\in\{1, \dots, p\}$.  The tensor operator $\ten{A}$ writes in a compact format as 
	 \[
	 \ten{A} = \I_p\otimes\ten{B}_0 + \text{diag}(\alpha_1,\dots, \alpha_p)\otimes\ten{B}_1.
	 \]
	 
  The $(\ell, \ell)$-th slice of $\ten{A}$ with respect to modes $(1,2)$ is denoted 
	 	\begin{equation}\label{eq:individual matrix}
			\ten{A}^{[\ell]} = \ten{B}_0 + \alpha_\ell \ten{B}_1 = \ten{A}_{\ell}	 	
	 	\end{equation}
	 	and similarly the $\ell$-th slice of $\ten{b}$ with respect to the first mode is $\ten{b}^{[\ell]} = \ten{b}_\ell$ by construction. 
	 	So that Equation~\eqref{eq:Param Ax=b} also writes
	 	$$
	 	  \ten{A}^{[\ell]} \ten{x}^{[\ell]} = \ten{b}^{[\ell]}
	 	$$
	 	with $\ten{x}^{[\ell]} = \ten{y}_\ell$.  It shows that, once the ``all-in-one'' system,  Equation~(\ref{eqn:allInOne}),  has been solved, the solution related to a specific parameter can be extracted as a slice of the ``all-in-one'' solution, obtaining an \emph{extracted individual solution}. In other words, given the $k$-th iterate $\ten{x}_{k}$ of the ``all-in-one'' system, the extracted individual solution for the $\ell$-th problem is $\ten{x}^{[\ell]}_k$, i.e., the $\ell$-th slice with respect to the first mode defined as
	 \[
	 	\ten{x}^{[\ell]}_k = \ten{x}_k(\ell,i_1, \dots, i_d ).
	 \]
	 In the following propositions, we investigate the relation between the backward error of the ``all-in-one'' system solution and the extracted individual one. The equalities given for the ``all-in-one'' system are clearly true if the tensor and the tensor operators are given in full format, but they hold also in TT-format. All the details related to the ``all-in-one'' construction in TT-format are given in  \reportPaper{Appendix~\ref{app:tt}}{\cite{coulaud2022}}. 
	 
	The proven bounds enable us to tune the convergence threshold when solving for multiple parameters while guaranteeing a prescribed quality for the individual extracted solutions.
	In particular, the bound given by Equation~\eqref{eq:beb_multiple operators} in Proposition~\ref{prop:eta_b} shows that if a  certain accuracy $\varepsilon$ is expected for the extracted individual solution in terms of the backward error in~\eqref{eq:BE-b}, a more stringent convergence threshold should be used for the ``all-in-one'' system solution that should be set to $\varepsilon / \sqrt{p}$.
	\begin{proposition}\label{prop:eta_b}
	Let the ``all-in-one'' operator $\ten{A}\in\R^{(p\times p)\times(n_1\times n_1)\times \dots (n_d\times n_d)}$ 
		and right-hand side $\ten{b}\in\R^{p\times n_1\times\dots  \times n_d}$ be as in Equations~\eqref{eqAi1:op} and~\eqref{eqAi1:rhs} respectively, we consider the ``all-in-one'' system 
		$$\ten{A}\ten{x} = \ten{b}.$$
			Let $\ten{A}_{\ell}\in\R^{(n_1\times n_1)\times \dots \times (n_d\times n_d)}$ be the tensor operator as in Equation~\eqref{eq:individual matrix} and let $\ten{b}_{\ell}\in\R^{n_1\times\dots\times n_d}$ be a tensor such that $||\ten{b}_{\ell}|| = 1$, that defines the individual linear systems 
		$$\ten{A}_{\ell}\ten{y}_\ell = \ten{b}_{\ell}$$
	 with $\ten{A}_{\ell} = \ten{A}^{[\ell]} $ and $ \ten{b}_\ell = \ten{b}^{[\ell]}$ for every $\ell\in\{1,\dots, p\}$.

Let $\ten{x}_k$ denote the ``all-in-one'' iterate, we have
		\begin{equation} \label{eq:beb_multiple operators}
		  \eta_{\ten{b}}(\ten{x}_k)\sqrt{p} \ge \eta_{\ten{b_{\ell}}}(\ten{x}_k^{[\ell]})
	        \end{equation}
		for $\ell\in\{1, \dots, p\}$.
	\end{proposition}
	\begin{proof}
		For the sake of simplicity we use $\eta_{\ten{b}}$ and $\eta_{\ten{b}_{\ell}}$ squared throughout the proof and discard the subscript of the $k$-th ``all-in-one'' iterate.
		The quantity $\eta^2_{\ten{b}_{\ell}}(\ten{x}^{[\ell]})$ explicitly gets
		\[
		\eta^2_{\ten{b}_{\ell}}(\ten{x}^{[\ell]}) = \frac{\norm{\ten{A}_{\ell}\ten{x}^{[\ell]} - \ten{b}_{\ell}}^2}{\norm{\ten{b}_{\ell}}^2}
		\]
		while $\eta^2_{\ten{b}}(\ten{x})$ is
		\begin{equation}
		\label{eqP:2}
		\eta^2_{\ten{b}}(\ten{x}) = \frac{\norm{\ten{A}\ten{x} - \ten{b}}^2}{\norm{\ten{b}}^2}.
		\end{equation}
		Thanks to the diagonal structure of $\ten{A}$ and the Frobenius norm definition, Equation~\eqref{eqP:2} writes
		\begin{equation}
		\label{eqP:3}
		\eta^2_{\ten{b}}(\ten{x}) = 
		\frac{\sum_{\ell = 1}^n \norm{\bigl(\ten{A}\ten{x} - \ten{b}\bigr)^{[\ell]}}^2}{\sum_{k = 1}^p \norm{\ten{b}^{[k]}}^2} =
        \frac{\sum_{\ell = 1}^n \norm{\ten{A}_{\ell} \ten{x}^{[\ell]} - \ten{b}_{\ell}^2}}{\sum_{k = 1}^{p}\norm{\ten{b}_{k}}^2} =
        \frac{\sum_{\ell = 1}^p \eta_{\ten{b}_{\ell}}^2(\ten{x}^{[\ell]})}{p} 
		\end{equation}
		since $||\ten{b}||^2 = \sum_{k = 1}^{n}||\ten{b}_{k}||^2 = p$. 
     From the square root of both sides of this last equation, the result follows.
	\end{proof}

        For the backward error based on perturbation of both the linear operator and the right-hand side defined by~\eqref{eq:BE-Ab}, a similar result can be derived.
 While informative this result has a lower practical interest as the term $\rho_\ell(x)$ in~\eqref{eqP2:T} depends on the solution; so defining the convergence threshold for
the 'all-in-one' solution to guarantee the individual backward error requires some a priori information on the solution norms.
	\begin{proposition}
		\label{prop:eta_Ab}
		With the same hypothesis and notation as for Proposition~\ref{prop:eta_b} for $\eta_{\ten{A}, \ten{b}}(\ten{x})$ and $\eta_{\ten{A}_{\ell}, \ten{b}_{\ell}}(\ten{x}^{[\ell]})$ associated with the linear systems $\ten{A}\ten{x} = \ten{b}$ and $\ten{A}_{\ell}\ten{y}_\ell = \ten{b}_{\ell}$ respectively, for every $\ell\in\{1,\dots, p\}$, we have
		\begin{equation}
		\label{eqP2:T}
		\eta_{\ten{A},\ten{b}}(\ten{x}_k)\,\rho_\ell(\ten{x}_k) \ge \eta_{\ten{A}_{\ell}, \ten{b}_{\ell}}(\ten{x}_k^{[\ell]})\qquad\text{where}\qquad \rho_\ell(\ten{x}_k) = \frac{\norm{\ten{A}}_2\norm{\ten{x}_k} + \sqrt{p}}{\norm{\ten{A}_{\ell}\ten{x}^{[\ell]}_k} + 1 }
		\end{equation}
	with $\ten{x}_k$ the $k$-th ``all-in-one'' iterate and $\ten{x}_k^{[\ell]}$ its $\ell$-th slice with respect to mode $1$.
	\end{proposition}
	\begin{proof}
		For the sake of simplicity, as previously, the subscript of the $k$-th ``all-in-one'' iterate is dropped. The quantity $\eta_{\ten{A},\ten{b}}(\ten{x})$ explicitly writes
		\[
		\eta_{\ten{A},\ten{b}}(\ten{x}) = \frac{\norm{\ten{A}\ten{x} - \ten{b}}}{\norm{\ten{A}}_2\norm{\ten{x}} + \norm{\ten{b}}}.
		\]
		If the previous equation is multiplied equivalently by $\eta_\ten{b}(\ten{x})$, it gets
		\begin{equation}
		\label{eqP2:1}
		\eta_{\ten{A},\ten{b}}(\ten{x}) = \frac{\norm{\ten{A}\ten{x} - \ten{b}}}{\norm{\ten{A}}_2\norm{\ten{x}} + \norm{\ten{b}}}\frac{\eta_{\ten{b}}(\ten{x})}{\eta_{\ten{b}}(\ten{x})} = \frac{\norm{\ten{b}}}{\norm{\ten{A}}_2\norm{\ten{x}} + \norm{\ten{b}}}\eta_{\ten{b}}(\ten{x})
		= \frac{\sqrt{p}}{\norm{\ten{A}}_2\norm{\ten{x}} + \sqrt{p}}\eta_{\ten{b}}(\ten{x})
		\end{equation}
		by the definition of $\eta_{\ten{b}}(\ten{x})$ and $\norm{\ten{b}} = \sqrt{p}$. Similarly $\eta_{\ten{A}_\ell, \ten{b}_\ell}(\ten{x}^{[\ell]})$ is expressed in function of $\eta_{\ten{b}_{\ell}}(\ten{x}^{[\ell]})$ as 
		\begin{equation}
		\label{eqP2:2}
		\eta_{\ten{A}_{\ell},\ten{b}_{\ell}}(\ten{x}^{[\ell]}) = \frac{\norm{\ten{b}_{\ell}}}{\norm{\ten{A}_{\ell}}_2\norm{\ten{x}^{[\ell]}} + \norm{\ten{b}_{\ell}}}\eta_{\ten{b}_{\ell}}(\ten{x}^{[\ell]}) = \frac{1}{\norm{\ten{A}_{\ell}}_2\norm{\ten{x}^{[\ell]}} + 1}\eta_{\ten{b}_{\ell}}(\ten{x}^{[\ell]})
		\end{equation}
		since $\norm{\ten{b}_{\ell}} = 1$. Multiplying each side of Equation~\eqref{eqP2:1} by $(\norm{\ten{A}}_2\norm{\ten{x}} + \sqrt{p})$, it follows
		\[
		\bigl(\norm{\ten{A}}_2\norm{\ten{x}} + \sqrt{p}\bigr)\eta_{\ten{A}, \ten{b}} = \eta_{\ten{b}}\sqrt{p}. 
		\]
		Thanks to the result of Proposition \ref{prop:eta_b}, we have
		\begin{equation}
		\label{eqP2:3}
		\bigl(\norm{\ten{A}}_2\norm{\ten{x}} + \sqrt{p}\bigr)\eta_{\ten{A}, \ten{b}}(\ten{x}) = \eta_{\ten{b}}(\ten{x})\sqrt{p} \ge \eta_{\ten{b}_{\ell}}(\ten{x}^{[\ell]}) = \bigl(\norm{\ten{A}_{\ell}}_2{\norm{\ten{x}^{[\ell]}}} + 1\bigr)\eta_{\ten{A}_{\ell}, \ten{b}_{\ell}}(\ten{x}^{[\ell]})
		\end{equation}
		from Equation~\eqref{eqP2:2}. Dividing both sides of Equation~\eqref{eqP2:3} by $\norm{\ten{A}_{\ell}}_2\norm{\ten{x}^{[\ell]}} + 1$, it becomes
		\begin{equation}
		\label{eqP2:4}
		\frac{\norm{\ten{A}}_2\norm{\ten{x}} + \sqrt{p}}{\norm{\ten{A}_{\ell}\ten{x}^{[\ell]}} + 1}\eta_{\ten{A}, \ten{b}}(\ten{x}) \ge \eta_{\ten{A}_{\ell}, \ten{b}_{\ell}}(\ten{x}^{[\ell]})
		\end{equation}
		since $\norm{\ten{A}_{\ell}}_2\norm{\ten{x}^{[\ell]}} \ge \norm{\ten{A}_{\ell}\ten{x}^{[\ell]}}$ by the definition of the L$2$ norm. 
	\end{proof}

 \subsection{Parameter dependent right-hand sides}
We consider a particular case of this ``all-in-one'' approach. We intend to solve $p$ linear systems with the same linear operator and different right-hand sides. 
Given a linear tensor operator $\ten{A}_0\in\R^{(n_1\times n_1)\times \dots\times(n_d\times n_d)}$, we define the $\ell$-th linear system as
\begin{equation}
\label{eqMrhs:1}
\ten{A}_0\ten{y}_\ell = \ten{b}_\ell
\end{equation}
with $\ten{b}_\ell\in\R^{n_1\times\dots\times n_d}$ for every $\ell\in\{1,\dots, p\}$.
To solve simultaneously all the right-hand sides expressed in Equation~\eqref{eqMrhs:1}, we repeat the construction introduced in Subsection~\ref{ssec:Ai1}, except that $\ten{A}_0$ is repeated on the `diagonal' of tensor linear operator $\ten{A}$ defined in Equation~\eqref{eqAi1:op}. 
Thanks to the tensor properties, the tensor operator $\ten{A}\in\R^{(p\times p)\times(n_1\times n_1)\times\dots\times (n_d\times n_d)}$ writes
\[
\ten{A} = \I_p \otimes \ten{A_0}
\]
so that  $\ten{A}^{[\ell]} = \ten{A}_0$ for every $\ell\in\{1,\dots, p\}$. 
The right-hand side $\ten{b}$ is defined similarly to the previous section, that is $\ten{b}^{[\ell]} = \ten{b}_{\ell}$.
If the initial guess is $\ten{x}_{0}\in\R^{p\times n_1\times \dots\times n_d}$ equal to the null tensor, then at the $k$-th iteration TT-GMRES minimizes with respect to $\ten{x}_k$ the norm of the residual $\ten{r}_{k} = \ten{A}\ten{x}_k - \ten{b}$ on the space
\[
\mathcal{K}_k(\ten{A},\ten{b}) = \text{span} \bigl\{\ten{b}, \ten{A}\ten{b},\ten{A}^2\ten{b},\dots, \ten{A}^{k-1}\ten{b}\bigr\} ,
\]
i.e., we seek a tensor $\ten{x}_{k}\in\mathcal{K}_k(\ten{A}, \ten{b})$ such that
\[
\ten{x}_{k} = 
 \argmin_{\ten{x}\in\mathcal{K}_k(\ten{A}, \ten{b})}\norm{\ten{A}\ten{x} - \ten{b}}.
\]
Due to the diagonal structure of $\ten{A}$, the Frobenius norm of $\ten{r}_{k} = \ten{A}\ten{x}_k - \ten{b}$ is naturally written as follows
\[
||\ten{r}_{k}||^2 = \sum_{\ell = 1}^p||\ten{b}_{\ell} - \ten{A}_0 \ten{x}_k^{[\ell]}||^2
\]
with, similarly to the previous section,  $\ten{x}^{[\ell]}_k$ is the $\ell$-th slice with respect to the first mode of $\ten{x}_{k}$. 
Thanks to the diagonal structure of $\ten{A}$, we have  that the $\ell$-th slice of the Krylov basis vector $\ten{A}^h \ten{b}$ with respect to the first mode is $\ten{A}_0^h \ten{b}_\ell$.
Consequently  the $\ell$-th slices of the basis vectors of $\mathcal{K}_k(\ten{A},\ten{b})$ span the Krylov space $\mathcal{K}_k(\ten{A}_0,\ten{b}_\ell)$. 
It means that the individual solutions defined by the slices $\ten{x}_k^{[\ell]}$  of the iterate from the ``all-in-one''  TT-GMRES scheme lie in the same space as the $\ten{{y}}_{\ell,k}$ generated by TT-GMRES applied to the individual systems $\ten{A}_0\ten{y}_\ell = \ten{b}_\ell$ with $\ten{y}_{\ell,{0}}= 0$. 
While the two iterates belong to the same space, they are different since the former, $\ten{x}^{[\ell]}_k$, is build by minimizing the residual norm of $\ten{A} \ten{x} - \ten{b}$
over $\mathcal{K}_k(\ten{A},\ten{b})$ and the latter, $\ten{{y}}_{\ell,k}$, by minimizing the residual norm of $\ten{A}_0 \ten{y}_\ell = \ten{b}_\ell$ over $\mathcal{K}_k(\ten{A}_0,\ten{b}_\ell)$.
If we neglect the effect of the rounding, one can expect that
\[
  \norm{\ten{b}_\ell - \ten{A}_0 \ten{x}^{[\ell]}_k}\ge \norm{\ten{b}_\ell - \ten{A}_0 \ten{{y}}_{\ell,k}}
\]

\begin{remark}
We notice that a block TT-GMRES method could also be defined for the solution of such multiple right-hand side problems.
In that situation each individual residual norm would be minimized over the same space spanned by the sum of the individual Krylov space.
This would be somehow the dual approach to the one described above, where we minimize the sum of the residual norms on each individual Krylov space.
\end{remark}


Regarding the numerical quality of the extracted solution compared to the individually computed solution,
the bound stated in Proposition~\ref{prop:eta_b} is still true. 
As in the previous section an informative, but with lower practical interest, bound similar Proposition~\ref{prop:eta_Ab} can be derived.
\begin{proposition}
	\label{prop:eta_Ab_m}
	Under the hypothesis of Proposition \ref{prop:eta_Ab}, if $\ten{A} = \I_p\otimes \ten{A}_0$, then for $\eta_{\ten{A}, \ten{v}}(\ten{x})$ and $\eta_{\ten{A}_\ell, \ten{b}_\ell}(\ten{x}^{[\ell]})$ associated with the linear systems $\ten{A}\ten{x} = \ten{b}$ and $\ten{A}_0\ten{y}_\ell = \ten{b}_\ell$ respectively, for every $\ell\in\{1,\dots, p\}$ the following inequality holds  
	\begin{equation}
	\label{eqP3:T}
	\eta_{\ten{A},\ten{b}}(\ten{x}_k)\,\psi_\ell(\ten{x}_k) \ge \eta_{\ten{A}_\ell, \ten{b}_\ell}(\ten{x}_k^{[\ell]})\qquad\text{where}\qquad \psi_\ell(\ten{x}_k) = \frac{\norm{\ten{x}_k} + \sqrt{p}/\norm{\ten{A}_0}_2}{\norm{\ten{x}_k^{[\ell]}} + 1/\norm{\ten{A}_0}_2}.
	\end{equation}
\end{proposition}
\begin{proof}
	The result follows from the thesis of Proposition~\ref{prop:eta_Ab}, since $\norm{\ten{A}}_2 = \norm{\ten{A}_0}_2$
\end{proof}

\begin{corollary}
	\label{cor:eta_Ab}
	Given a sequence of iterative solutions $\{\ten{x}_{k}\}_{k\in\N}$ and a value $\nu$, if there exists a $k_{\ell}^*\in\N$ such that $|\,||\ten{A}_\ell \ten{x}_{k}^{[\ell]}|| - 1|\le \nu$ for every $k\ge k_{\ell}^*$, then 
	\begin{equation}
	\label{eqCP2:T1}
		\eta_{\ten{A},\ten{b}}(\ten{x}_{k})\,\rho^*(\ten{x}_{k}) \ge \eta_{\ten{A}_\ell, \ten{b}_\ell}(\ten{x}_{k}^{[\ell]})\qquad\text{where}\qquad \rho^*(\ten{x}_{k}) = \frac{\norm{\ten{A}}_2\norm{\ten{x}_{k}}+ \sqrt{p}}{2-\nu}
	\end{equation}
	for every $\ell\in\{1,\dots, p\}$ and for every $k\in\N$ such that $k \ge k^{**}$ where $k^{**} = \max k_\ell^{*}$. 
\end{corollary}

The thesis of Corollary~\ref{cor:eta_Ab} is independent of the structure of the operator and consequently remains valid in this multiple right-hand side structure described above.
\section{Numerical experiments\label{sec3}}
In this section we investigate the numerical behaviour of the TT-GMRES solver for linear problems with increasing {dimension} as it naturally arises in some partial differential equation (PDE)  studies. {We start by illustrating how the TT-operators of our numerical examples are directly constructed in TT-format, thanks to their peculiarity.}
For all the examples, we illustrate numerical concerns related to the algorithm convergence and computational costs, with a focus on memory growth and memory saving.


	   
	   The linear operators of the main problems, we will address, are \emph{Laplace-like} operators. The Laplace-like tensor operator $\ten{A}\in\R^{n_1\times m_1 \dots\times n_d\times m_d}$ is the sum of operators written as
	\begin{equation}
	\label{eqTT:2}
	\begin{split}
	\ten{A} = &M_1\otimes R_2\otimes R_3\otimes\cdots \otimes R_{d-2}\otimes R_{d-1}\otimes R_d \\
	&+ L_1\otimes M_2\otimes R_3\otimes \cdots\otimes R_{d-2}\otimes R_{d-1}\otimes R_d\\
	&+\cdots + L_1\otimes L_2\otimes L_3\otimes \cdots\otimes L_{d-2}\otimes M_{d-1}\otimes R_d\\
	&+ L_1\otimes L_2\otimes L_3\otimes \cdots\otimes L_{d-2}\otimes L_{d-1}\otimes M_d
	\end{split}
	\end{equation} 
	with $L_k, M_k, R_k \in\R^{n_k\times m_k}$ for every $k\in\{1,\dots, d\}$. As relevant property, these linear operators are expressed in TT-format with TT-rank $2$, i.e.,
	\begin{equation}
	\ten{A} = \begin{bmatrix}
	L_1 & M_1
	\end{bmatrix}\otimes 
	\begin{bmatrix}
	L_2 & M_2 \\
	0 & R_2 
	\end{bmatrix}
	\otimes\cdots\otimes 
	\begin{bmatrix}
	L_{d-1} & M_{d-1} \\
	0 & R_{d-1} 
	\end{bmatrix}
	\otimes
	\begin{bmatrix}
	M_d \\ R_d
	\end{bmatrix}
	\label{eqTT:2a}
	\end{equation}
	as proved in~\cite[Lemma 5.1]{Kazeev2012}. Remarking that the general expression of the discrete $d$-dimensional Laplacian on a uniform grid of $n$ points in each direction is
	\[
	\ten{\Delta}_d = \Delta_1\otimes \I_n \otimes\cdots\otimes \I_n +\cdots + \I_n\otimes \I_n\otimes\dots\otimes \Delta_1
	\]
	where $\I_n$ is the identity matrix of size $n$ and $\Delta_1\in\R^{n\times n}$ is the discrete $1$-dimensional Laplacian using the central-point finite difference scheme with discretization step $h=\frac{1}{n+1}$, i.e.,
	\[
	\Delta_1 = \frac{1}{h^2}\begin{bmatrix}
	-2 & 1  & 0&\dots&0\\
	1  & -2 & 1 &\dots&0\\
	\vdots&\ddots&\ddots&\ddots&\vdots\\
	0 & \dots& 1 & -2&1\\
	0 &0    & \dots& 1 & -2
	\end{bmatrix}.
	\]
	Then the TT-expression of $\ten{\Delta}_d$ is
	\begin{equation}
	\label{eqTT:4}
	\ten{\Delta}_d = \begin{bmatrix}
	\I_n& \Delta_1 
	\end{bmatrix}\otimes 
	\begin{bmatrix}
	\I_n & \Delta_1 \\
	\0 & \I_n 
	\end{bmatrix}
	\otimes\cdots\otimes 
	\begin{bmatrix}
	\I_n & \Delta_1 \\
	\0 & \I_n
	\end{bmatrix}
	\otimes
	\begin{bmatrix}
	\Delta_1 \\ \I_n
	\end{bmatrix}.
	\end{equation}
	
	To solve linear systems efficiently,  we consider an approximation of the inverse of the discrete Laplacian operator, $\ten{M}$,  as a preconditioner~\cite{Hackbusch2006I, Hackbusch2006II}. This operator writes
	\begin{equation}
	\label{eqTT:5}
	\ten{M} = \sum_{k=-q}^{q}c_k\exp(-t_k\Delta_1)\otimes\dots\otimes\exp(-t_k\Delta_1)
	\end{equation}
	where $c_k = \xi t_k$, $t_k = \exp(k\xi)$ and $\xi = \frac{\pi}{q}$. Thanks to the previously stated property of sum of TT-tensors, we conclude that the TT-ranks of $\ten{M}$ will be at least $2q+1$. 
	In Section~\ref{sec3} we consider the linear system $\ten{A}\ten{x} = \ten{b}$ and to speed up its convergence we apply the preconditioner TT-matrix $\ten{M}$, effectively solving $\ten{A}\ten{M}\ten{t} = \ten{b}$. The preconditioner TT-matrix $\ten{M}$ is always computed by a number of addends $q$ equal to a quarter of the grid step {dimension}. To keep the TT-rank of the preconditioner small, we choose to round it to $10^{-2}$. The choice of the number of addends and of the rounding compression is further discussed in  \reportPaper{Appendix~\ref{app:prec}}{\cite{coulaud2022}}. 
	
 To evaluate the converge of the TT-GMRES at the $k$-th iteration, we display in Section~\ref{sec3} the stopping criterion $\eta_{\ten{A}\ten{M}, \ten{b}}$, that is 
	\begin{equation}\label{eq:BE-AMb}
	\eta_{\ten{A}\ten{M}, \ten{b}}(\ten{t}_k) = \frac{\norm{\ten{A}\ten{M}\ten{t}_k - \ten{b}}}{\norm{\ten{A}\ten{M}}_2\norm{\ten{t}_k} + \norm{\ten{b}}}
	\end{equation}
	with $\ten{t}_k$ the preconditioned approximated solution at the $k$-th iteration. We compute exactly the norm of residual, of the right-hand side and of the iterative preconditioned approximated solution. The L2-norm of the preconditioner operator $\ten{A}\ten{M}$ is instead computed by the following sampling approximation. Let $\set{W}$ be a set of normalized TT-vectors generated randomly from a normal distribution, then $\norm{\ten{A}\ten{M}}_2$ is approximated by the maximum of the norm of the image of the elements of $\set{W}$ through $\ten{A}\ten{M}$, i.e.,
	\[
	\norm{\ten{A}\ten{M}}_2 \approx \max_{\ten{w}\in\set{W}}\norm{\ten{A}\ten{M}\ten{w}}.
	\]
	Similarly, the L2-norm of $\ten{A}$ is also approximated by $\max\bigl\{\norm{\ten{A}\ten{w}}s.t. \ten{w}\in\set{W}\bigr\}$. Because we are interested in the magnitude of these norms, we keep this norm estimation process simple and only compute  $10$ random TT-vectors of $\set{W}$.
	
In order to investigate main numerical features of the GMRES implementation described in the previous section we consider 
two classical PDEs that are the Poisson and convection-diffusion equation.

  	The Poisson problem writes
  	\begin{equation}
  	\label{eqNE:Lap}
		\begin{cases}
	  	&-\Delta u = f \quad\text{in}\quad \Omega = [0, 1]^3 ,\\
	  	&\;\;\; u = 0 \quad\text{in}\quad \partial \Omega,
		\end{cases}
  	\end{equation}
  	where $f:\R^{3}\rightarrow \R$ is such that the analytical solution of this Poisson problem is $u:[0,1]^3\rightarrow \R$ defined as $u(x,y,z) = (1-x^2)(1-y^2)(1-z^2)$. Let set a grid of $n$ points per mode over $\Omega$, the discretization of the Laplacian over the Cartesian grid is the linear operator $\ten{-\Delta}_d$ defined in Equation~\eqref{eqTT:4} with $d = 3$. Let $\ten{b}\in\R^{n\times n\times n}$ be the discrete right-hand side in TT-format such that $\ten{b}(i_1, i_2, i_3) = f(x_{i_1}, y_{i_2}, z_{i_3})$. 

      The convection-diffusion problem, identical to the one considered in~\cite{Dolgov2013}, writes 
  \begin{equation}
  \label{eqNE:CD}
   \begin{dcases}
	   	&-\Delta u + 2y(1-x^2)\frac{\partial u}{\partial x} -2x(1-y^2)\frac{\partial u}{\partial y} = 0 \quad\text{in}\quad \Omega = [-1, 1]^3\, , \\[5pt]
	  	&u_{\{y = 1\}} = 1\qquad\text{and}\qquad u_{\partial\Omega \setminus \{y = 1\}} = 0 \, .
   \end{dcases}
  \end{equation}
  Setting a grid of $n$ points per mode over $[-1,1]^3$, the Laplacian is discretized as in Equation~\eqref{eqTT:4} with $d = 3$. Let $\Nabla_x$ be discretization of the first derivative of $u$ with respect to mode $1$ defined as $\Nabla_x = \nabla_1 \otimes \I_n \otimes \I_n$, similarly $\Nabla_y$ is the discrete first derivative with respect to mode $2$ written as $\Nabla_y = \I_n\otimes \nabla_1 \otimes \I_n$, where $\nabla_1$ is the order-$2$ central finite difference matrix, i.e.,
  \[
  \nabla_1 = \frac{1}{2h}\begin{bmatrix}
  0  &1 & 0&\dots&0\\
  -1 & 0 & 1 &\dots&0\\
  \vdots&\ddots&\ddots&\ddots&\vdots\\
  0 & \dots& -1 & 0&1\\
  0 &0    & \dots& -1 & 0
  \end{bmatrix}.
  \]
  Let $v:[-1,1]^{3}\rightarrow \R^2$ be a function such that $v(x,y,z) = \Bigl(2y(1-x^2), -2x(1-y^2)\Bigr)$, the two components of $v$ are descretized over the Cartesian grid set on $[-1,1]^3$ defining two tensors $\ten{V}_1,\ten{V}_2 \in\R^{(n\times n)\times(n\times n)\times  (n\times n)}$ such that $\ten{V}_1 = \text{diag}(1-x^2)\otimes \text{diag}(2y)\otimes \I_n$ and $\ten{V}_2 = \text{diag}(-2x)\otimes \text{diag}(1-y^2)\otimes \I_n$.  Then the discrete diffusion term $\ten{D}$ writes
  \begin{equation}
  \label{eqCD:1}
  \begin{split}
  	\ten{D} &= \ten{V}_1\bullet\Nabla_x + \ten{V}_2\bullet\Nabla_y \\
  	&= \text{diag}(1-x^2)\nabla_1\otimes \text{diag}(2y)\otimes \I_n + \text{diag}(-2x)\otimes \text{diag}(1-y^2)\nabla_1\otimes \I_n \, .
  \end{split}
  \end{equation}
  The final operator passed to the TT-GMRES algorithm is $\ten{A} = -\ten{\Delta}_3 + \ten{D}$, the right-hand side is the TT-vector $\ten{b}\in\R^{n\times n \times n}$ and the initial guess is the zero TT-vector $\ten{x}_0$.
To ensure a fast convergence, similarly to~\cite{Dolgov2013}, we consider a right preconditioner $\ten{M}$ from Equation~\eqref{eqTT:5} for this test example.

\subsection{Main features and robustness properties \label{sec3:s1}}
In this section, we first illustrate in Section~\ref{ssb:Dolgov} the major differences between our GMRES implementation and the one proposed in~\cite{Dolgov2013} 
that mostly highlights the robustness of our variant.
We motivate the need of effective preconditioners in Section~\ref{sec3:Lap} and illustrate the performance and the
main features of preconditioned GMRES in Section~\ref{sssec:CD}.
All the experiments were performed using \texttt{python 3.6.9} and with the tensor toolbox \texttt{ttpy 1.2.0}~\cite{ttpy}.

\subsubsection{Comparison with previous tensor GMRES algorithm\label{ssb:Dolgov}}
 In this section we describe the TT-GMRES introduced in~\cite{Dolgov2013}, that we refer to as relaxed TT-GMRES, that attempts to use advanced features enabled by the inexact GMRES theory~\cite{bofr:04, gigl:07, sisz:03, sles:04}.
In particular, these inexact GMRES theoretical results show that some perturbations can be introduced in the linear operator when enlarging the
Krylov space so that the magnitude of these perturbations can grow essentially as the inverse of the true residual norm of the current iterate.
In that context the accuracy of computation of the linear operator can be relaxed, that motivated the use of this terminology in~\cite{bofr:04, gigl:07}.
The inexact GMRES theory assumes exact arithmetic so that Equation~\eqref{eq:resLS_vs_trueRes} holds.
In practice, this equality becomes invalid as soon as some loss of orthogonality appears in the Arnoldi basis so that
\begin{equation}\label{eq:diff_LS_trure_res}
\Vert \tilde{r}_k \Vert = \Vert \beta e_1 - \bar{H}_ky \Vert  \ne \Vert r_k \Vert = \Vert b  - A x_k \Vert;
\end{equation} 
that is, the norms of the least squares residual  and the true residual differ.

In a TT-computational context these inexact Krylov results motivated the heuristic presented in~\cite{Dolgov2013}, that consists in transferring the perturbation policy from the matrix to the output of the matrix-vector product. More precisely, the variable perturbation magnitude is implemented by varying the rounding threshold $\delta$ applied to the tensor resulting from the matrix-vector product along the iterations.
Furthermore, the magnitude of the rounding $\delta$ is computed using the least squares residual norm rather than the true residual norm for practical computational reasons.
A possible consequence of this choice is that  $\delta$ is somehow artificially increased.

Although the rounding are performed exactly at the same step in the two algorithms, 
there are  two differences between our TT-GMRES and the relaxed TT-GMRES~\cite{Dolgov2013}.
The first difference is related to the rounding threshold policy that is variable (or relaxed to use the terminology of the pioneer  paper on inexact GMRES~\cite{bofr:04}) and constant in our case.
We  simply define the value of $\delta$ essentially to the value of  the target accuracy in terms of backward error~\eqref{eq:BE-Ab} (\eqref{eq:BE-AMb} when a preconditioner is used).
The second difference is related to the stopping criterion that is defined in terms of backward error~\eqref{eq:BE-Ab}  in our case (\eqref{eq:BE-AMb} when a preconditioner is used) while it is based on a scaled least squares residual defined by Equation~\eqref{eq:stopCritDolgov} in~\cite{Dolgov2013}:
\begin{equation}\label{eq:stopCritDolgov}
  \tilde{\eta}_{\ten{b}} (\ten{x}_k) = \frac{\norm{\tilde{r}_k}}{\norm{\ten{b}}}.
\end{equation}
Because in practice the true residual differs from the least squares residual, this latter is monotonically decreasing towards zero, such a stopping criterion can lead to an earlier stop.

We choose this  stopping criterion based on backward error because it is the one for which, in the matrix framework, GMRES is backward stable in finite precision~\cite{paige2006}.
  	Through intensive numerical experiments~\cite{agullo2022}, we observed that our TT-GMRES inherits the same backward stability property. Indeed if $\delta$ is the rounding accuracy and $\ten{x}_k$ the GMRES solution at iteration $k$, then $\etaAb{\ten{A}}{\ten{b}}(\ten{x}_k)$ is $\mathcal{O}(\delta)$ as $\delta$ is the dominating part of the rounding error occurring during the numerical calculation. Consequently assuming $\delta \le \varepsilon$, our GMRES variant is able to ensure a $\varepsilon$-backward stable solution. 
  	\begin{figure}[!hbt]
  	\centering
  		\includegraphics[scale = 0.45, width=0.33\linewidth]{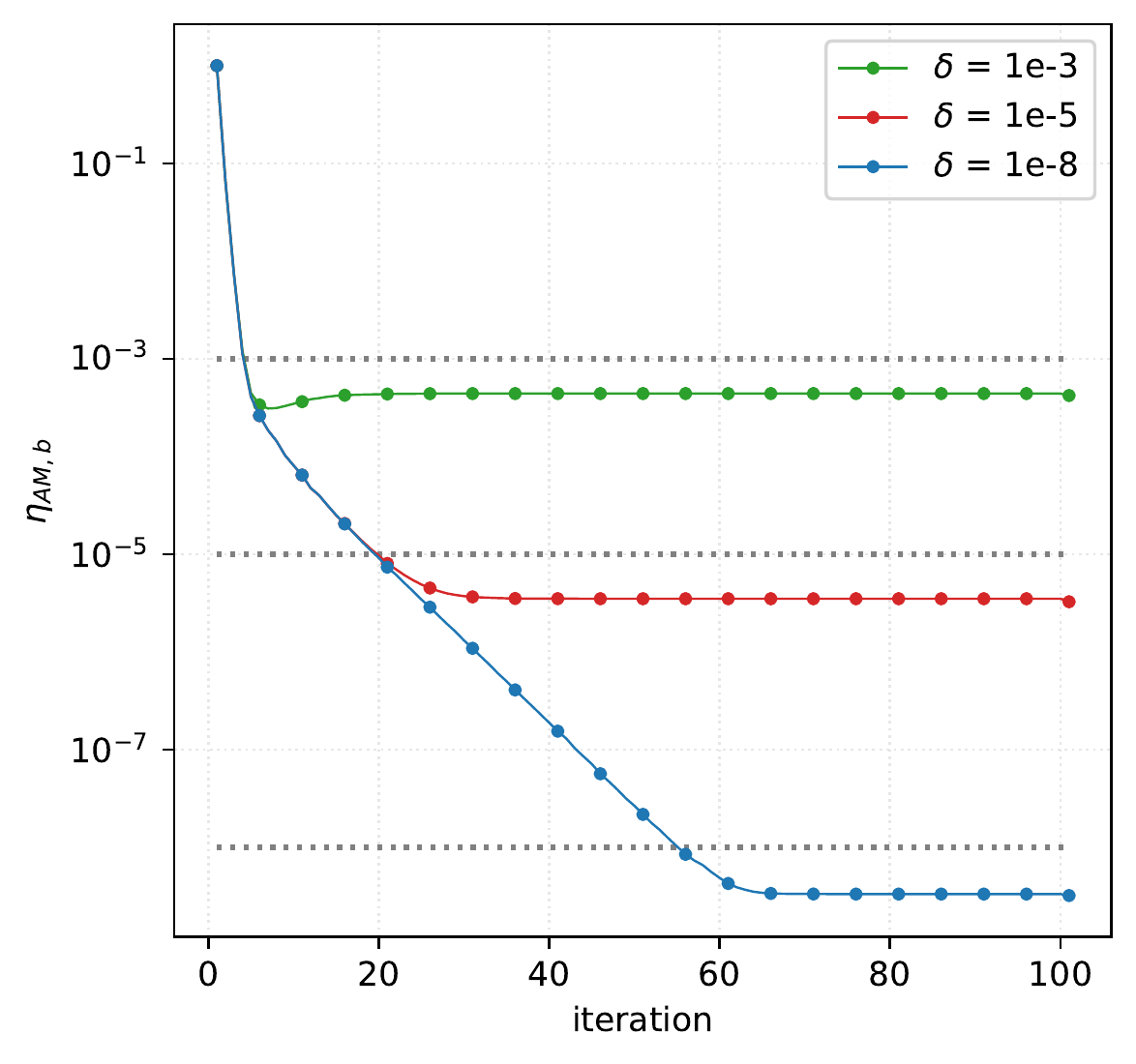}
  		\caption{Convergence history of TT-GMRES on a 3-d convection diffusion problem, $n = 64$, for three different rounding accuracies $\delta$}
  		\label{fig:CCD}
 	\end{figure}
  	{This property is well illustrated in Figure~\ref{fig:CCD} in the case of preconditioned GMRES. The $3$d convection-diffusion problem with $63$ discretization points is solved  using $3$ different rounding accuracies, i.e., $\delta\in\{10^{-3}, 10^{-5}, 10^{-8}\}$, and a maximum of $100$ iterations. For each value of $\delta$, the backward error $\eta_{\ten{AM}, \ten{b}}(\ten{t}_k)$ decreases and stagnates around $\delta$.} 

 The second significant difference between the two GMRES variants is the choice of the rounding threshold along the iterations that is constant for us and varies as the inverse $\Vert \tilde{{r}}_k \Vert$ in~\cite{Dolgov2013}.
  This variation of the rounding is illustrated in Figure~\ref{fig:CD_Dolg}.
We solve with the two different algorithms the same convection-diffusion problem with $63$ discretization  points in each space dimension. We select three different rounding accuracies $\delta\in\{10^{-3}, 10^{-5}, 10^{-8}\}$ and
perform $100$ iterations of full GMRES (i.e., no restart). In Figure~\ref{fig:CD_Dolg_delta} we see the extreme growth of the rounding threshold, when it is scaled by the norm of $\tilde{r}_k$, the least-squares residual norm that becomes smaller and smaller.
  	\begin{figure}[!htb]
  		\centering 		
  		\subfloat[Convergence history with $\tilde{\eta}_{\ten{b}}$ with least squares residual]{\includegraphics[scale=0.3, width=0.3\linewidth, height = 0.3\linewidth]{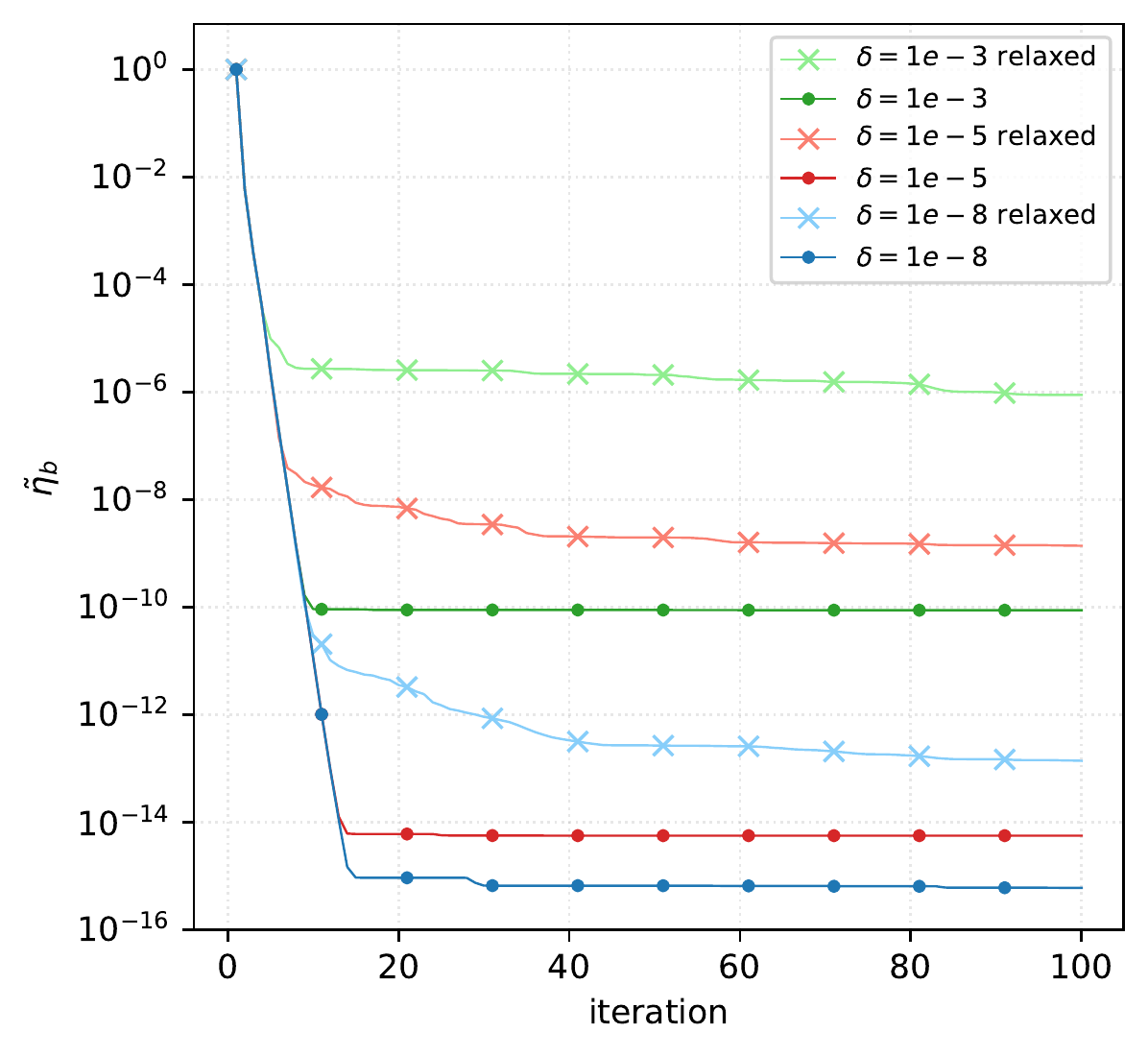}\label{bfig:CD_Dolg_etabLS}}
  		\vspace{3mm}
  		\subfloat[Convergence history with $\eta_{\ten{b}}$ with true residual]{\includegraphics[scale=0.3, width=0.3\linewidth, height = 0.3\linewidth]{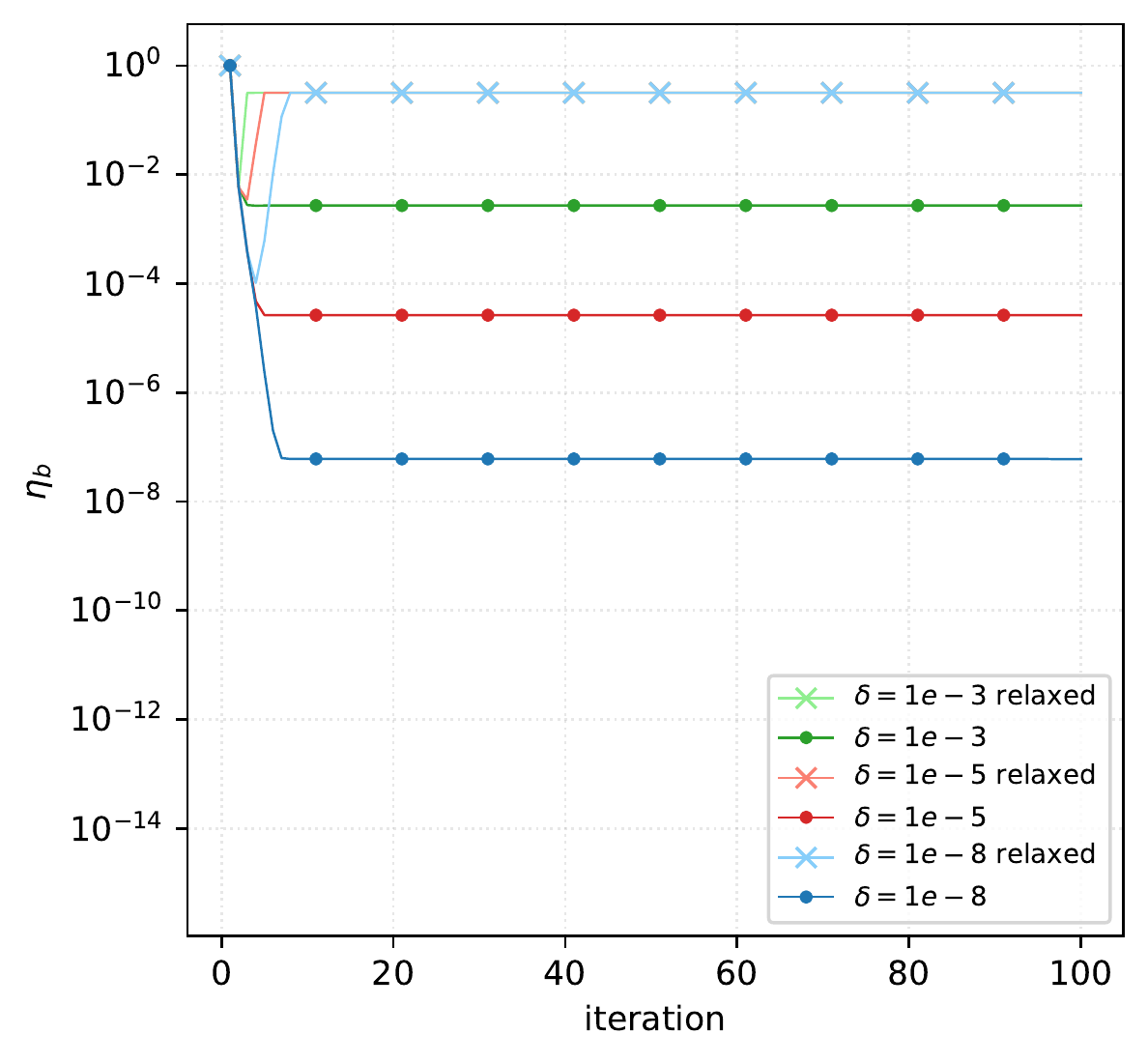}\label{bfig:CD_Dolg_etab}}
  		\vspace{3mm}
  		\subfloat[Convergence history $\eta_{\ten{AM}, \ten{b}}$ with true residual]{\includegraphics[scale=0.3, width=0.3\linewidth,height = 0.3\linewidth]{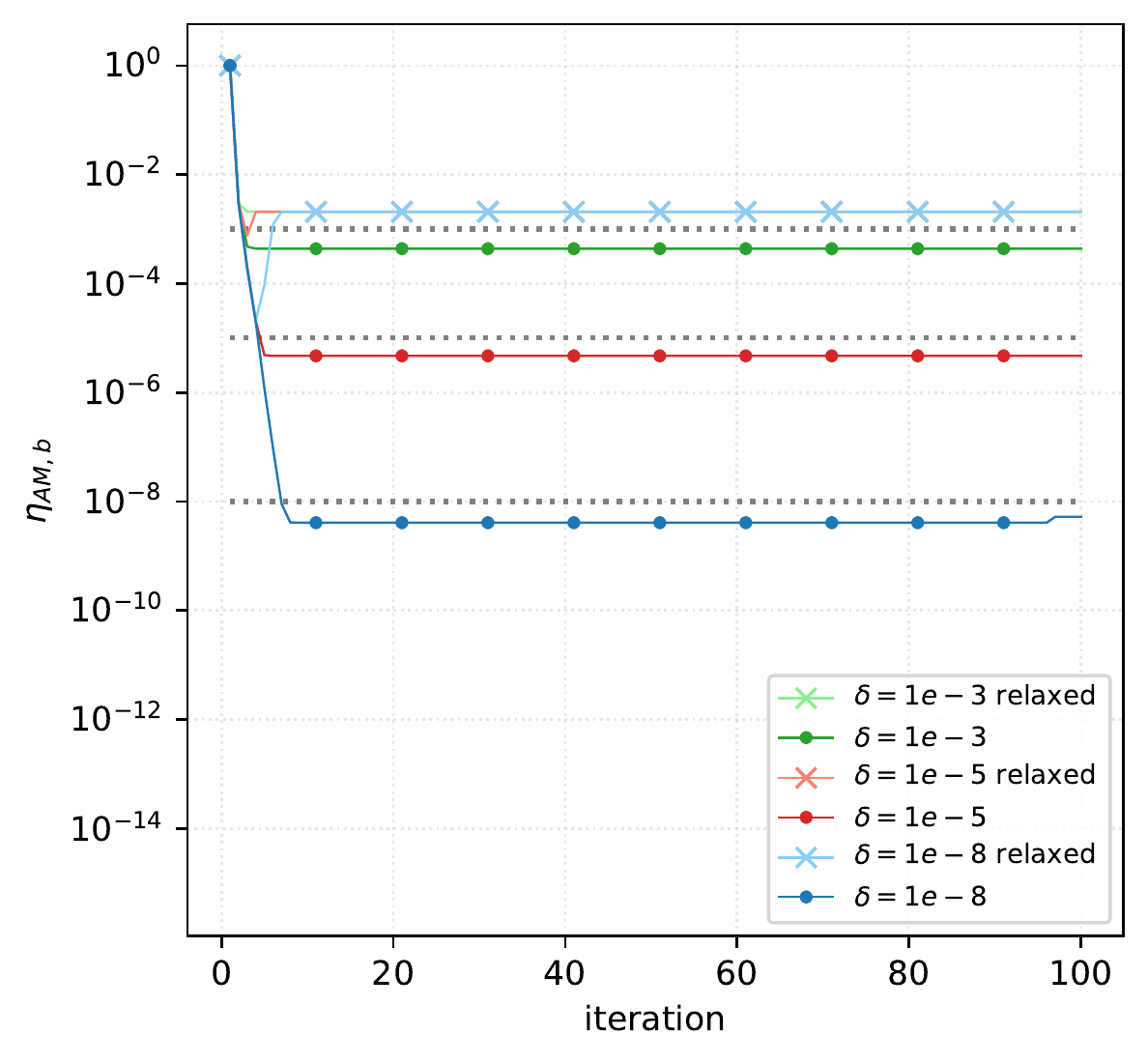}\label{bfig:CD_Dolg_etaAMb}}
  		
  		\vskip\baselineskip
  		
  		\subfloat[Max TT-rank of the last Krylov vector]{\includegraphics[scale=0.35, width=0.315\linewidth, height=0.315\linewidth]{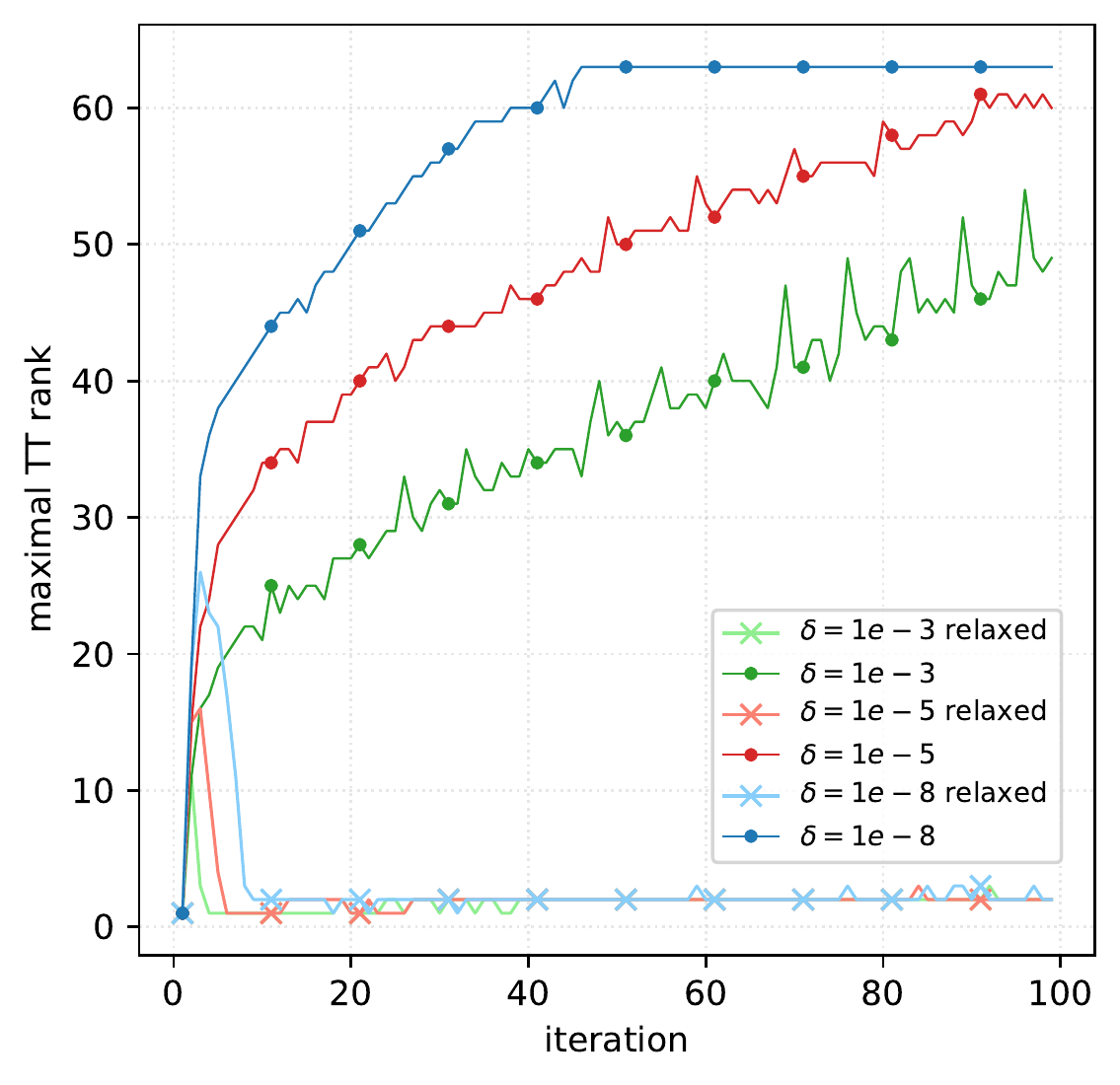}\label{fig:CD_Dolg_v-rank}}
  		\vspace{5mm}
  		\subfloat[Max TT-rank of the iterative solution]{\includegraphics[scale=0.35, width=0.315\linewidth, height=0.315\linewidth]{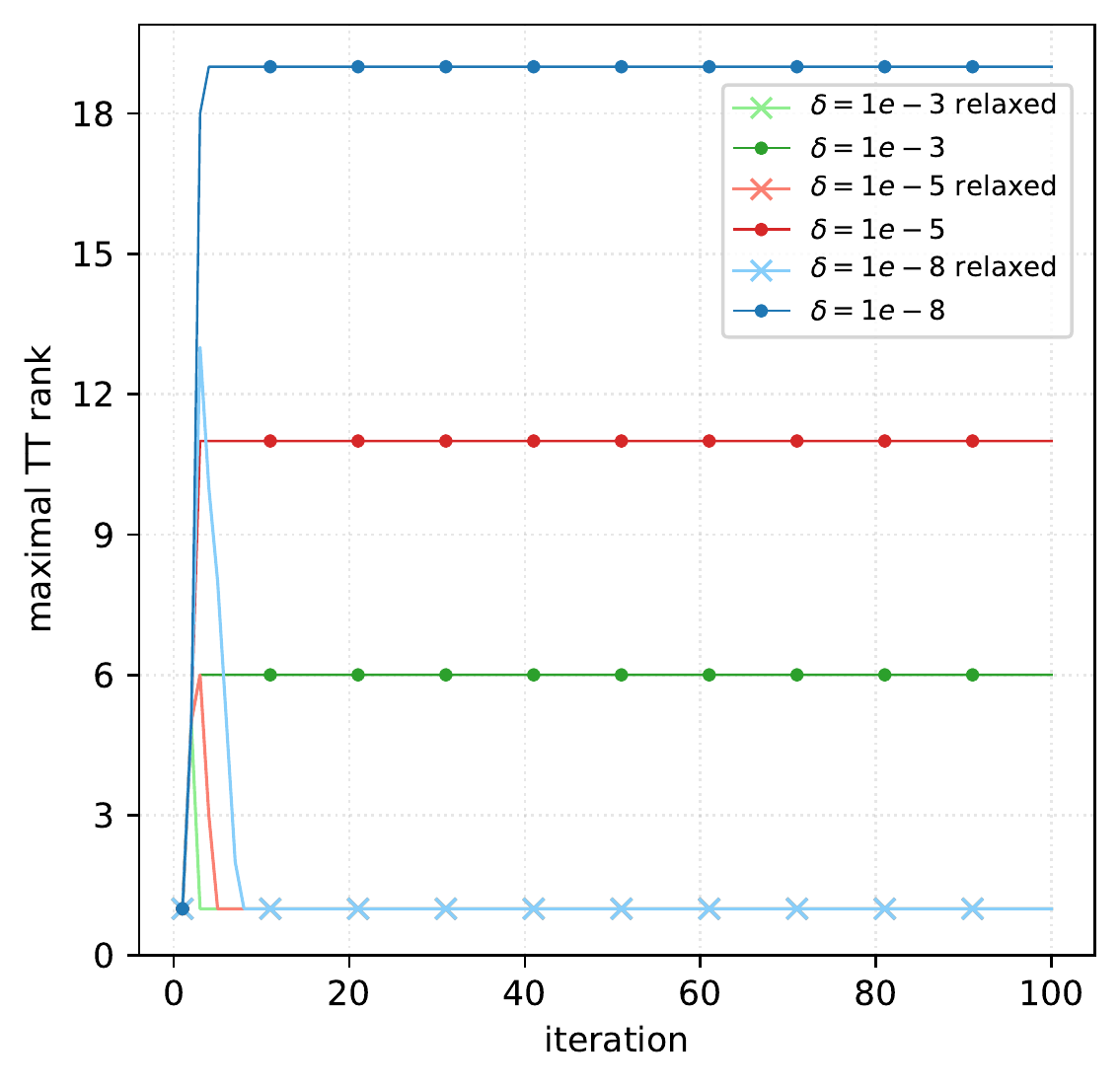}\label{fig:CD_Dolg_x-rank}}
  		\vspace{5mm}
  		\subfloat[History of the relaxed $\delta$ values]{\includegraphics[scale=0.35, width=0.315\linewidth, height=0.315\linewidth]{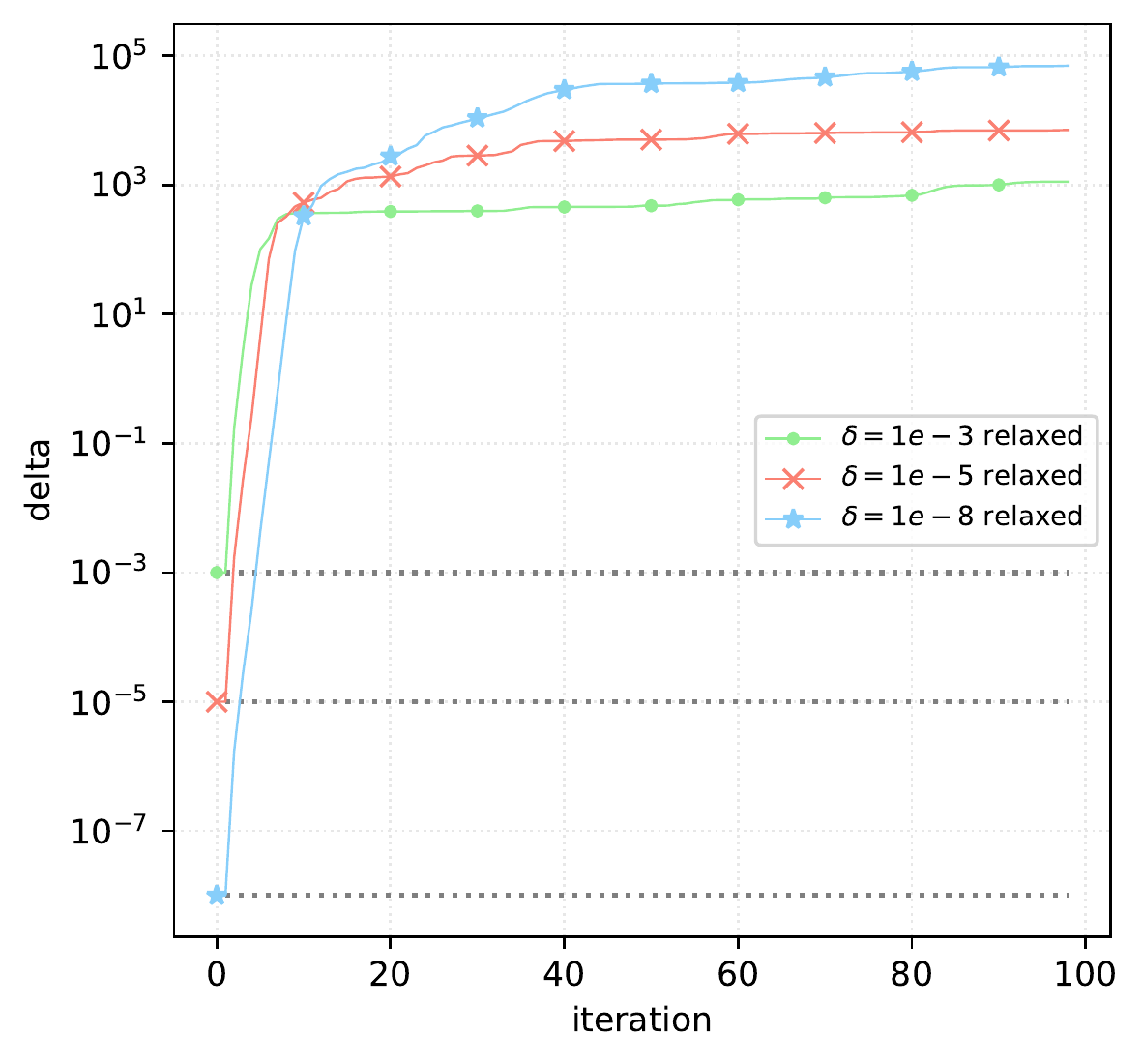}\label{fig:CD_Dolg_delta}}
  		
  		\caption{TT-GMRES and relaxed TT-GMRES for the solution of 3-d convection diffusion problem with $n = 63$}
  		\label{fig:CD_Dolg}

  	\end{figure}
 When the rounding accuracy becomes significantly large, the TT-ranks in relaxed TT-GMRES are cut to $1$, losing almost all the information carried in the tensor. 
  	Figure~\ref{bfig:CD_Dolg_etabLS} shows the scaled residual used as stopping criterion in~\cite{Dolgov2013}. 
  	We observe that if $\delta$ is not relaxed along the iterations, the value of $\tilde{\eta}_{\ten{b}}$ decreases extremely quickly, reaching $10^{-10}$ for $\delta = 10^{-3}$ and at least $10^{-14}$ for the other rounding accuracies. On the other hand if the rounding accuracy is relaxed during the iterations, we see that in all the cases $\tilde{\eta}_{\ten{b}}$ reaches at least $10^{-6}$.
 However, the comparison of Figure~\ref{bfig:CD_Dolg_etabLS} and Figure~\ref{bfig:CD_Dolg_etab} illustrates the numerical difference of the least squares residual norm and the true residual norms given by Equation~\eqref{eq:diff_LS_trure_res}. 
This comparison reveals that $\tilde{\eta}_{\ten{b}}(\ten{x}_k)$ with the relaxed $\delta$ converges, but ${\eta}_{\ten{b}}(\ten{x}_k)$, that is also a
backward error as defined in~\eqref{eq:BE-b}, does not. 
It means that the solutions computed using the relaxed $\delta$ are meaningless in terms of backward error accuracy.
Similar conclusions can be drawn from
  	Figure~\ref{bfig:CD_Dolg_etaAMb} that presents the history of $\eta_{\ten{AM}, \ten{b}}$ for the two algorithms. 
When the rounding accuracy is kept constant, we recover a backward stable behaviour similar to the one proved for finite precision calculation in classical linear system solution in matrix format. 
Indeed $\eta_{\ten{AM}, \ten{b}}$ always reaches and stagnates around the selected constant value of $\delta$.
On the contrary, when $\delta$ is relaxed at each iteration, the quantity $\eta_{\ten{AM}, \ten{b}}$ stagnates quickly slightly above $10^{-3}$, whatever the starting value of $\delta$. From these two figures, we conclude that relaxing the rounding accuracy and using $\tilde{\eta}_{\ten{b}}$ as stopping criterion, together or independently, do not provide any insight on the  quality of the computed solution.

 Obviously the choice of relaxing the rounding accuracy has a powerful effect on the rank of the last Krylov basis vector and on the solution, as illustrated by Figure~\ref{fig:CD_Dolg_v-rank} and~\ref{fig:CD_Dolg_x-rank}. Indeed in the case of the last Krylov basis vector its TT-rank oscillates around $1$ for all the iterations, after the $15$-th one approximately. Similarly the solution TT-rank stays equal to $1$, after increasing at the very first steps.  Unfortunately the computed solutions are numerically meaningless.

In the following,  we consider calculation with convergence threshold and rounding accuracy equal to $10^{-5}$, that is, $ \delta = \varepsilon =  10^{-5}$, with a maximum of $500$ iterations and restart $m = 25$. 
  	
\subsubsection{Poisson problem\label{sec3:Lap}}
We consider restarted TT-GMRES  for  the solution of the  $3$-d Poisson problem  with $n\in\{63, 127, 255\}$.
 	 \begin{figure}[!htb]
  	 	\centering
  	 	\subfloat[Convergence history]{\includegraphics[scale=0.45, width=0.33\linewidth, height=0.33\linewidth]{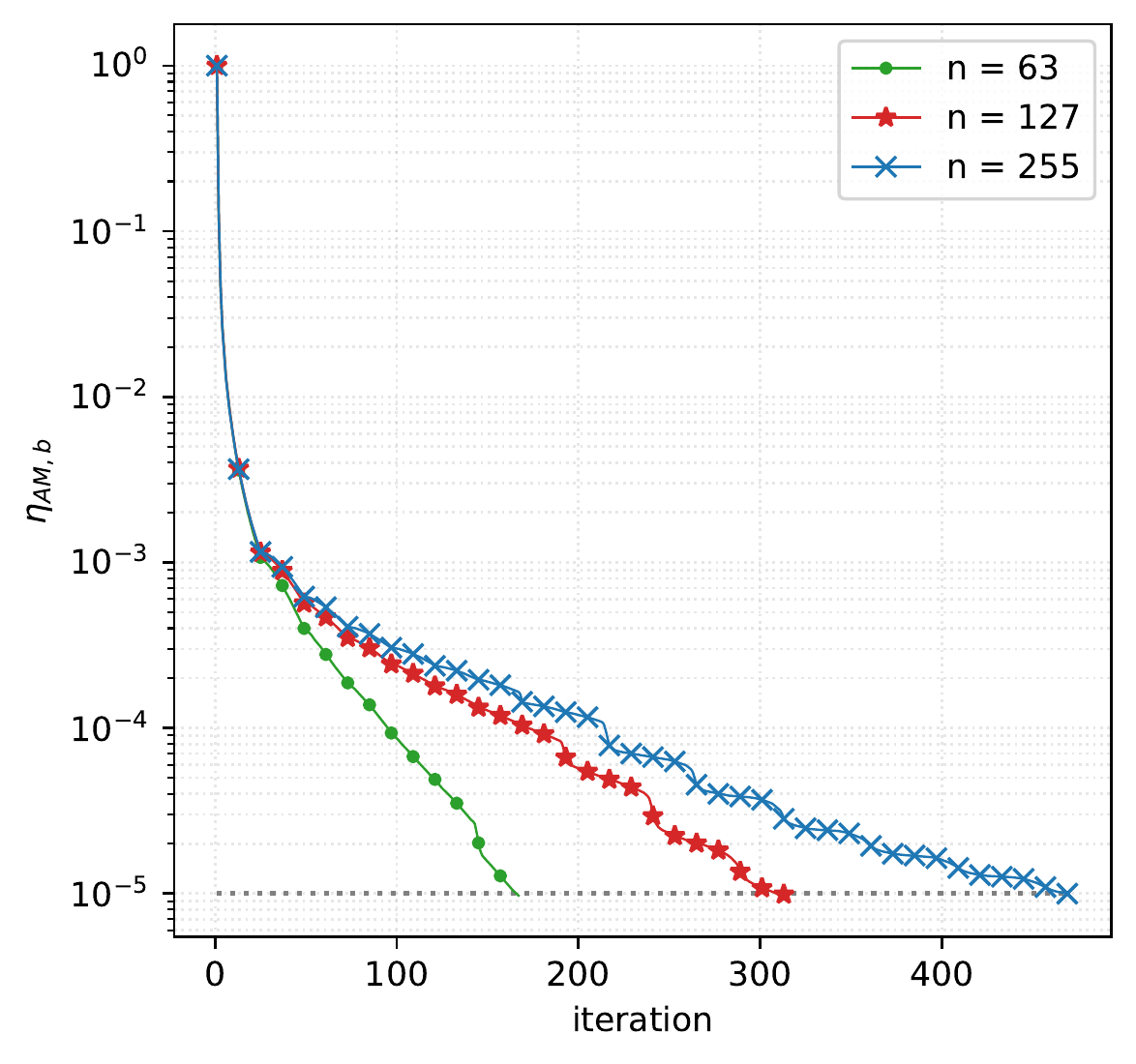}\label{fig:Lap_CH}}
  	 	\quad
  	 	\subfloat[Maximal TT-rank of the iterative solution]{\includegraphics[scale=0.45, width=0.33\linewidth, height=0.33\linewidth]{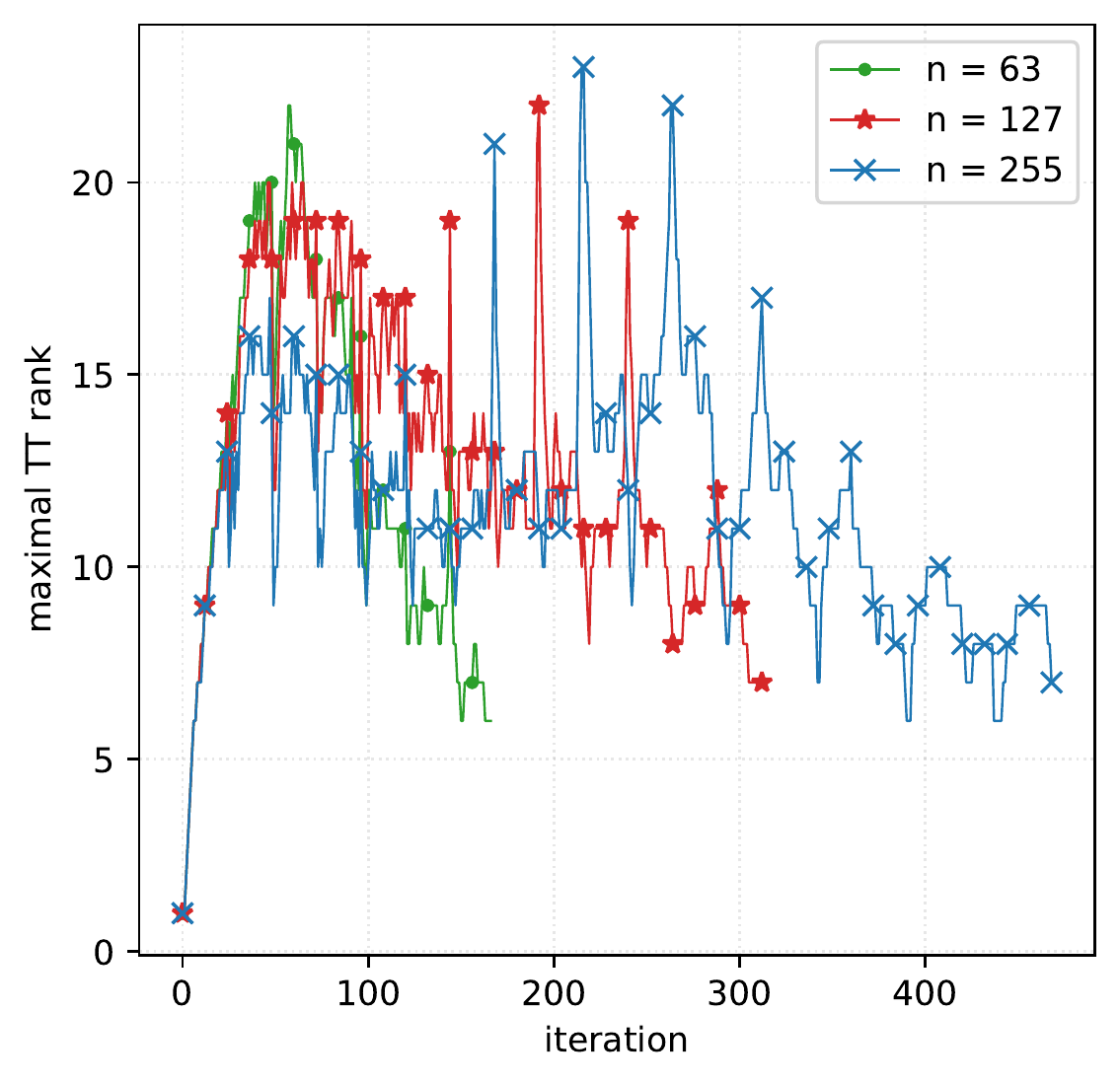}\label{fig:Lap_x-rank}}
  	 	\vskip\baselineskip
  	 	\subfloat[Maximal TT-rank of the last Krylov vector]{\includegraphics[scale=0.45, width=0.33\linewidth, height = 0.33\linewidth]{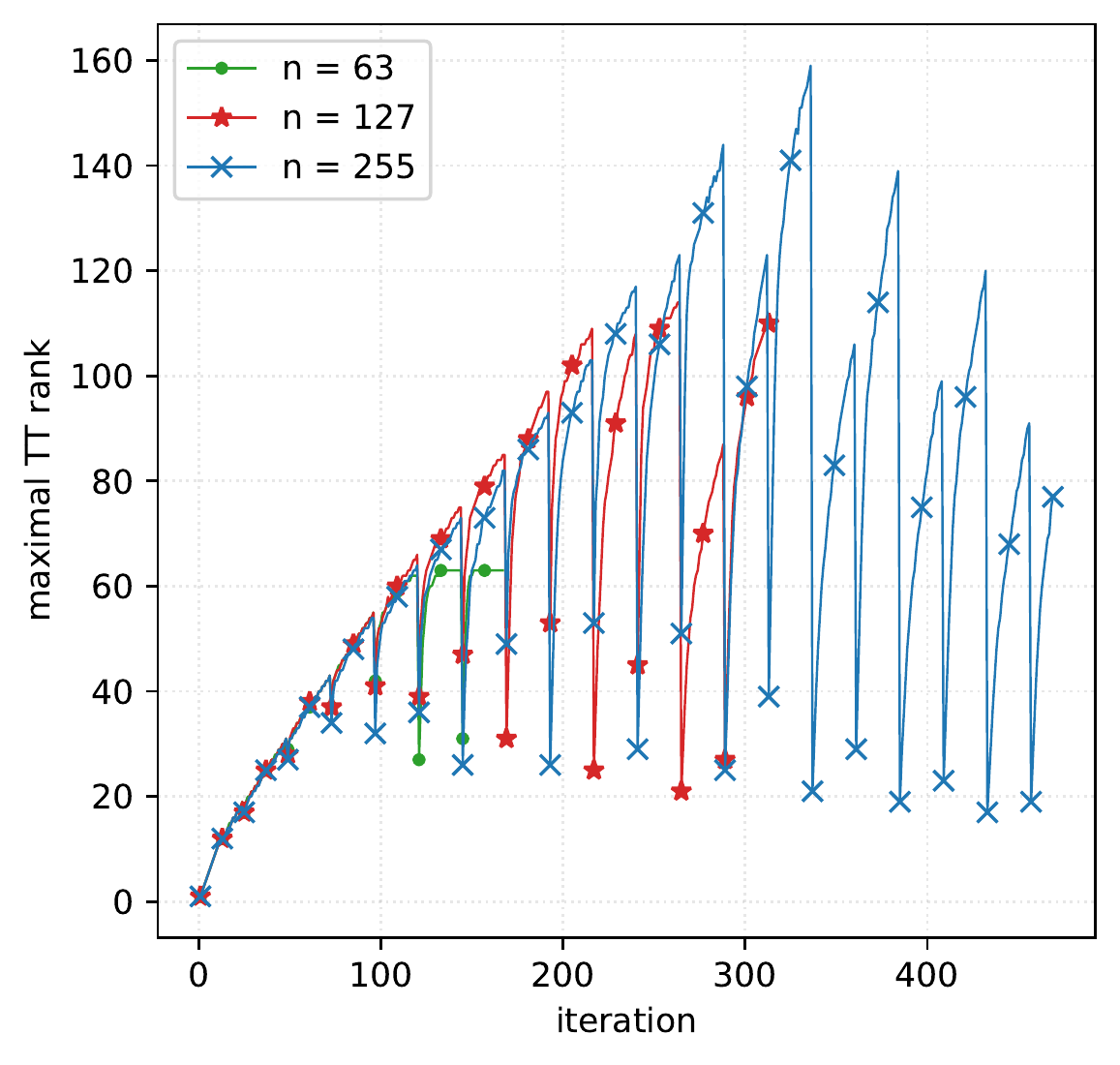}\label{fig:Lap_v-rank}}
  	 	\subfloat[Compression ratio for the last Krylov vector]{\includegraphics[scale=0.45, width=0.33\linewidth,height = 0.33\linewidth]{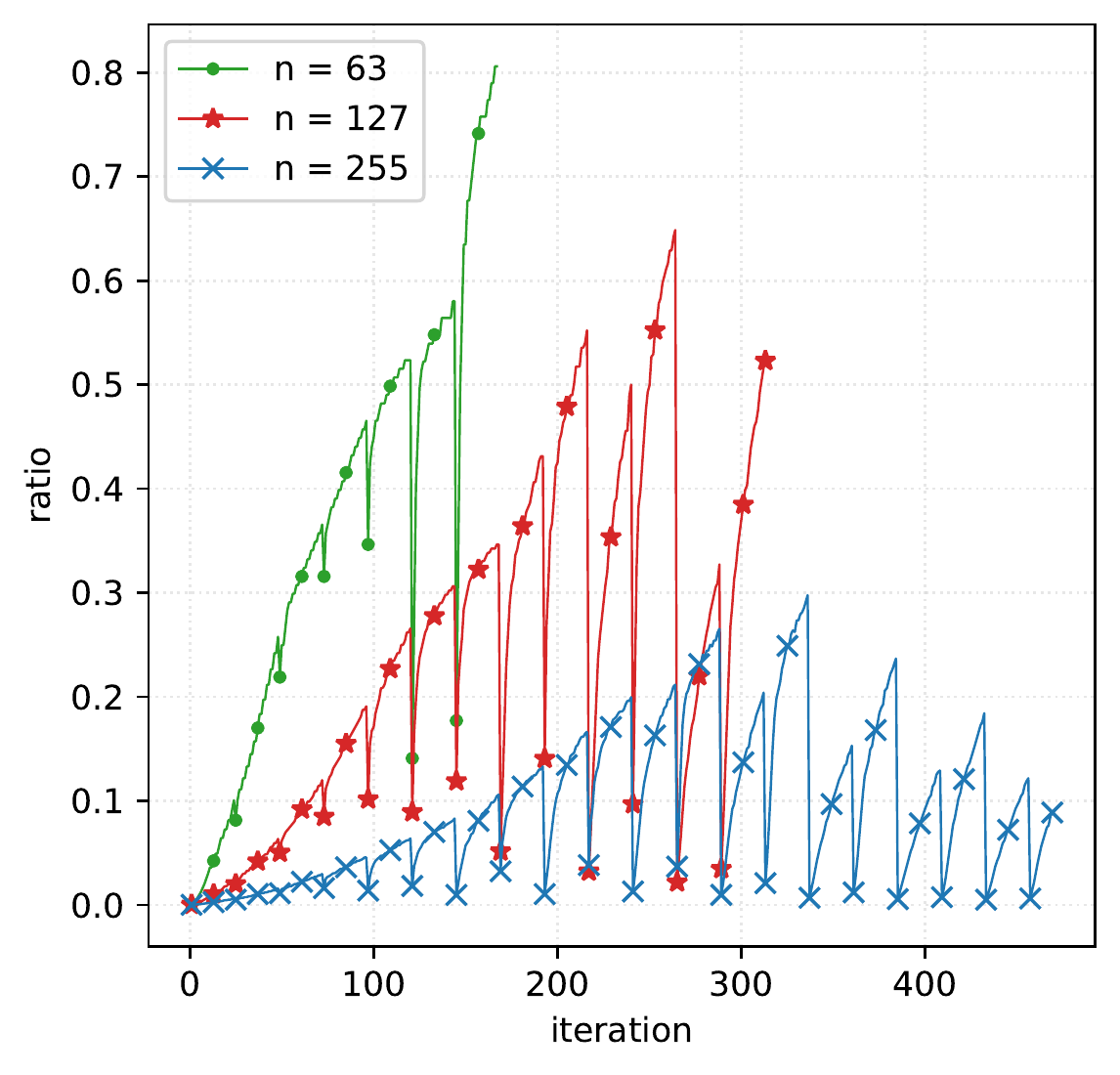}\label{fig:Lap_v-ratio}}
  	 	\subfloat[Compression ratio for the entire Krylov basis]{\includegraphics[scale=0.45, width=0.33\linewidth, height = 0.33\linewidth]{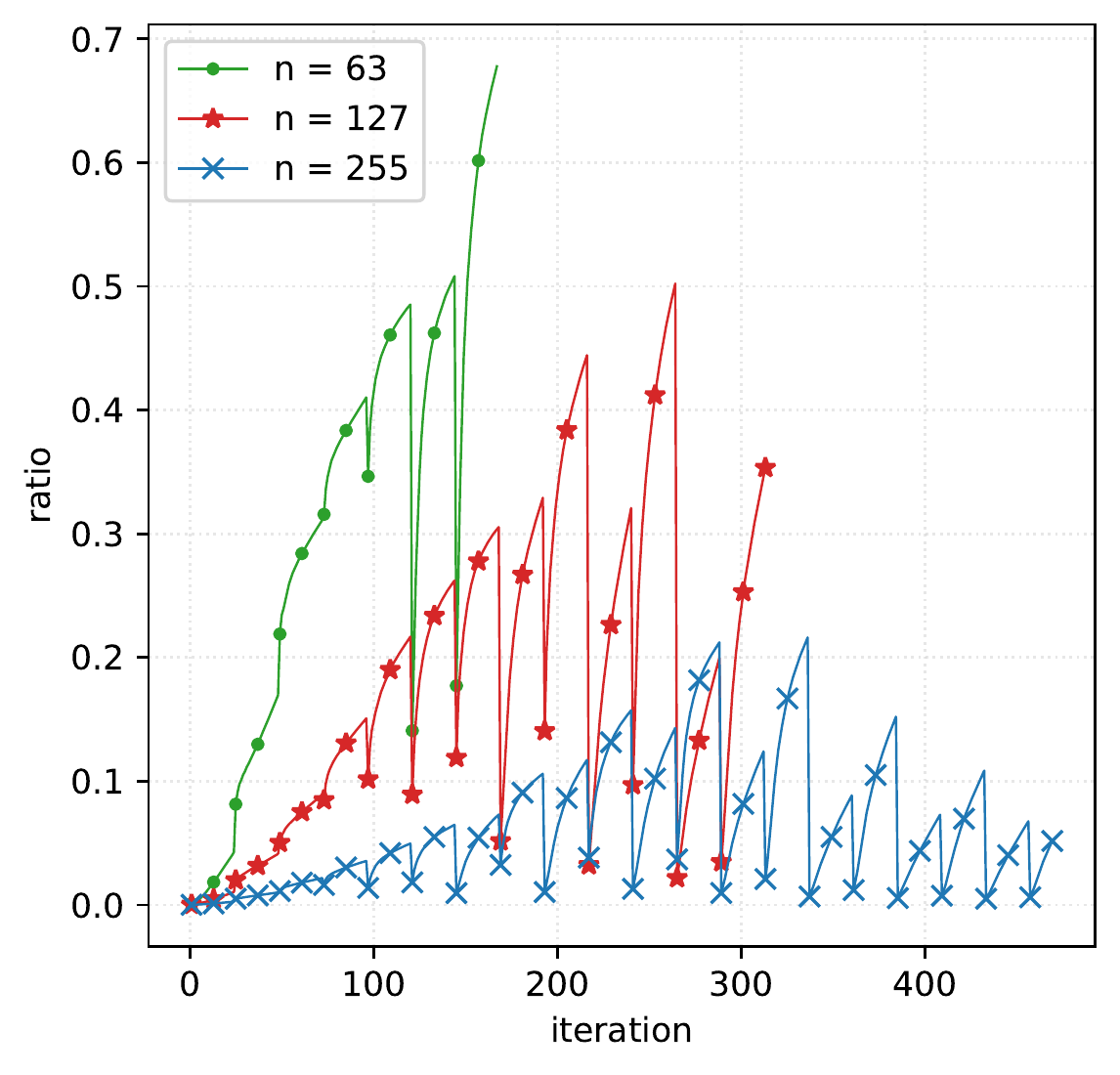}\label{fig:Lap_v-ratio_s}}
  	 	\caption{3-d Poisson problem using $\delta = \varepsilon = 10^{-5}$}
  	 	\label{fig:Lap}
  	 \end{figure}
Figure~\ref{fig:Lap_CH} shows that the algorithm is able to converge to the prescribed tolerance
$\varepsilon = 10^{-5}$  with a number of iterations that increases with the number of discretization points.
This high number of steps to solve a quite simple PDE motivates the need of a preconditioner.
Indeed in general the larger the number of TT-GMRES iterations, the larger the TT-rank growth for the Krylov basis vectors; consequently, the higher the computational cost per iteration.
For that example, it can be seen in Figure~\ref{fig:Lap_x-rank}, that the rank of the current iterate grows significantly during the
first iterations  (first $100$ iterations for $n = 63$ and the first $200$ steps for $n\in\{127, 255\}$) before decreasing in a non monotonic way.\ignore{In this specific case the TT-ranks of the iterative solution, after reaching a pick during the first $100$ iterations for $n = 63$ and the first $200$ steps for $n\in\{127, 255\}$, stay bounded with very small oscillations between $5$ and $10$, near to the last iterations.} We infer that this particular behaviour is related to  the separable nature of the analytical solution, which is \ignore{in TT-format has TT-ranks, actually I guess it's the rank to be equal to 1 and this implies TT-rank equal to 1 equal to $1$, i.e.,}
  	\[
  		u(x, y, z) = \begin{bmatrix}
  		\text{diag}(1-x^2)
  		\end{bmatrix}\otimes
  		\begin{bmatrix}
  		\text{diag}(1-y^2)
  		\end{bmatrix}\otimes
  		\begin{bmatrix}
  		\text{diag}(1-z^2)
  		\end{bmatrix}
  	\]
  	with rank $1$ and as consequence its TT-rank is also bounded by $1$.
  	After some iterations TT-GMRES seems to capture the main structure of the solution, being able to almost halve the TT-ranks, as it is visible in Figure\ref{fig:Lap_x-rank}. Another quantity monitored during the iterations is the growth of the last Krylov vector TT-ranks. In Figure~\ref{fig:Lap_v-rank} the maximum TT-rank of the last Krylov vector presents a steep increase during a first phase, followed by slight decreasing phase. The behaviour of the maximum TT-rank establishes the trend in the compression ratio of the last vector and of the entire basis. Indeed the curves of Figures \ref{fig:Lap_v-ratio} and \ref{fig:Lap_v-ratio_s} are the same of Figure~~\ref{fig:Lap_v-rank} scaled by a constant, equal to $n^3$ for the first and $kn^3$ for the second where $k$ is equal to the current iteration in the restart. Lastly in Figure~\ref{fig:Lap_v-rank} mainly during the second phase, there are many consecutive drops in the maximum TT-ranks which appear with a specific frequency. They are due to the restart after every other $25$-th iteration. In fact at restart the new Krylov vector is equal to the normalized rounded residual, whose basic starting TT-ranks is the one of $\ten{x}$, equal to $21$ at maximum. Lastly notice that in the worst case storing the last Krylov vector and the entire Krylov basis request approximately $80\%$ for $n = 63$  of the memory that would be used for storing entirely them.   
Furthermore, this ratio decreases when the number of points per mode increases (i.e., $n\in\{63, 127, 255\}$), that is an appealing feature of the TT-format that allows the solution of larger problems for a given memory budget compared to  the situation where the full tensors would have to be stored.
    \subsubsection{Convection-diffusion\label{sssec:CD}}
    \begin{figure}[!htb]
    	\centering
    	\subfloat[Convergence history]{\includegraphics[scale=0.45, width=0.33\linewidth, height=0.33\linewidth]{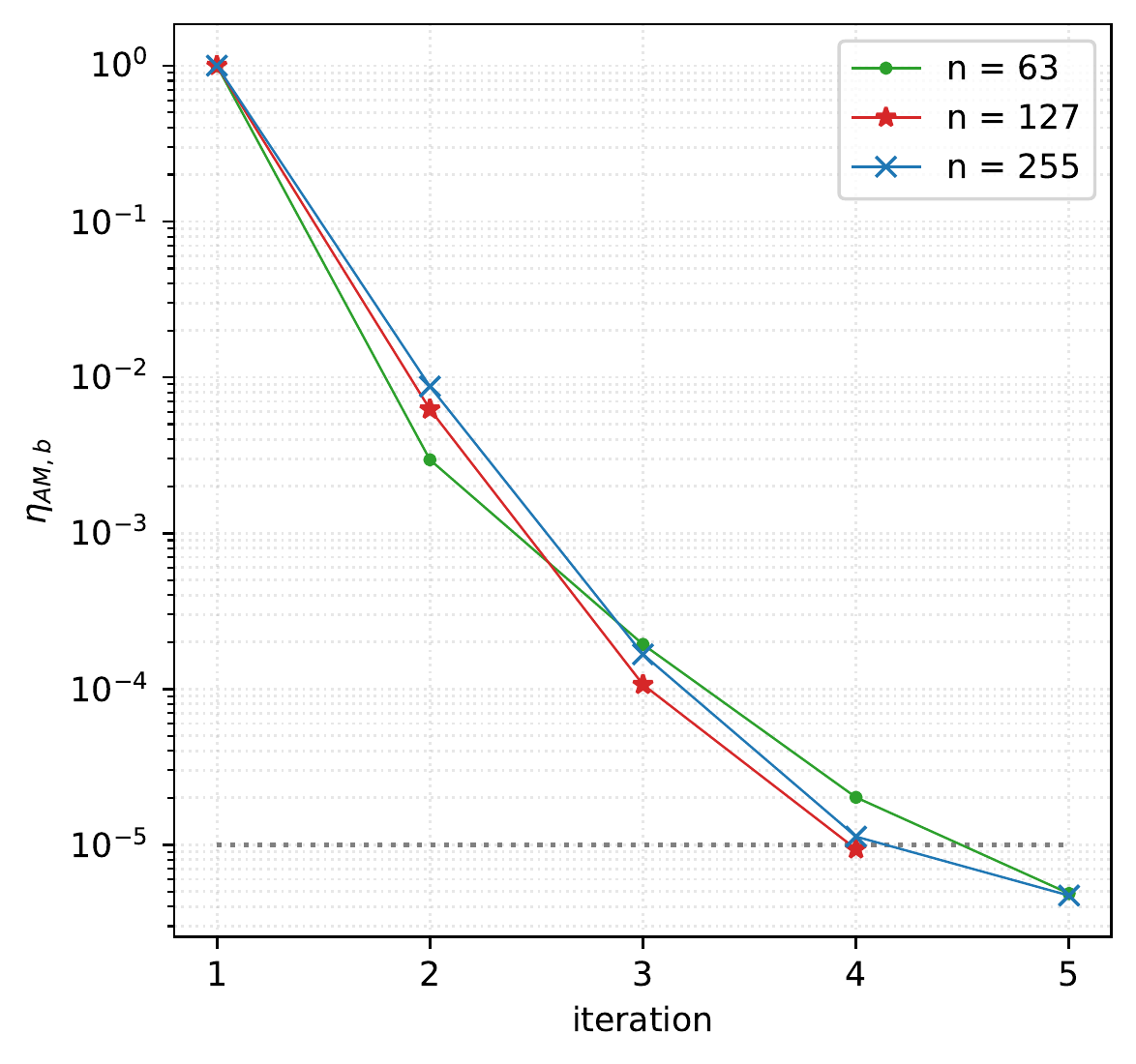}\label{fig:CD_CH}}
    	\quad
    	\subfloat[Maximal TT-rank of the last Krylov vector ]{\includegraphics[scale=0.45, width=0.33\linewidth, height=0.33\linewidth]{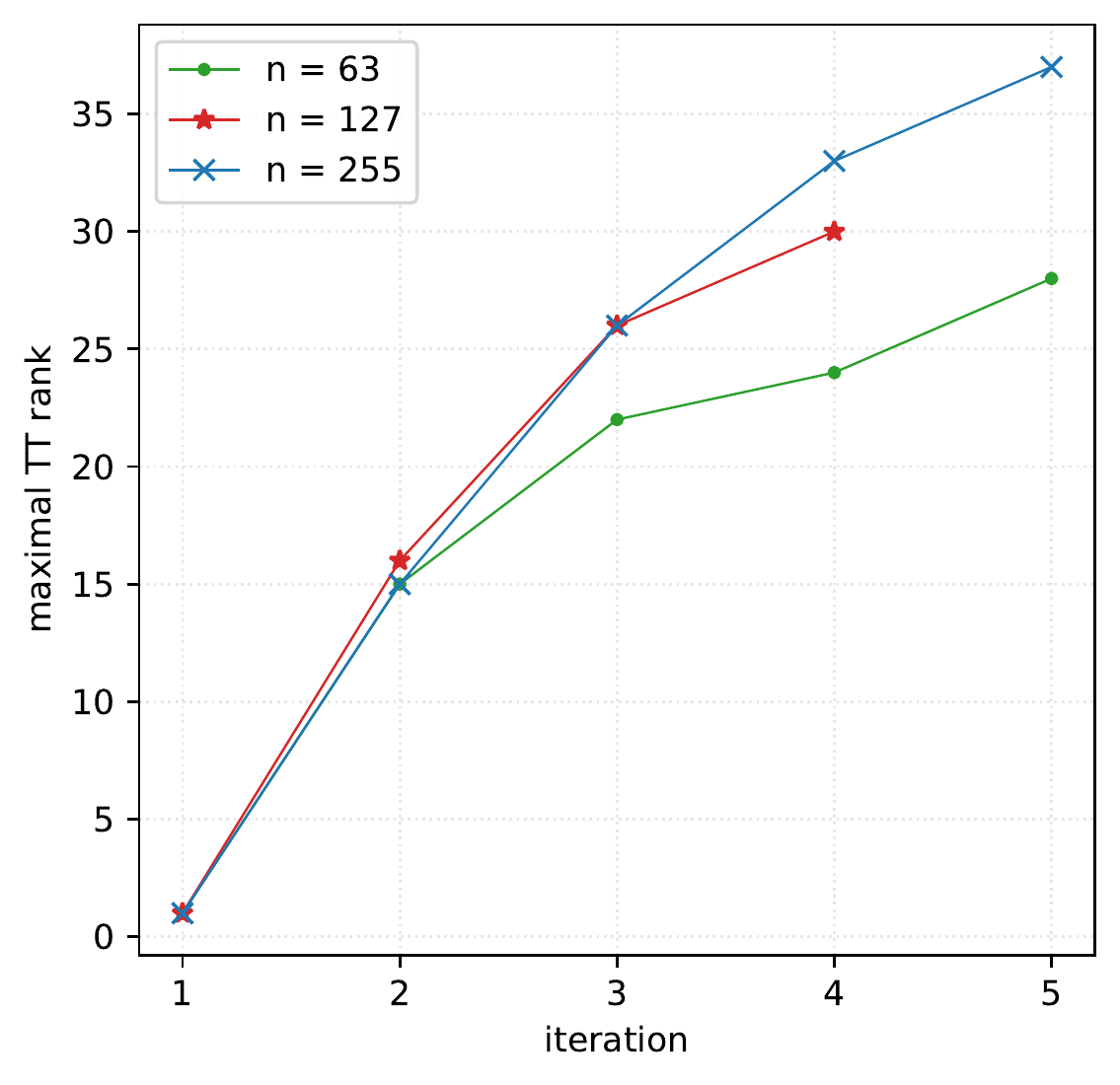}\label{fig:CD_v-rank}}
    	\vskip\baselineskip
    	\subfloat[Compression ratio for the last Kyrolv vector]{\includegraphics[scale=0.45, width=0.33\linewidth, height=0.33\linewidth]{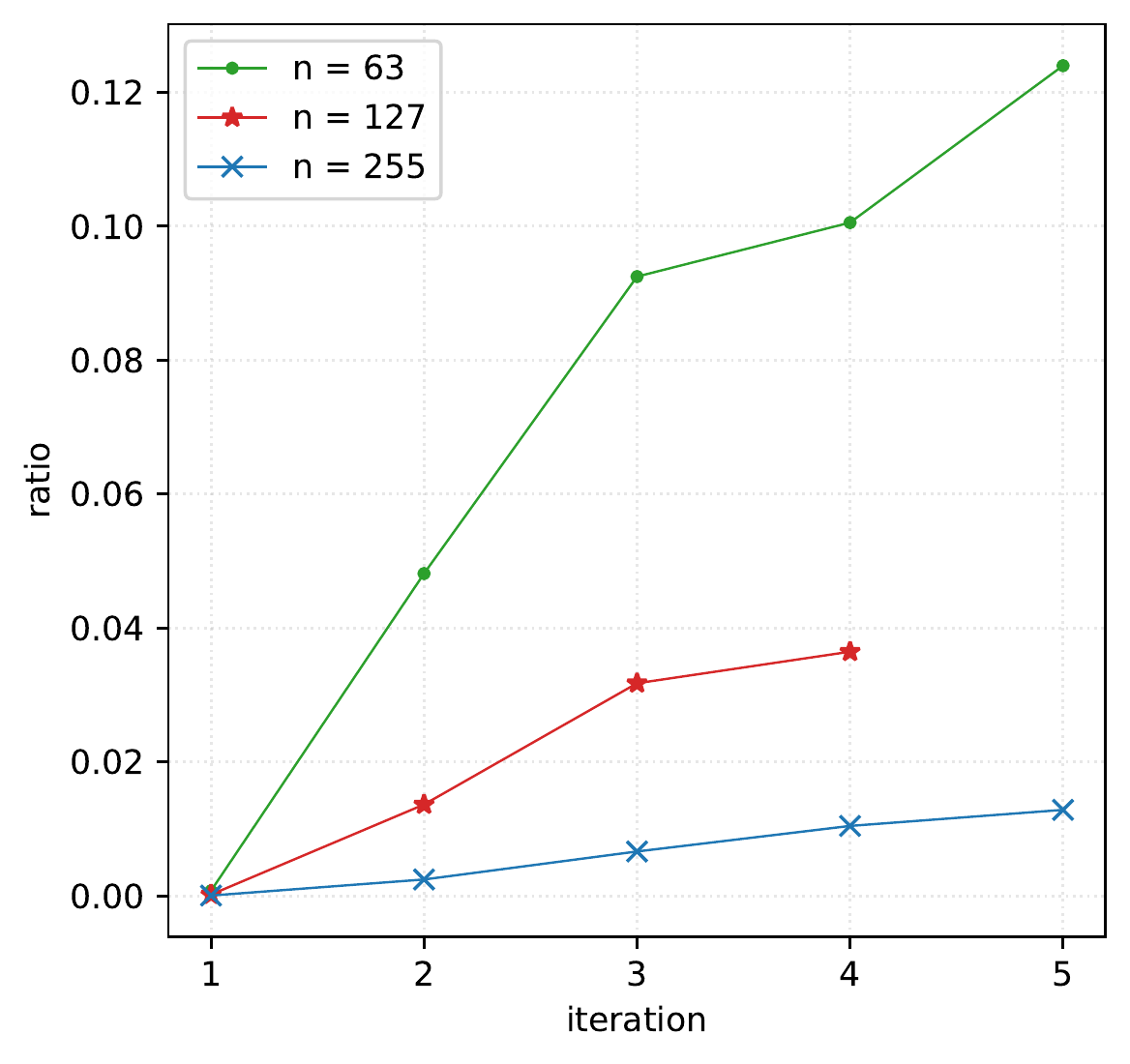}\label{fig:CD_v-ratio}}
    	\quad
    	\subfloat[Compression ratio for the entire Krylov basis]{\includegraphics[scale=0.45, width=0.33\linewidth, height=0.33\linewidth]{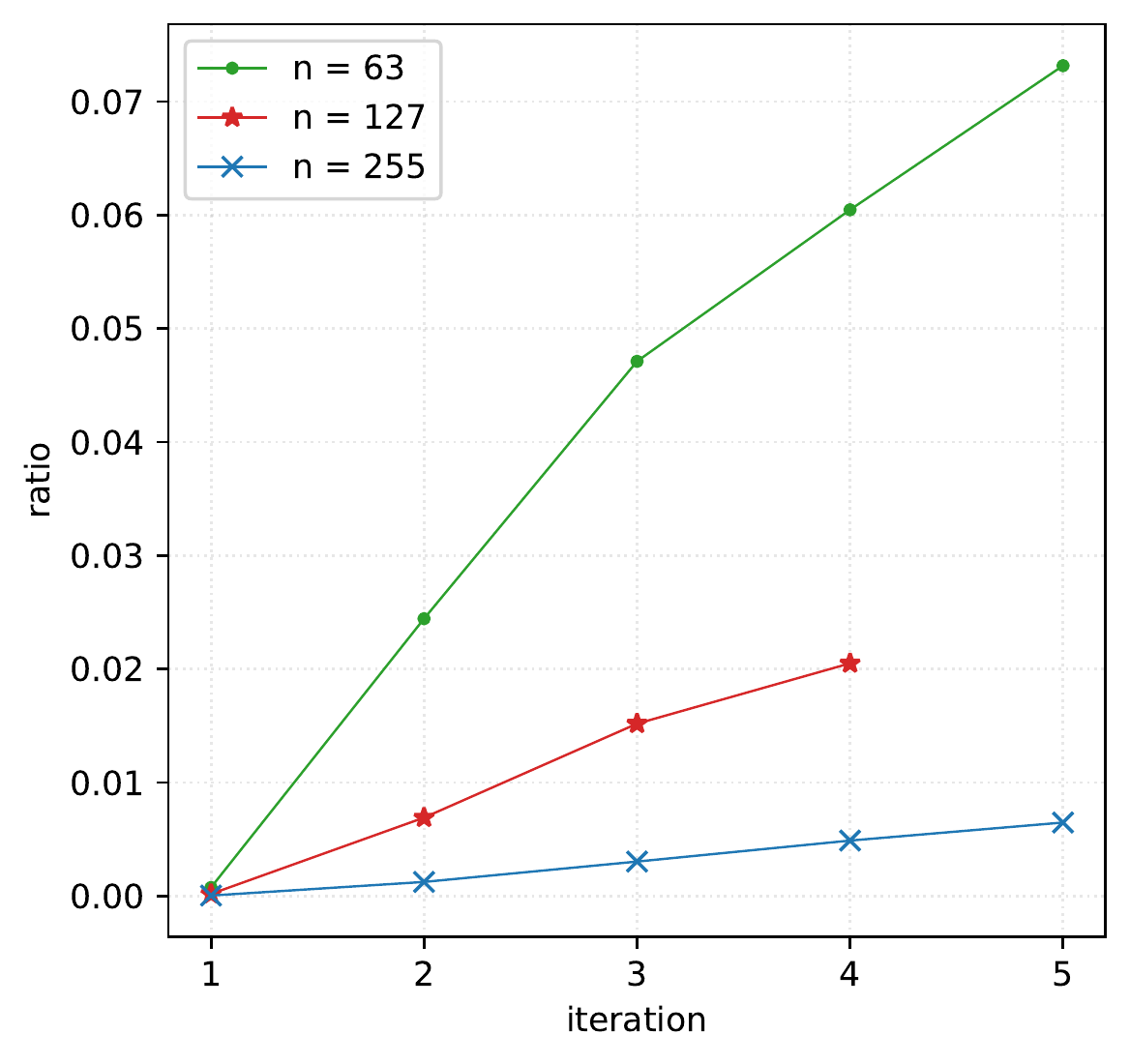}\label{fig:CD_v-ratio_s}}
    	\caption{3-d Convection diffusion using $\delta = \varepsilon = 10^{-5}$}
    	\label{fig:CD}
    \end{figure}
  We test three different grid {dimensions}, i.e., $n\in\{63, 127, 255\}$, with preconditioner $\ten{M}$ from Equation~\eqref{eqTT:5} with $q\in\{16,32\}$. Indeed without it, even with the smallest {dimension}, TT-GMRES does not converge to the prescribed tolerance $\varepsilon = 10^{-5}$ in a reasonable number of iterations. However using the preconditioner defined in Equation~\eqref{eqTT:5}, an approximated solution is found in $5$ or less iterations, as displayed in Figure~\ref{fig:CD_CH}. The preconditioner in this case has an extremely strong effect, from which  the TT-rank growth and the memory consumption benefit. In Figure~\ref{fig:CD_v-rank} the maximum TT-rank exceeds in the worst case the value $35$, but to fully interpret this information the compression ratio must be taken into consideration. In fact, Figure~\ref{fig:CD_v-ratio} shows that in the worst case to store the last Krylov vector in TT-format we use approximately $12\%$ of the memory we would need to store the full tensor. Similarly in Figure~\ref{fig:CD_v-ratio_s} we see that storing in TT-format the entire Krylov basis request in the worst case only $7\%$ of the memory that would be used to store the full tensors basis.
 Although not reported in this document, a more stringent accuracy would require a smaller rounding threshold and consequently a larger memory to store the TT-vectors.

  
	
%
 	\subsection{Solution of parameter dependent linear operators\label{sec3:s2}}
  	This section focuses on 4-d PDEs, namely parametric convection-diffusion and stationary heat equations. The domain of both problems is obtained as a Cartesian product of a $3$-d space domain and a further parameter space. The common idea for these PDEs is solving for all discrete parameter values simultaneously, getting an ``all-in-one'' solution. The structure of the operators enables us to check numerically the quality of the theoretical bounds stated in Section\ref{ssec:Ai1}.   

  \subsubsection{Parametric convection diffusion\label{sssec:PCD}}
  The parametric convection diffusion problem is a variation of Problem~\eqref{eqNE:CD}, defined as
  \begin{equation}
  	\label{eqNE:PCD}
  	\begin{dcases}
  		&-\alpha\Delta u + 2y(1-x^2)\frac{\partial u}{\partial x} -2x(1-y^2)\frac{\partial u}{\partial y} = 0 \quad\text{in}\quad \Omega = [-1, 1]^3 \, ,\\[5pt]
  		&u_{\{y = 1\}} = 1\qquad\text{and}\qquad u_{\partial\Omega \setminus \{y = 1\}} = 0 \, .
  	\end{dcases}
  \end{equation}
  As in Section~\ref{sssec:CD} let define a grid of $n$ points along each direction of $\Omega$, then the final discrete operator of this PDE is $\ten{A}_\alpha = \alpha\ten{\Delta}_3 + \ten{D}$ with $\alpha \in [1,10]$ and $\ten{D}$ defined in Equation~\eqref{eqCD:1}. Similarly, the right-hand side $\ten{c}_\alpha\in\R^{n\times n\times n}$ depends on the parameter $\alpha\in [1,10]$ because of the boundary conditions.
  To solve for multiple discrete values of $\alpha$, getting an ``all-in-one'' problem and solution, we tensorize $\ten{\Delta}_3$ and $\ten{D}$ by a diagonal matrices, adding a fourth dimension. The tensor operator  for the simultaneous solution is $\ten{A}\in\R^{(p\times p)\times(n\times n)\times(n\times n)\times(n\times n)}$ defined as
  \[
  	\ten{A} = A\otimes\ten{\Delta}_d + \I_p\otimes \ten{D} \, ,
  \]
  where $A = \text{diag}(\alpha_1, \dots \alpha_p)$ with $\alpha_i\in [1,10]$ logarithmically distributed for $i\in\{1,\dots, p\}$. {The right-hand side of the ``all-in-one'' problem is $\ten{b}\in\R^{p\times n\times n \times n }$ such that
  \[
  	\ten{b}^{[\ell]} = \frac{1}{\norm{\ten{c}_{\alpha_\ell}}}\ten{c}_{\alpha_\ell}
  	\qquad\text{for}\qquad \ell \in\{1, \dots, p\}
  \]
  using the slice notation introduced in Section~\ref{ssec:Ai1}. By construction $\norm{\ten{b}} = \sqrt{p}$, i.e., the discrete ``all-in-one'' problem fits into the hypothesis of Proposition~\ref{prop:eta_Ab} and \ref{prop:eta_Ab_m}}. Remark that the ``all-in-one'' linear operator is directly constructed as TT-matrix from the TT-matrix of the single linear system, while the ``all-in-one'' right-hand side is constructed as full tensor and then converted into a TT-vector. 
  
TT-GMRES is used for solving the ``all-in-one'' linear system for $n\in\{63, 127, 255\}$ and $p = 20$, with the preconditioner $\overline{\ten{M}}$ defined in Equation~\eqref{eqTT:5} with $q\in\{16,32\}$ tensorized with the identity
\begin{equation}
	\label{eq:prec_d+1}
	\ten{M} = \I_p \otimes \overline{\ten{M}}.
\end{equation} Figure~\ref{fig:PCD_CH} shows that the algorithm converges in less than $20$ iterations for the first two values of $n$ and in less than $25$ for $n = 255$; that is, no restart is needed. For the computational side, Figure~\ref{fig:PCD_v-rank} displays the maximal TT-rank of the last Krylov vector, which in the worst case in lower than $100$. This result translates in terms of memory by a need of slightly more than $4\%$ of the memory that would be required to store the full Krylov vector in the worst case, as highlighted by Figure~\ref{fig:PCD_v-ratio}. Looking at the cost of storing the entire Krylov basis in Figure~\ref{fig:PCD_v-ratio_s}, we see that TT-format requires around $2\%$ of the memory necessary to store the entire  Krylov basis in full tensor format. 

 We now investigate the tightness of the bound given in Proposition~\ref{prop:eta_Ab} and \ref{prop:eta_Ab_m}.
 Figure~\ref{fig:PCD_SvsM_eta_AM_b} shows the quality of the bound for $\eta_{\ten{b}_\ell}$ for $\ell\in\{1,\dots, p\}$. For all the values of $n$, the $\eta_{\ten{b}_1}$ curve dominates the other during the first half of the iterations. In the optimal case, the difference between $\eta_{\ten{b}_\ell}$ and $\eta_{\ten{b}}$ is lower than one order of magnitude. 
  To plot the $\eta_{\ten{AM}, \ten{b}}$ bound from Proposition~\ref{prop:eta_Ab}, we define a vector $\upsilon_\ell\in\R^{w}$ whose $k$-th component corresponds to the value of the coefficient $\rho_{\ell}$ from Equation~\eqref{eqP2:T} evaluated for the solution at the $k$-th iteration, i.e.,
  \[
  	\upsilon_\ell(k) = \rho_\ell(\ten{t}_k)\qquad\text{for every}\qquad k\in\{1,\dots, w\}
  \]
  with $w$ equal to the number of iterations to converge. Let $\ell_m$ and $\ell_M$ the parameter index for which the norm of $\upsilon_\ell$ is minimal and maximal respectively, i.e.,   
  \begin{equation}
  \label{eq3:mM}
  	\ell_m = \argmin_{\ell\in\{1,\dots,p\}} \norm{\upsilon_\ell}\quad\text{and}\quad\ell_M = \argmax_{\ell\in\{1,\dots,p\}} \norm{\upsilon_\ell}
  \end{equation}
   which in our specific case are equal to $1$ and $14$ respectively. 
   In Figure~\ref{fig:PCD_SvsM_eta_AM_b} we display in  $\eta_{\ten{A}\ten{M},\ten{b}}(\ten{t}_k)$ scaled by $\rho_{\ell}$ 
(see Equation~\eqref{eqP2:T} from Proposition~\ref{prop:eta_Ab}) 
and by $\rho^{*}$ (see Equation~\eqref{eqCP2:T1} from Corollary~\ref{cor:eta_Ab}) versus $\eta_{\ten{A}_\ell\overline{\ten{M}}, \ten{b}_\ell}(\ten{t}^{[\ell]}_{k})$ for $\ell\in\{1, 14\}$ and for all the values of $n$. 
	
   The three scaled curves overlap from the third iterations for all the grid {dimensions}, meaning that the approximation of the scaling coefficient given by $\rho^*$ is extremely valid in this example. We see that the orange curve corresponding to $\eta_{\ten{A}_5\overline{\ten{M}}, \ten{b}_5}$ and the blue one for $\eta_{\ten{A}_{20}\overline{\ten{M}}, \ten{b}_{20}}$
   intersect frequently, with a difference of one order at most. Moreover the difference between $\eta_{\ten{A}_{5}\overline{\ten{M}},\ten{b}_5}$ and $\eta_{\ten{AM},\ten{b}}$ scaled by $\rho_5$ is lower than one order of magnitude in the optimal case, while in the worst case it is not larger than two orders. Therefore we conclude that for this PDE the bound of the ``all-in-one'' for the individual solution is quite tight. {Notice that to estimate $\rho^*$ no extra computation is required, while the norm of $\ten{A}_\ell\overline{\ten{M}}\ten{t}_k^{[\ell]}$ has to be computed to get the value of $\rho_{\ell}(\ten{t}_k)$.}
    \begin{figure}[!htb]
        \centering
        \subfloat[Convergence history]{\includegraphics[scale =0.45, width=0.33\linewidth, height=0.33\linewidth]{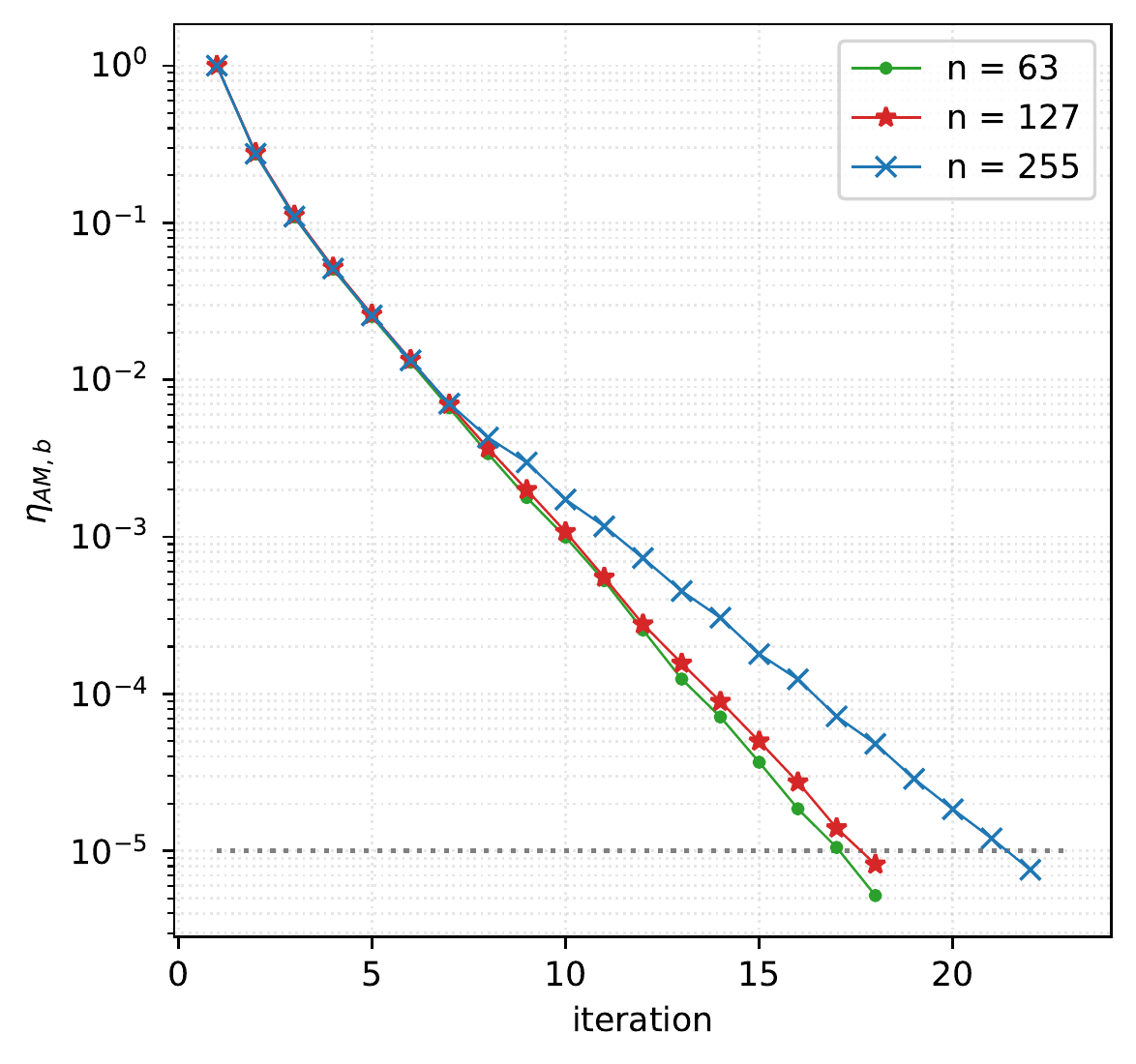}\label{fig:PCD_CH}}
        \quad
        \subfloat[Maximal TT-rank of the last Krylov vector ]{\includegraphics[scale=0.4, width=0.33\linewidth, height=0.33\linewidth]{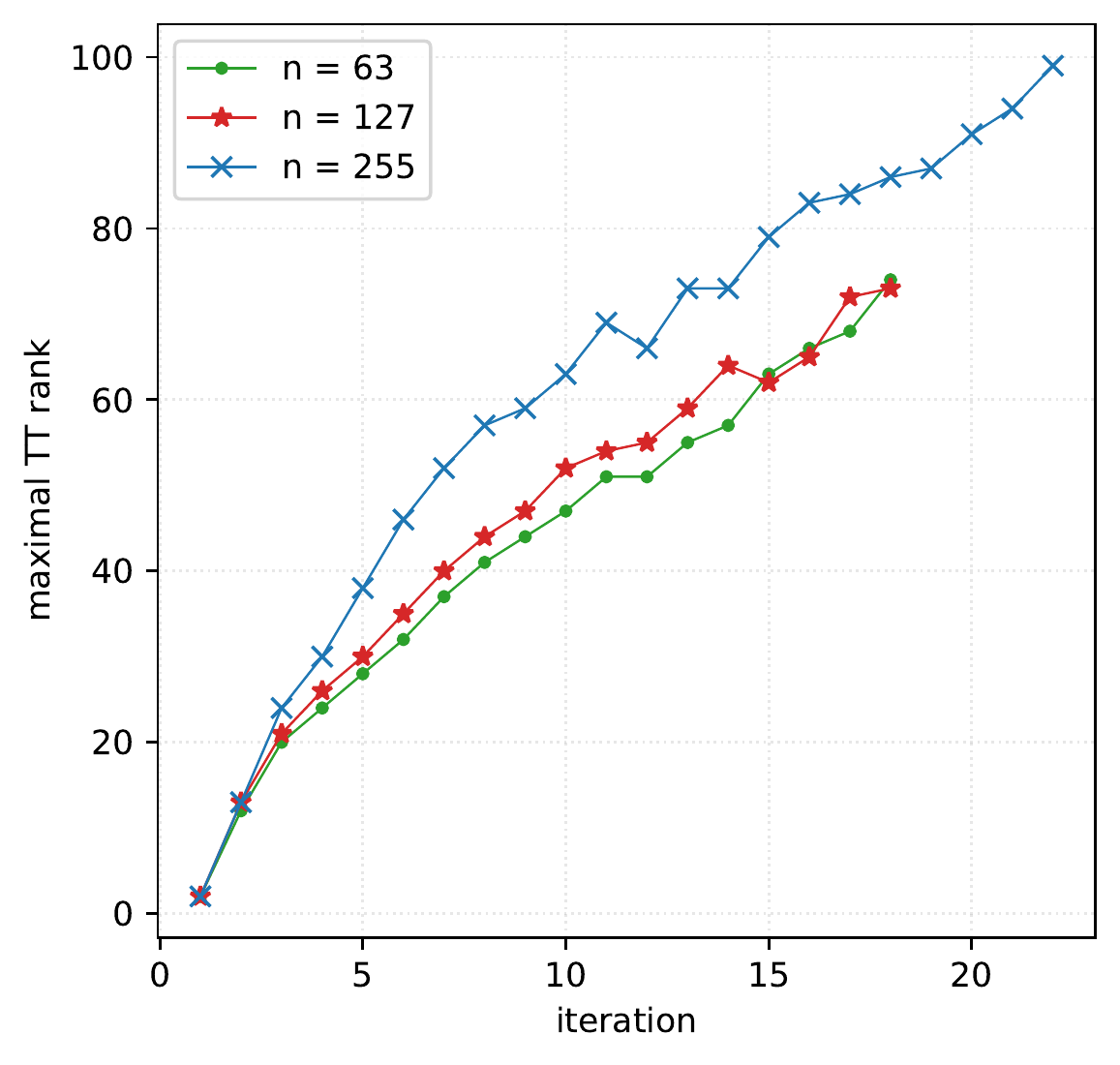}\label{fig:PCD_v-rank}}
        \vskip\baselineskip
        \subfloat[Compression ratio for the last Kyrolv vector]{\includegraphics[scale=0.45, width=0.33\linewidth, height=0.33\linewidth]{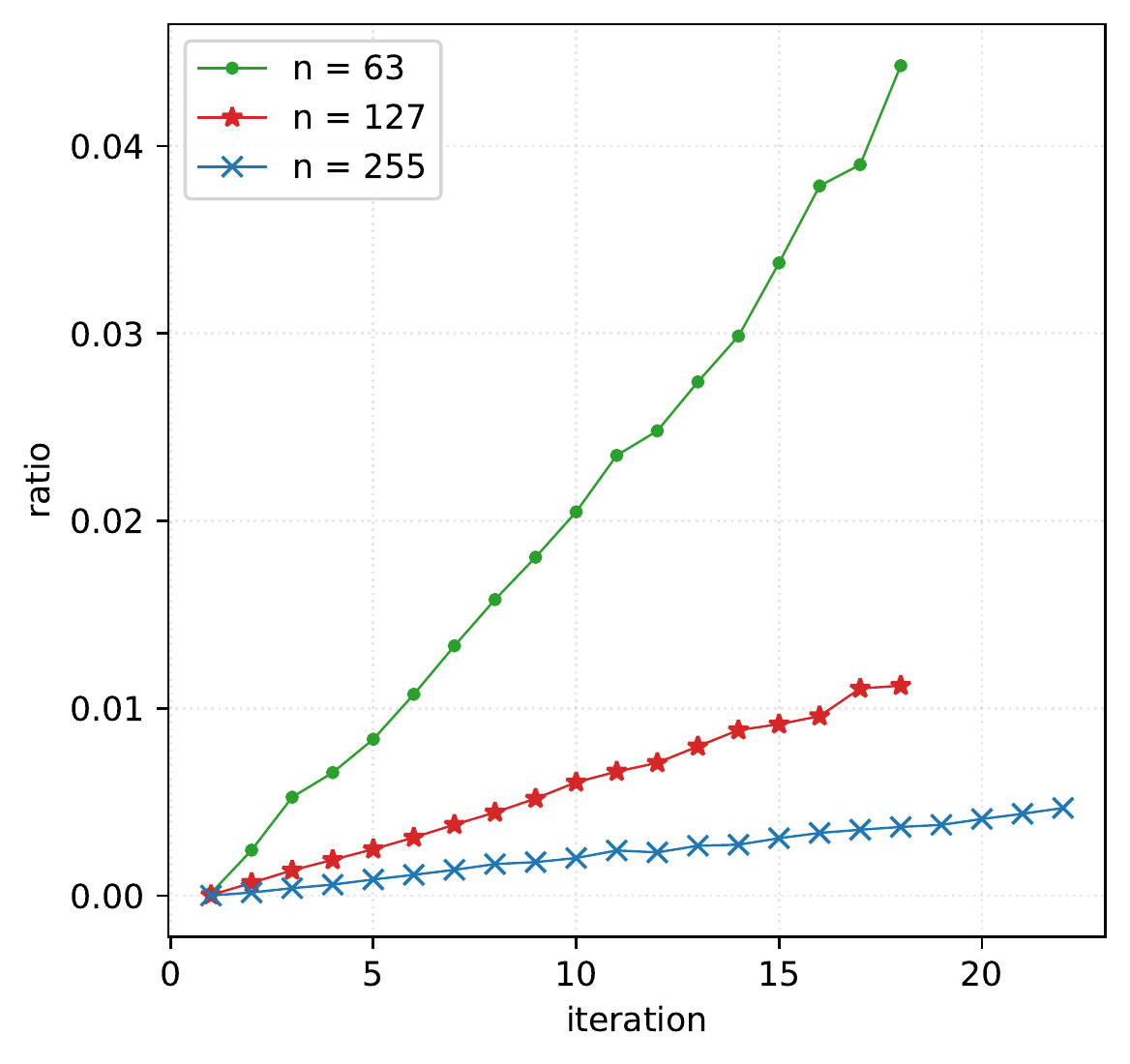}\label{fig:PCD_v-ratio}}
        \quad
        \subfloat[Compression ratio for the entire Krylov basis]{\includegraphics[scale=0.45, width=0.33\linewidth, height=0.33\linewidth]{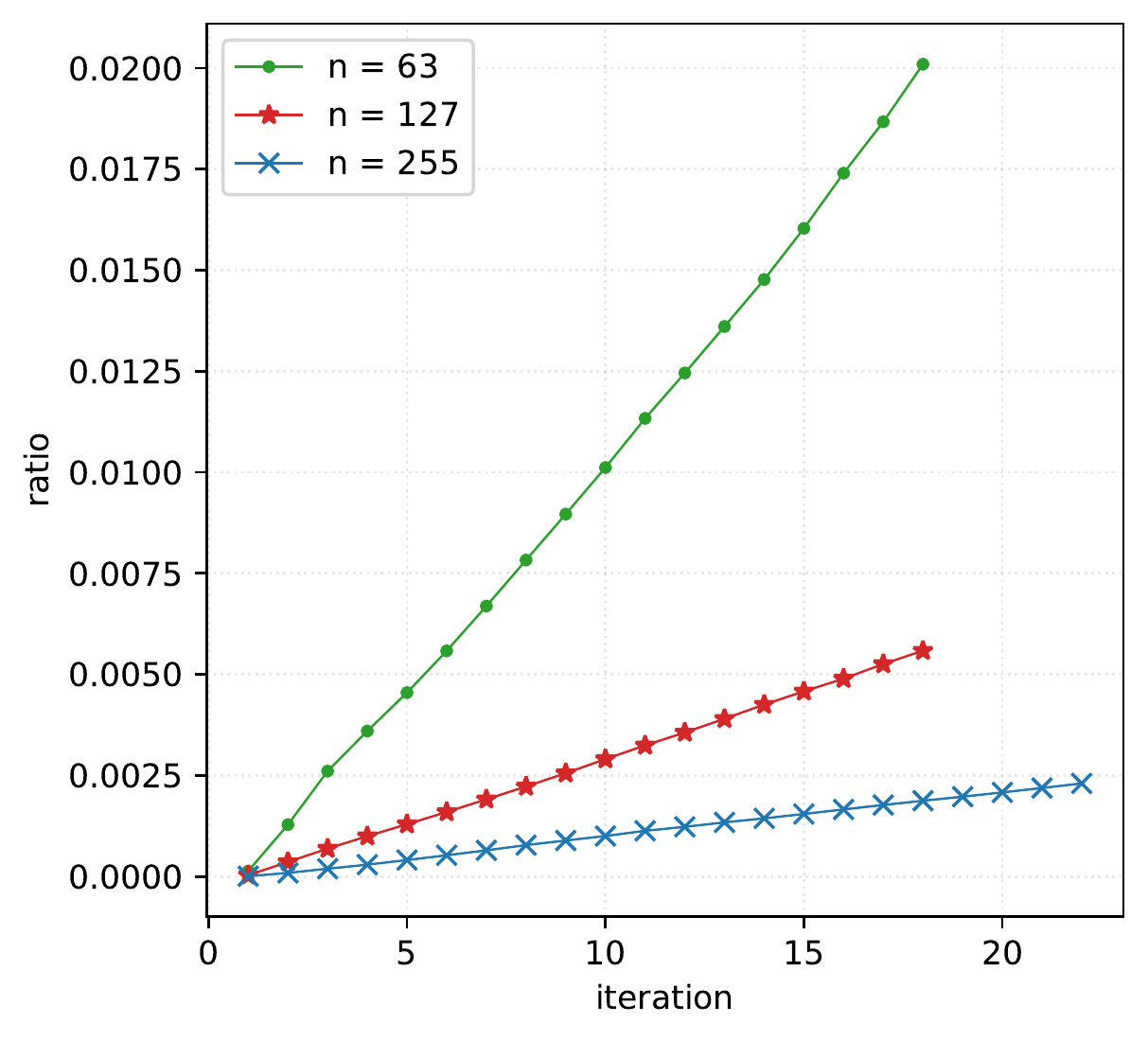}\label{fig:PCD_v-ratio_s}}
      \caption{4-d Parametric convection diffusion using $\delta=\varepsilon~=~10^{-5}$}
      \label{fig:PCD}
  	\end{figure}
 
  		\begin{figure}[!htb]
  		\centering
  		\subfloat[Convergence history in $\eta_{\ten{b}}$ for $n~=~63$]{\includegraphics[scale =0.3, width=0.31\linewidth, height=0.31\linewidth]{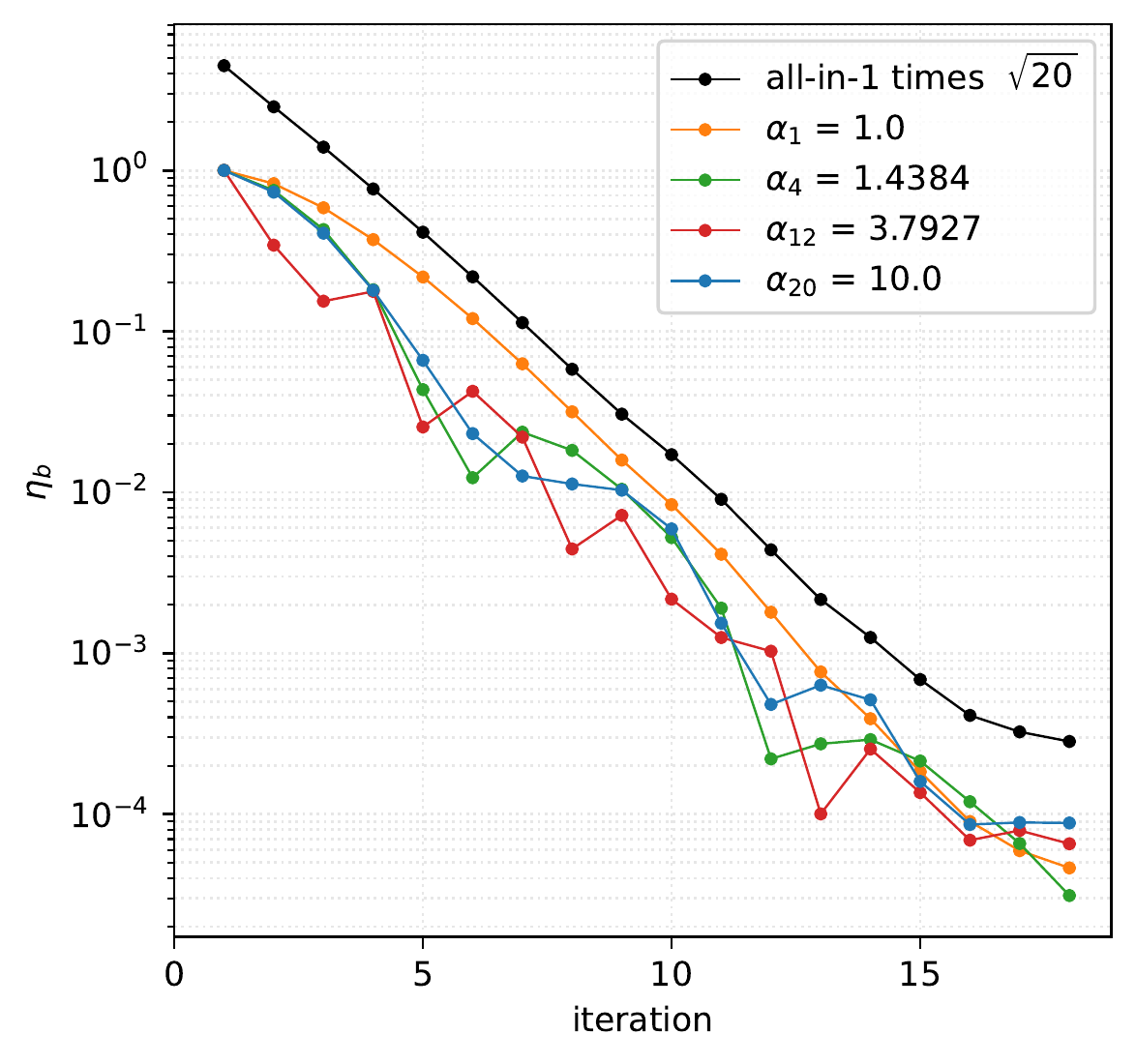}\label{fig:PCD_eta-b_CH_64}}
  		\quad
  		\subfloat[Convergence history in $\eta_{\ten{b}}$ for $n~=~127$]{\includegraphics[scale=0.3, width=0.31\linewidth, height=0.31\linewidth]{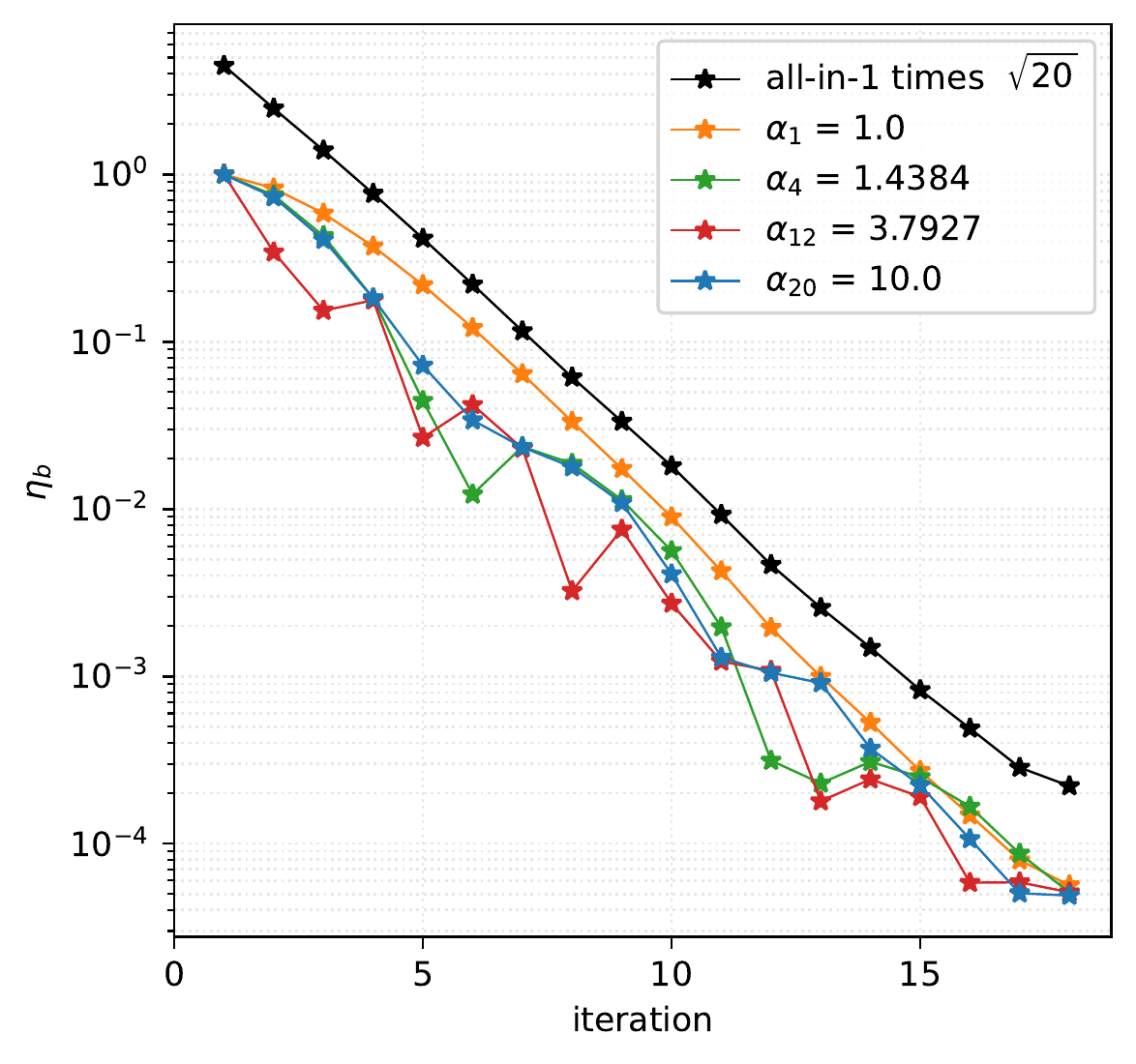}\label{fig:PCD_eta-b_CH_128}}
  		\quad
  		\subfloat[Convergence history in $\eta_{\ten{b}}$ for $n~=~255$]{\includegraphics[scale=0.3, width=0.31\linewidth, height=0.31\linewidth]{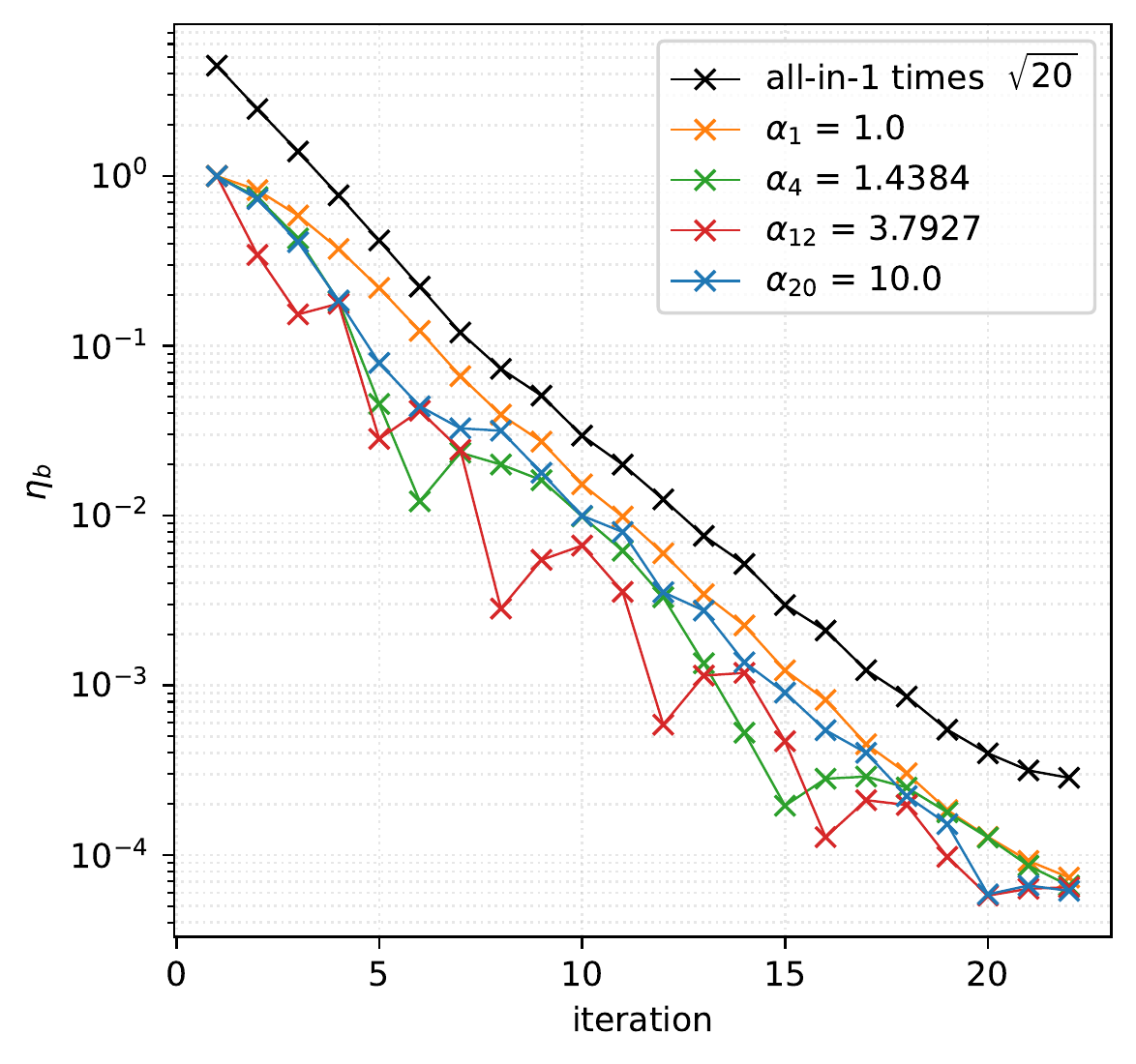}\label{fig:PCD_eta-b_CH_256}}
  		\caption{4-d Parametric convection diffusion $\eta_\ten{b}$ bound using $\delta=\varepsilon~=~10^{-5}$}
  		\label{fig:PCD_SvsM_eta_b}
  	\end{figure}
  	
  	\begin{figure}[!htb]
  		\centering
  		\subfloat[Convergence history in $\eta_{\ten{AM},\ten{b}}$ for $n~=~63$]{\includegraphics[scale =0.3, width=0.31\linewidth, height=0.31\linewidth]{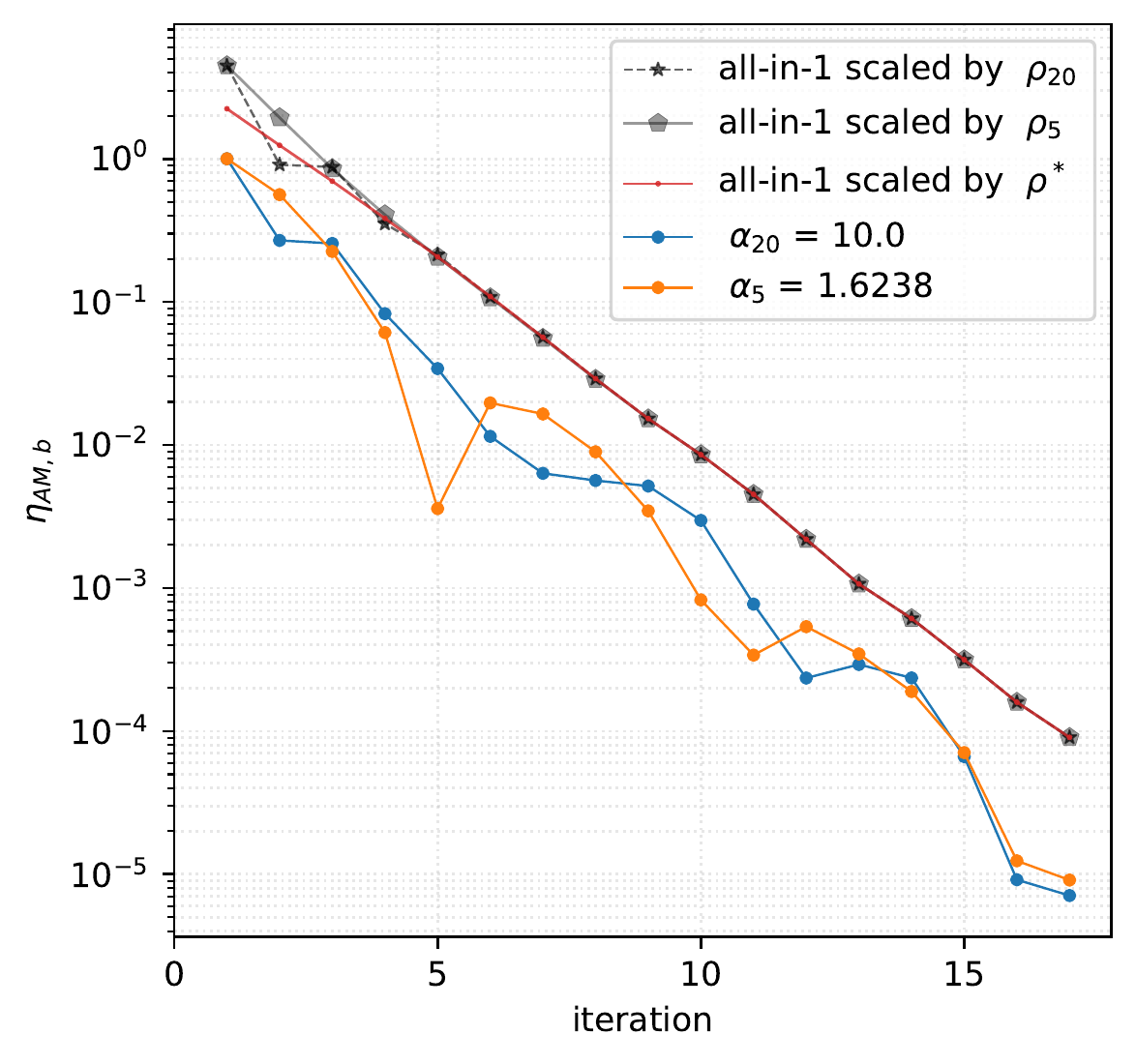}\label{fig:PCD_eta-AMb_CH_64}}
  		\quad
  		\subfloat[Convergence history in $\eta_{\ten{AM},\ten{b}}$ for $n~=~127$]{\includegraphics[scale=0.3, width=0.31\linewidth, height=0.31\linewidth]{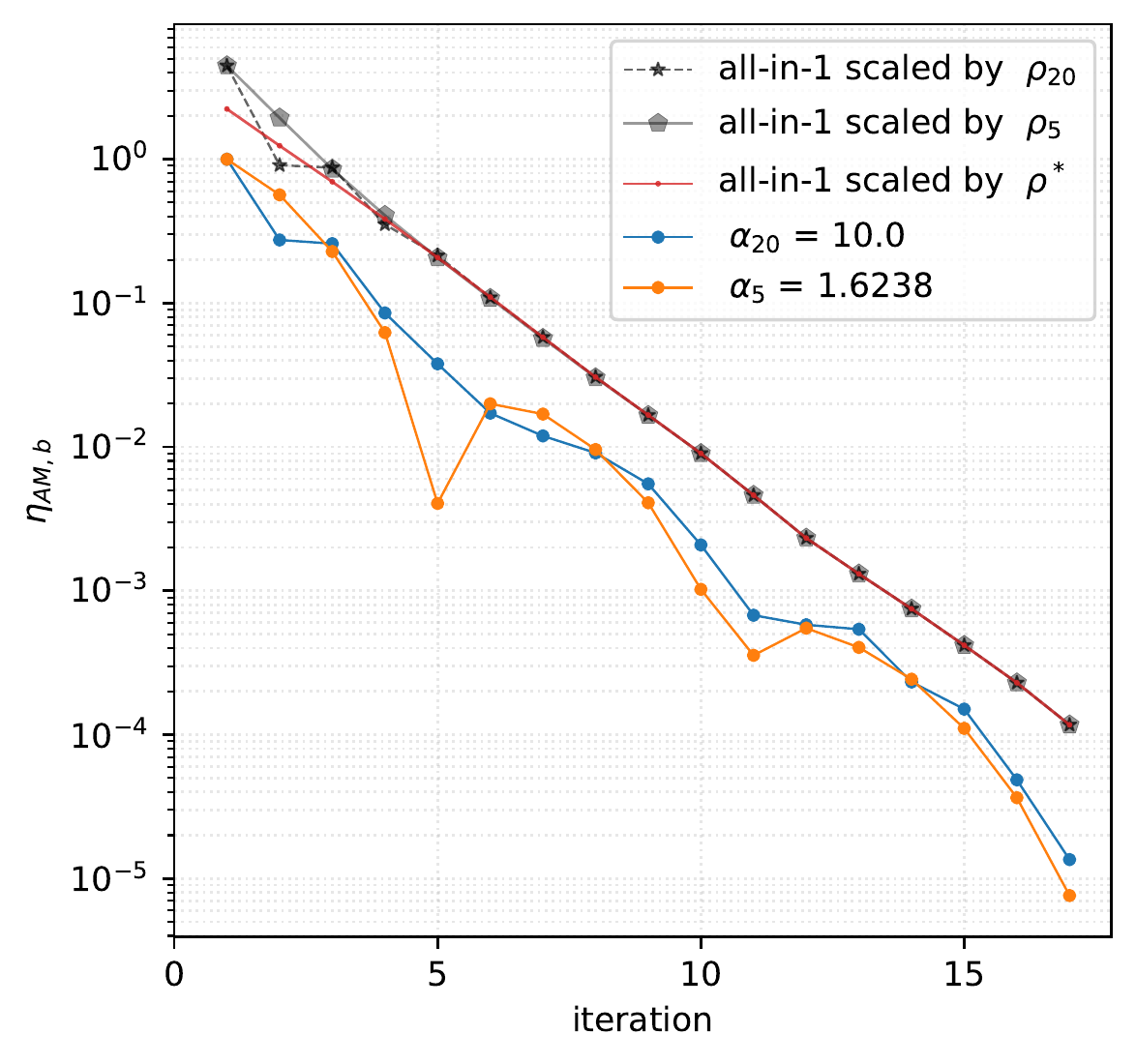}\label{fig:PCD_eta-AMb_CH_128}}
  		\quad
  		\subfloat[Convergence history in $\eta_{\ten{AM},\ten{b}}$ for $n~=~255$]{\includegraphics[scale=0.3, width=0.31\linewidth, height=0.31\linewidth]{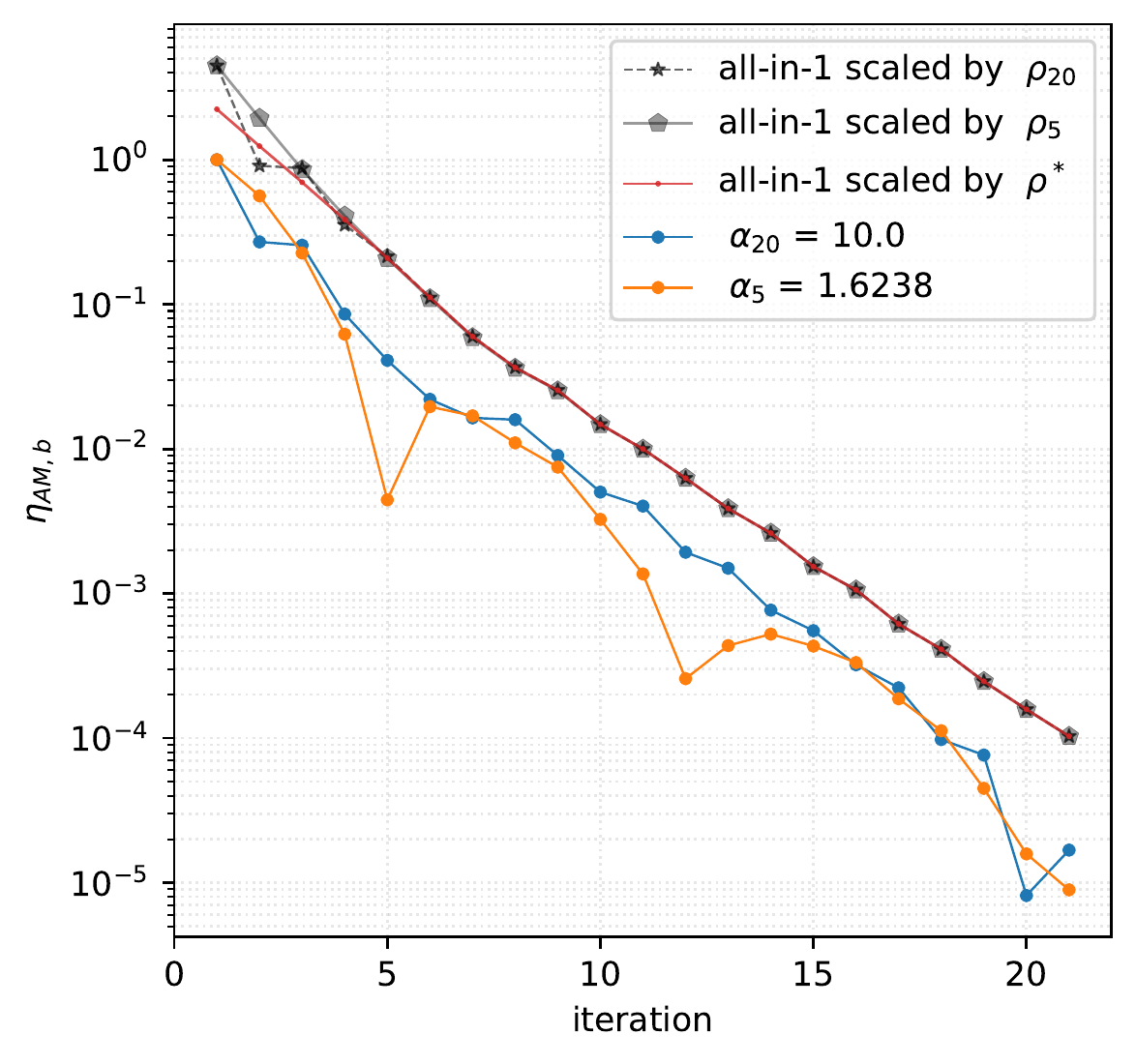}\label{fig:PCD_eta-AMb_CH_256}}
  		\caption{4-d Parametric convection diffusion $\eta_{\ten{AM},\ten{b}}$ bound using $\delta=\varepsilon = 10^{-5}$}
  		\label{fig:PCD_SvsM_eta_AM_b}
  	\end{figure}
%

  \subsubsection{Heat equation with parametrized diffusion coefficient\label{sssec:CK}}
  	We consider the heat equation with parametrized diffusion coefficient studied in~\cite{KressnerTobler2011} and defined as 
  	\begin{equation}
  	\label{eqNE:CK}
  	\begin{dcases}
  	-\nabla\cdot(\sigma_\theta(x,y,z)\nabla u(x,y,z)) &= 1 \quad\text{in}\quad \Omega = [-1,1]^3 \, ,\\
  	\hspace{125pt} u &= 0\quad \text{in}\quad\partial\Omega.
  	\end{dcases}
  	\end{equation}
  	where the coefficient $\sigma_\theta$ being a piece-wise constant function such that
  	\[
  	\sigma_\theta(x,y,z)=
  	\begin{dcases}
  	1 +\theta \quad  &\text{in}\quad[-0.5, 0.5]^3 ,\\
  	\hspace{10pt} 1 \quad  &\text{elsewhere},
  	\end{dcases}
  	\]
  	with $\theta\in [0,10]$. The function $\sigma_\theta$, rewritten as $\sigma_\theta(x, y, z)  = 1 + \theta\1_{\Xi}(x,y,z)$ where $\1_\Xi$ is the indicator function of $\Xi$, provides a linear dependency on $\theta$ for the PDE. If $\Xi_x$ is the projection of set $\Xi$ over the $x$-axis and similarly for $\Xi_y$ and $\Xi_z$, then  $\sigma_\theta(x,y,z) = 1 + \theta\1_{\Xi_x}(x)\1_{\Xi_y}(y)\1_{\Xi_z}(z)$. The problem stated in Equation~\eqref{eqNE:CK} writes equivalently 
  	\begin{equation}
  		\label{eqNE:CK1}
  		\begin{dcases}
  			-\Delta u(x,y,z) -\theta\nabla\cdot\Bigl( \1_{\Xi_x}(x)\1_{\Xi_y}(y)\1_{\Xi_z}(z)\nabla u(x,y,z)\Bigr) &= 1 \quad\text{in}\quad \Omega = [-1,1]^3 \, ,\\
  		\hspace{240pt} u &= 0\quad \text{in}\quad\partial\Omega.
  		\end{dcases}
  	\end{equation}
	After setting a grid on $n$ points along each direction on $\Omega$, the first term $\ten{B}_0$ of the operator in~\eqref{eqNE:CK1} is discretized by the $3$-d Laplacian $\ten{\Delta}_3$. For the second term $\ten{B}_1$, notice that the indicator function $\1_{\Xi}$ is trivially not differentiable on $\Xi$ boundaries. So it is approximated on the grid points, paying attention to not set them on $\partial\Xi$. The final expression of $\ten{B}_1$ is
	\[
	\ten{B}_1 = D_{x}\Delta_1 \otimes D_y \otimes D_z +  D_x \otimes D_{y}\Delta_1 \otimes D_z + D_x \otimes D_{y} \otimes D_z\Delta_1
	\]
	where $\Delta_1$ is the 1-d discrete Laplacian, $D_x = \text{diag}(\1_{\Xi_{x_i}})\in\R^{n\times n}$ and similarly for $D_y$ and $D_x$. Remark that $\ten{B}_1$ is a Laplacian-like operator, which is expressed in TT-format according to Equation~\eqref{eqTT:2} and~\eqref{eqTT:2a}. 
	The final discrete TT-operator of Problem~\eqref{eqNE:CK} is 
  	\[
  	\ten{A}_\theta = \ten{B}_0 + \theta\ten{B}_1. 
  	\]
  	The right-hand side is  $\ten{c}\in\R^{n\times n\times n}$ such that $\ten{c}(i_1, i_2, i_3) = 1$ for $i_k\in\{1, \dots, n\}$ for $k\in\{1,2,3\}$. To study the quality of the bounds expressed in Proposition~\ref{prop:eta_b} and~\ref{prop:eta_Ab}, the tensor $\ten{c}$ is normalized, i.e., it is scaled by $1/n^3$.
  	Since we want to solve for $p$ values of $\theta$ in $[0,10]$ simultaneously, i.e., we want to solve $p$-times the discrete Problem~\eqref{eqNE:CK} for different values of $\theta$, we tensorize $\ten{B}_0$ and $\ten{B}_1$ by a diagonal matrices, adding a fourth dimension. The tensor discrete operator $\ten{A}\in\R^{(p\times p)\times (n\times n)\times (n\times n)\times (n\times n)}$ of the ``all-in-one'' problem writes
  	\[
  	\ten{A} = \I_{p} \otimes \ten{B}_0 + \Theta\otimes \ten{B}_1
  	\]
  	where $\Theta = \text{diag}(\theta_1, \dots, \theta_p)$ for $\theta_i\in[0,10]$ uniformly distributed for $i\in\{1,\dots, p\}$. The right-hand side of the ``all-in-one'' problem is 
  	\[
  	\ten{b} = \1_p \otimes \ten{c}
  	\]
  	with $\1_p\in\R^{p}$ a vector of ones. Remark that since $\norm{\ten{c}} = 1$ by construction, then $\norm{\ten{b}} = \sqrt{p}$.
	We perform experiments with full TT-GMRES (i.e., no restart) for $n\in\{63, 127\}$ and $p = 20$, with the preconditioner defined in Equation~\eqref{eq:prec_d+1} with $q \in\{16,32\}$. Figure~\ref{fig:CK_CH} shows that TT-GMRES converges to the prescribed tolerance in approximately $20$ iterations. From the point of view of the memory consumption, in Figure~$\ref{fig:CK_v-rank}$ we see that for $n = 63$ the maximum TT-rank is lower $200$, while for $n = 127$ it is lower than $250$. In terms of memory saving, Figure~\ref{fig:CK_v-ratio} shows that in the worst case we are using only $10\%$ and less than $5\%$ of the memory necessary to store one full tensor of the Krylov basis and the entire full basis respectively.
	
In Figure~\ref{fig:CK_SvsM_eta_b} we have the relation of $\eta_{\ten{b}}$ and $\eta_{\ten{b}_\ell}$ for $\ell\in\{1,\dots, p\}$. All the curves present the same shape, with the one associated with $\theta_1 = 0$ being the most peculiar one. We see that in the optimal case the distance between the ``all-in-one'' curve and the individual ones is lower than one order of magnitude, while in the worst case, realized by $\theta_1 = 0$, the difference is approximately almost of two orders. A similar argument holds for $\eta_{\ten{AM}, \ten{b}}$ bound. As in Section~\ref{sssec:PCD}, we compute $\ell_m$ and $\ell_M$, as defined in Equation~\eqref{eq3:mM}, which are equal to $\ell_m=20$ and $\ell_M=1$ respectively. In Figure~\ref{fig:CK_SvsM_eta_AM_b} we see that the two curves $\eta_{\ten{A}_\ell\overline{\ten{M}}, \ten{b}_\ell}$ have a starting and ending overlapping part, while in the internal part they differ by less than one order of magnitude. The three scaled curves for $\eta_{\ten{A}\ten{M}, \ten{b}}$ overlap from the third iteration. As in the previous studied case, $\rho^*$ from Corollary~\ref{cor:eta_Ab} provides a good approximation of the scaling coefficient. In the optimal case the distance is of one order of magnitude approximately, while in the worst one a little more than one order. 
  	\begin{figure}[!htb]
  		\centering
  		\subfloat[Convergence history]{\includegraphics[scale =0.45, width=0.33\linewidth, height=0.33\linewidth]{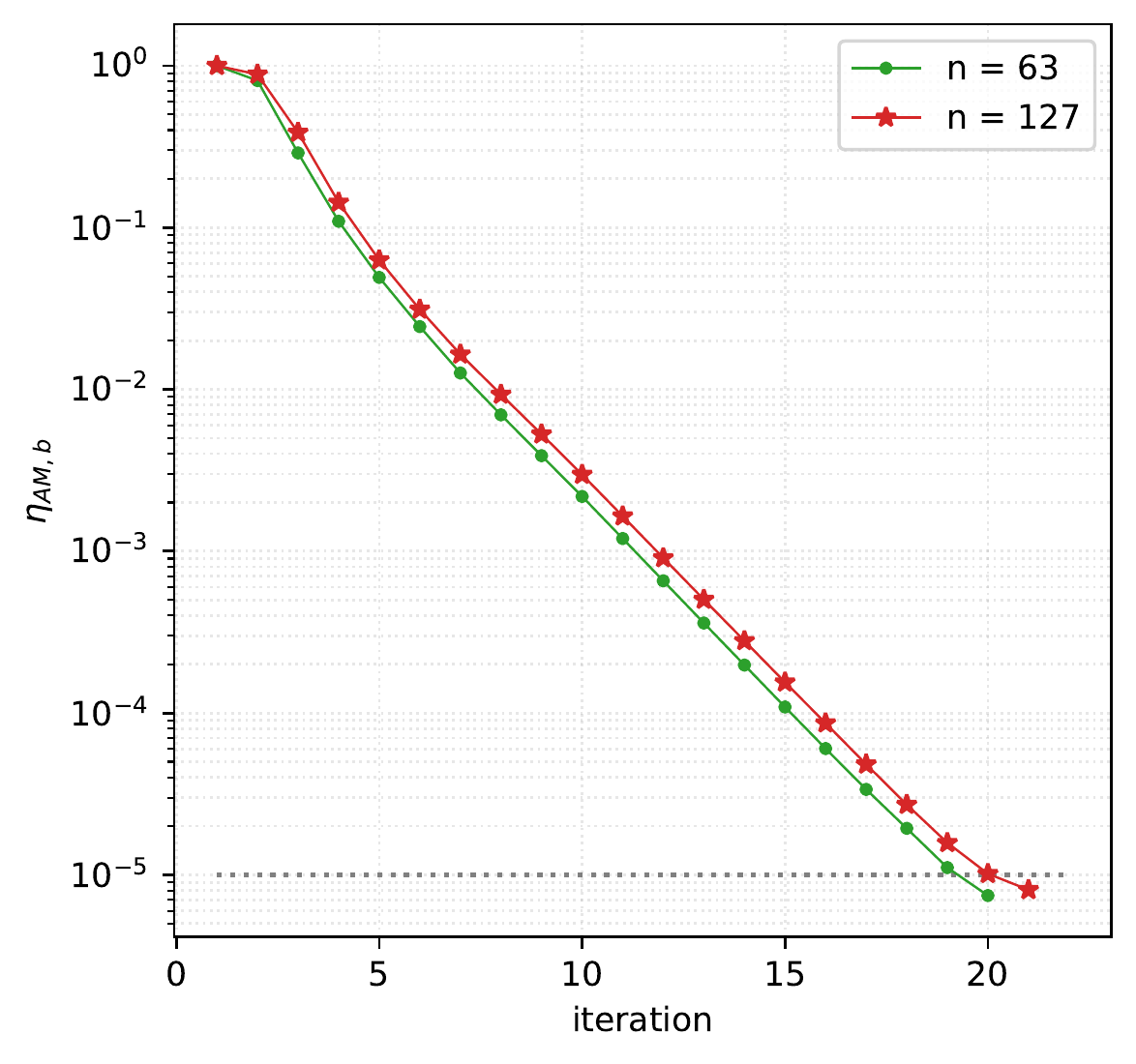}\label{fig:CK_CH}}
  		\quad
  		\subfloat[Maximal TT-rank of the last Krylov vector ]{\includegraphics[scale=0.45,width=0.33\linewidth, height=0.33\linewidth]{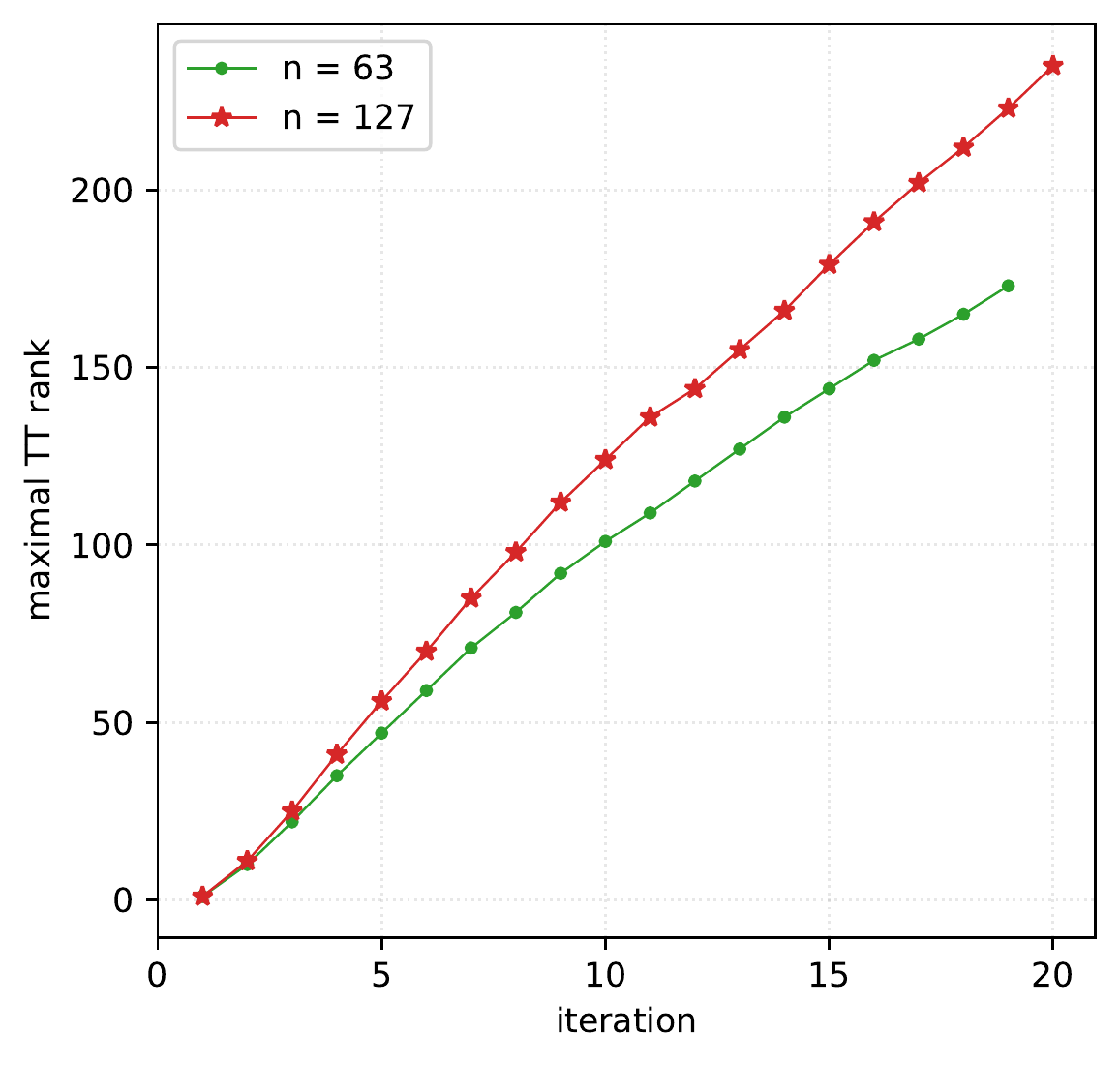}\label{fig:CK_v-rank}}
  		\vskip\baselineskip
  		\subfloat[Compression ration for the last Kyrolv vector]{\includegraphics[scale=0.45, width=0.33\linewidth, height=0.33\linewidth]{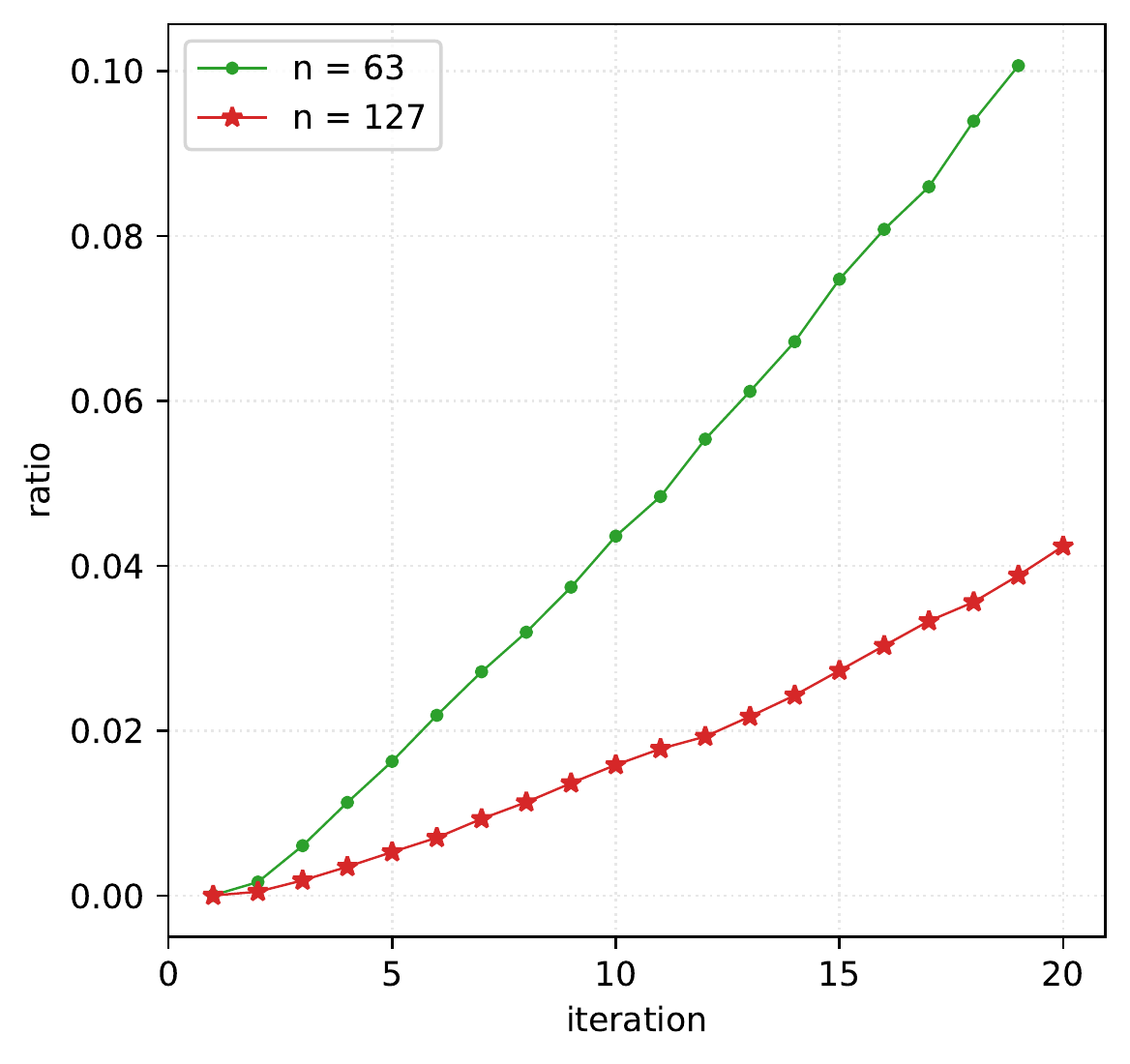}\label{fig:CK_v-ratio}}
  		\quad
  		\subfloat[Compression ratio for the entire Krylov basis]{\includegraphics[scale=0.45, width=0.33\linewidth, height=0.33\linewidth]{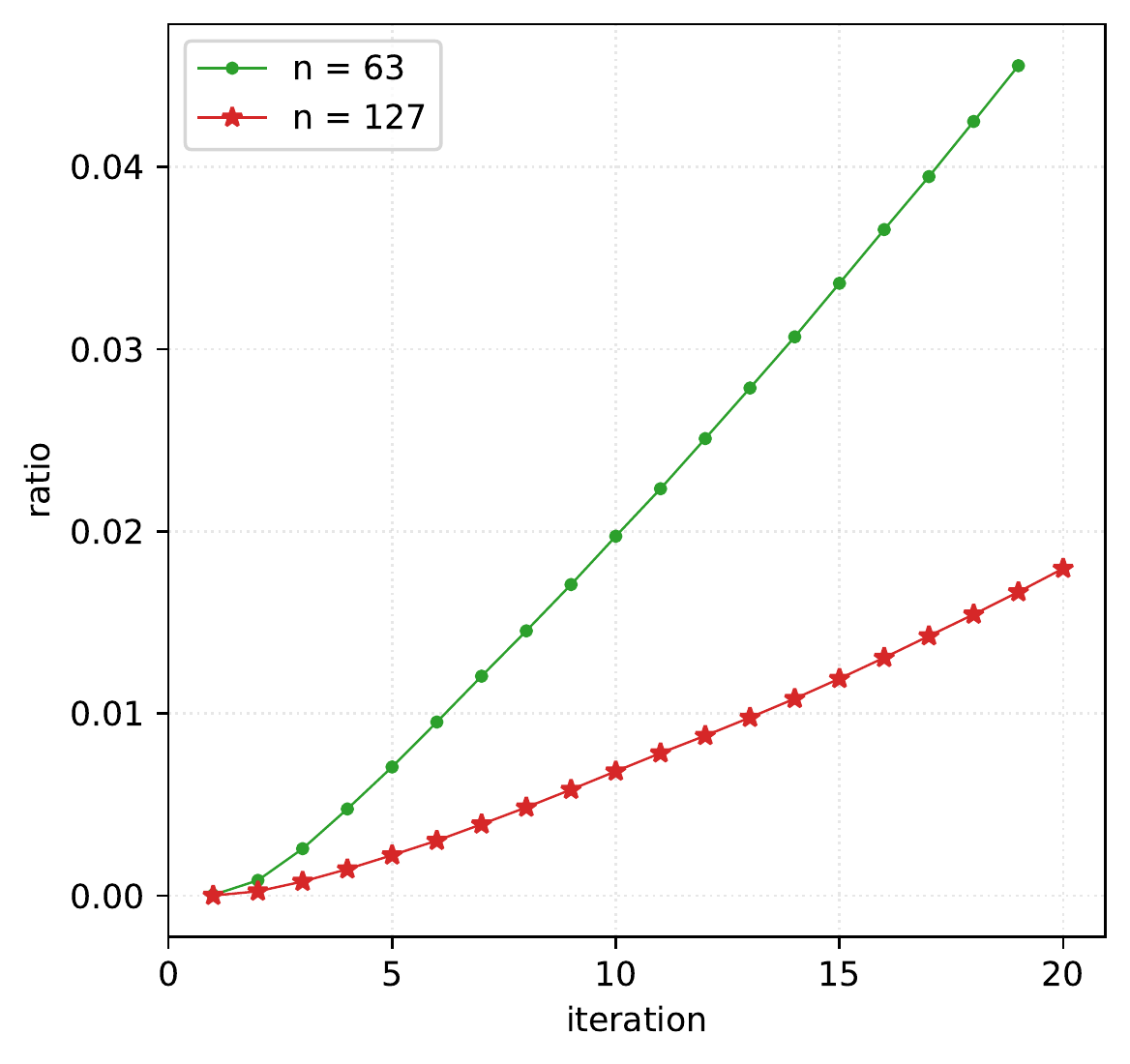}\label{fig:CK_v-ratio_s}}
  		\caption{4-d Heterogeneous convection diffusion using $\delta = \varepsilon = 10^{-5}$}
  		\label{fig:CK}
  	\end{figure}
  
  	\begin{figure}[!htb]
  		\centering
  		\subfloat[Convergence history in $\eta_{\ten{b}}$ for $n~=~63$]{\includegraphics[scale =0.45, width=0.33\linewidth, height=0.33\linewidth]{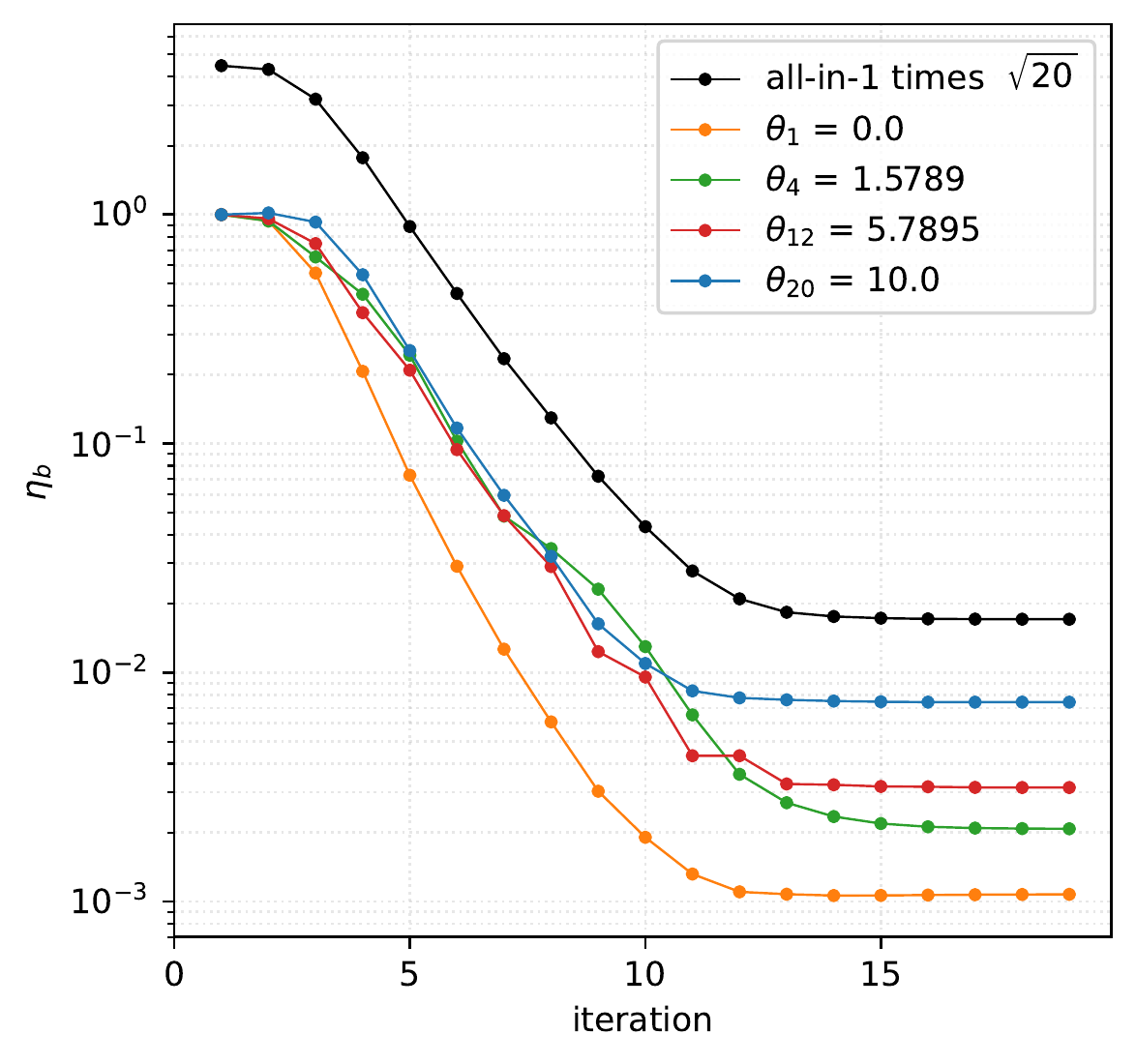}\label{fig:PCD_eta-b_CK_64}}
  		\quad
  		\subfloat[Convergence history in $\eta_{\ten{b}}$ for $n~=~127$]{\includegraphics[scale=0.45, width=0.33\linewidth, height=0.33\linewidth]{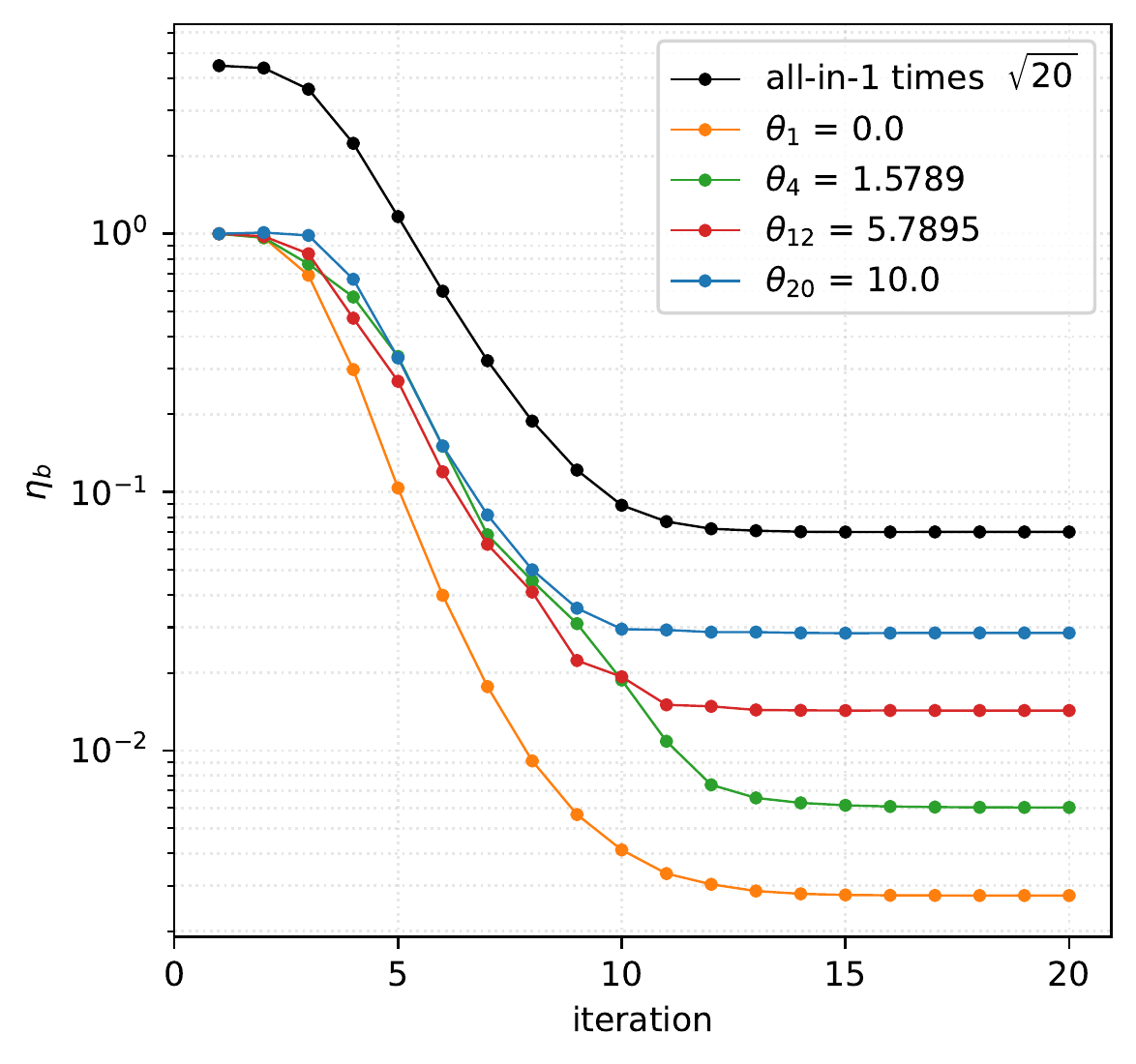}\label{fig:PCK_eta-b_CH_128}}
  	\caption{4-d Heterogeneous convection diffusion $\eta_{\ten{b}}$ bound using $\delta = \varepsilon = 10^{-5}$} 
  	\label{fig:CK_SvsM_eta_b}
  	\end{figure}
  	
  	\begin{figure}[!htb]
  		\centering
  		\subfloat[Convergence history in $\eta_{\ten{AM}, \ten{b}}$ for $n = 63$]{\includegraphics[scale =0.45, width=0.33\linewidth, height=0.33\linewidth]{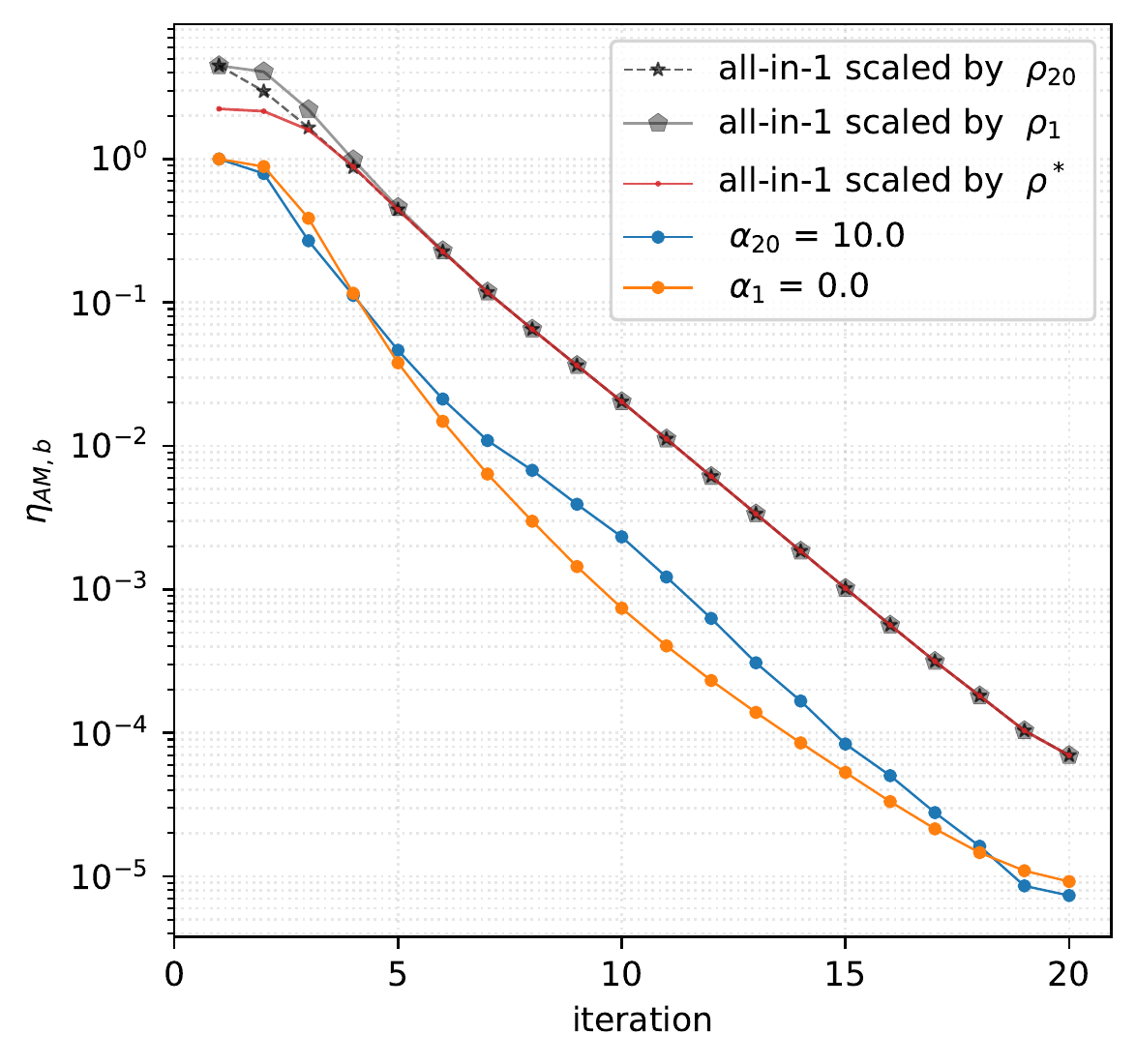}\label{fig:PCD_eta-AMb_CK_64}}
  		\quad
  		\subfloat[Convergence history in $\eta_{\ten{AM},\ten{b}}$ for $n = 127$]{\includegraphics[scale=0.45, width=0.33\linewidth, height=0.33\linewidth]{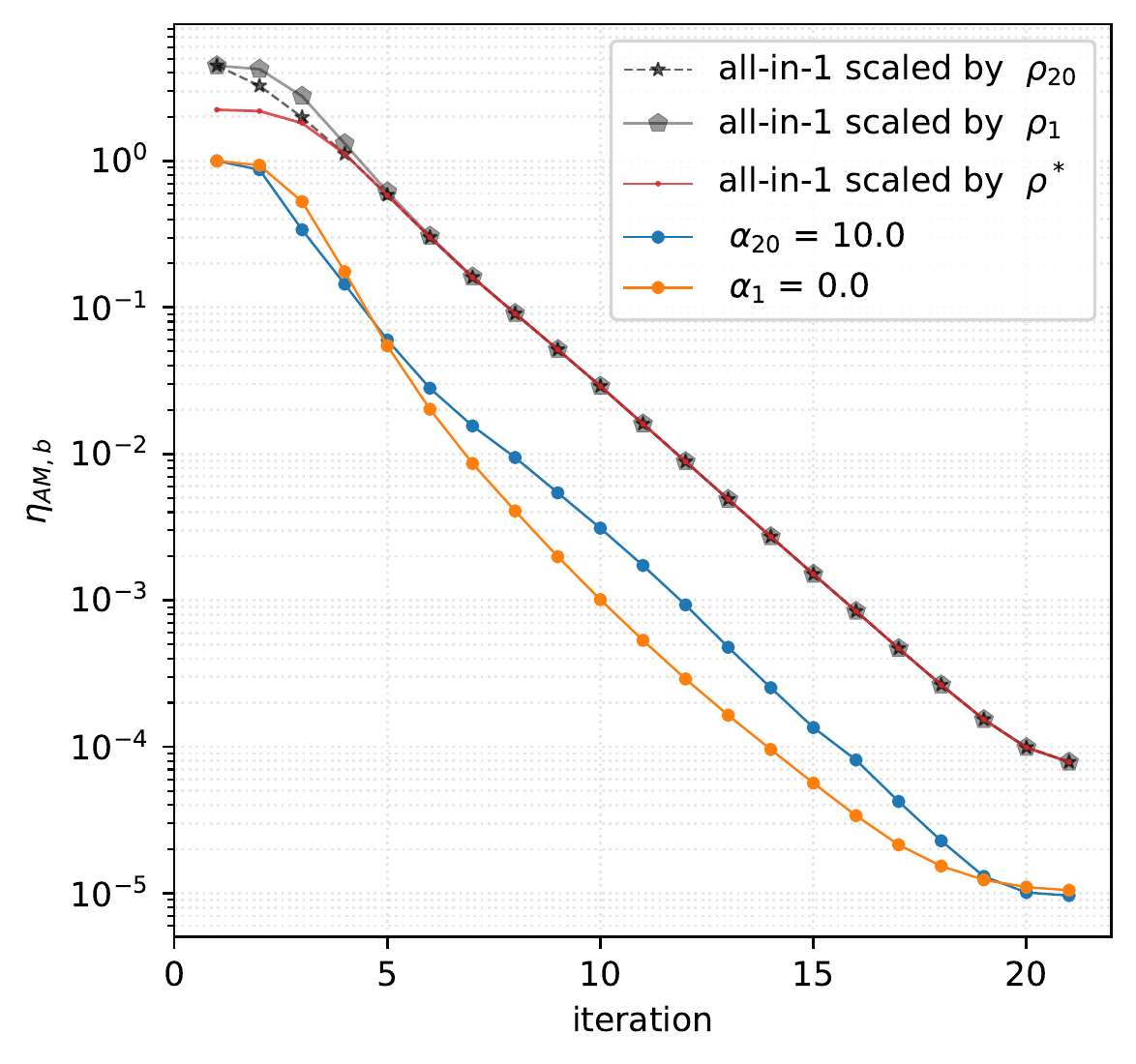}\label{fig:PCK_eta-AMb_CH_128}}
  		\caption{4-d Heterogeneous convection diffusion $\eta_{\ten{AM},\ten{b}}$ bound using $\delta = \varepsilon = 10^{-5}$} 
  	\label{fig:CK_SvsM_eta_AM_b}
  	\end{figure}
%
  	\subsection{Solution of parameter dependent right-hand sides}
  	The aim of this section is to investigate the numerical properties of some examples in the context of the multiple right-hand side solution, following the tensorized approach described in~\ref{ssec:Ai1}.
  	\subsubsection{Poisson problem\label{sec4:mLap}} In this subsection we solve simultaneously multiple Poisson problems stated in Equation~\eqref{eqNE:Lap} with modified right-hand sides. Let $-\ten{\Delta}_3$ be the discretization of the Laplacian over a Cartesian grid of $n$ points per mode for the domain $\Omega = [0, 1]^3$. Let $\ten{b}\in\R^{n\times n\times n}$ be the right-hand side discretization defined in Section~\ref{sec3:Lap}. We define the individual linear system as
  	\[
  	-\ten{\Delta}_3 \ten{u}_\ell = \ten{b} + \ten{e}^{[\ell]}
  	\]
  	where $\ten{e}^{[\ell]}\in\R^{n\times n \times n}$ is the $\ell$-th slice with respect to the first mode of $\ten{e} \in \R^{p\times n\times n \times n}$ a realization of the normal distribution $\mathcal{N}(0,1)$. Since the aim is solving simultaneously the $p$ problems, as in Section~\ref{ssec:Ai1}, we define the ``all-in-one'' tensor linear operator $\ten{A}\in\R^{(p\times p)\times(n\times n)\times(n\times n)\times(n\times n)}$
  	\[
  		\ten{A} = \I_p \otimes (\ten{-\Delta}_3)
  	\]
  	while the ``all-in-one'' right-hand side is $\ten{c}\in\R^{p\times n\times n\times n}$ such that
  	\[
  		\ten{c} = \1_p\otimes\ten{b} + \ten{e}.
  	\]
  	We consider the solution of the problem with $n\in\{63, 127, 255\}$ and $p = 20$. To speed up the convergence we introduce the preconditioner defined in~\eqref{eq:prec_d+1} with $q\in\{16, 32\}$. Notice that theoretically the TT-rank of $\ten{c}$ may become extremely large, 
	leading to a memory over-consumption and higher computational costs. To face this drawback, we impose a small TT-rank to $\ten{e}^{[\ell]}$, so that the TT-rank of $\ten{c}$ ends up being $11$ at maximum. To study the bounds stated in Section~\ref{ssec:Ai1}, 
  	we need to comply with the hypothesis so that we scale each individual right-hand side by its norm, so that $\norm{\ten{c}} =\sqrt{p}$. 
  	
  	\begin{figure}[!htb]
  		\centering
  		\subfloat[Convergence history]{\includegraphics[scale =0.45, width=0.33\linewidth, height=0.33\linewidth]{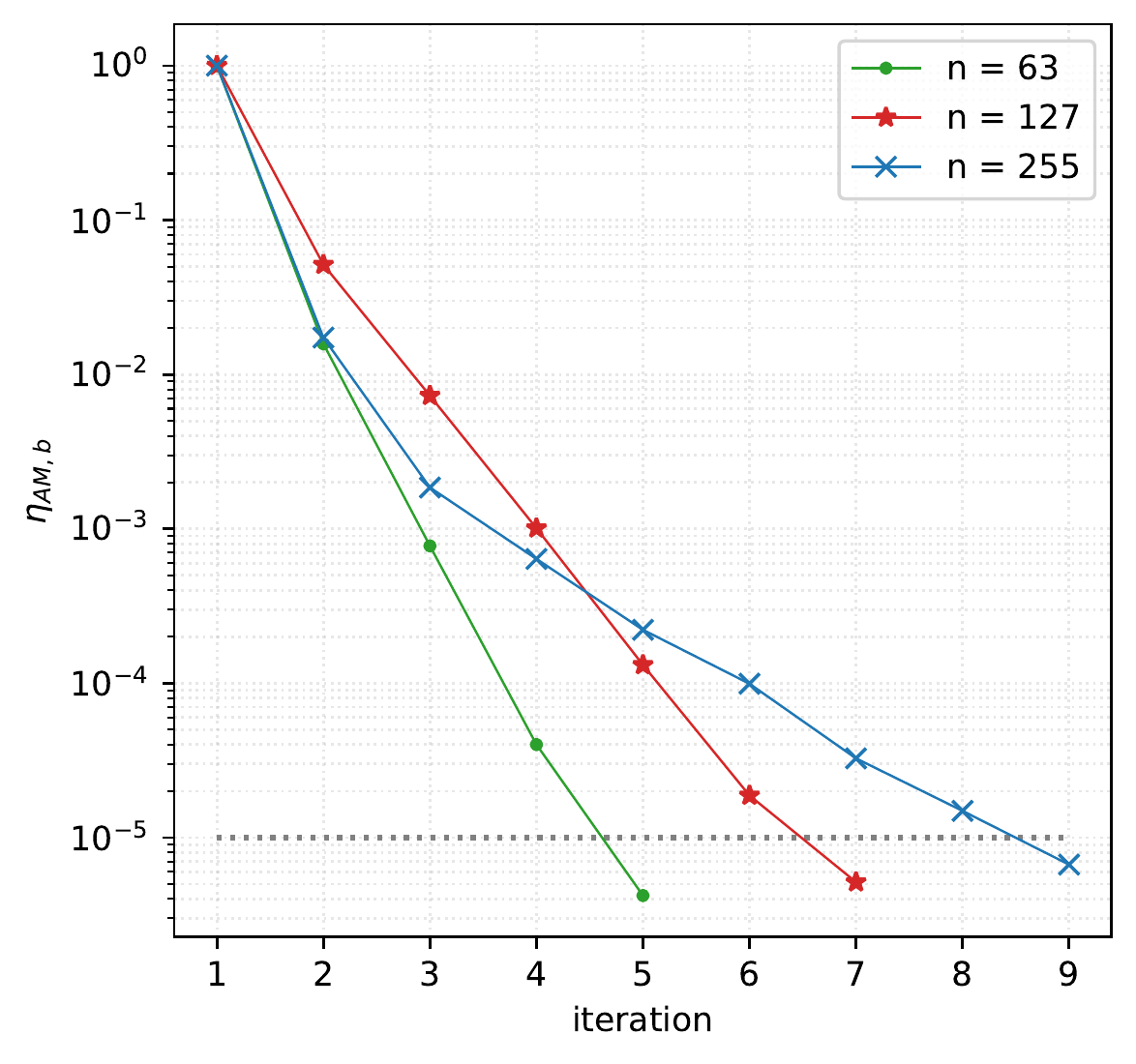}\label{fig:mLap_CH}}
  		\quad
  		\subfloat[Maximal TT-rank of the last Krylov vector ]{\includegraphics[scale=0.45,width=0.33\linewidth, height=0.33\linewidth]{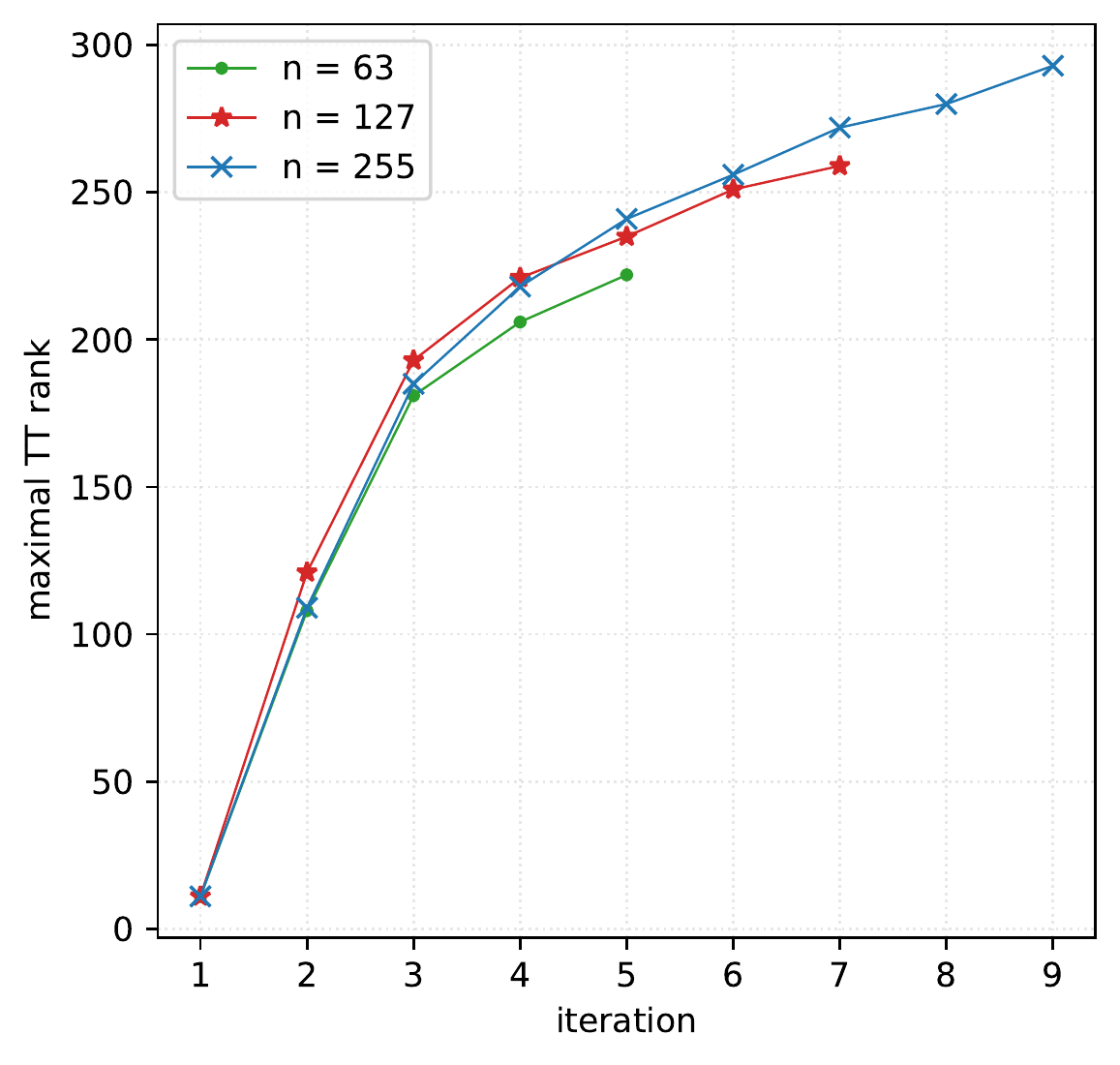}\label{fig:mLap_v-rank}}
  		\vskip\baselineskip
  		\subfloat[Compression ration for the last Kyrolv vector]{\includegraphics[scale=0.45, width=0.33\linewidth, height=0.33\linewidth]{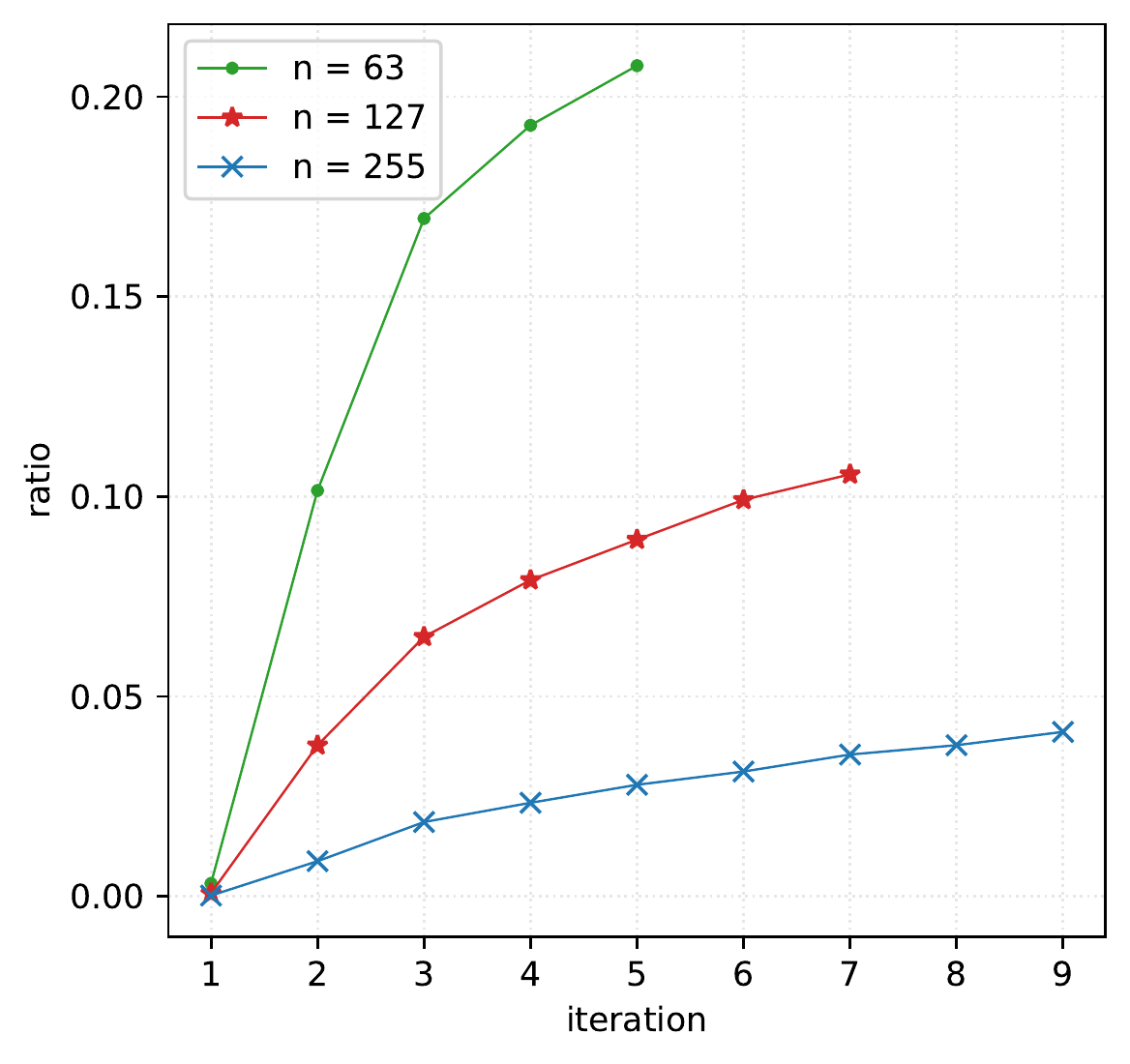}\label{fig:mLap_v-ratio}}
  		\quad
  		\subfloat[Compression ratio for the entire Krylov basis]{\includegraphics[scale=0.45, width=0.33\linewidth, height=0.33\linewidth]{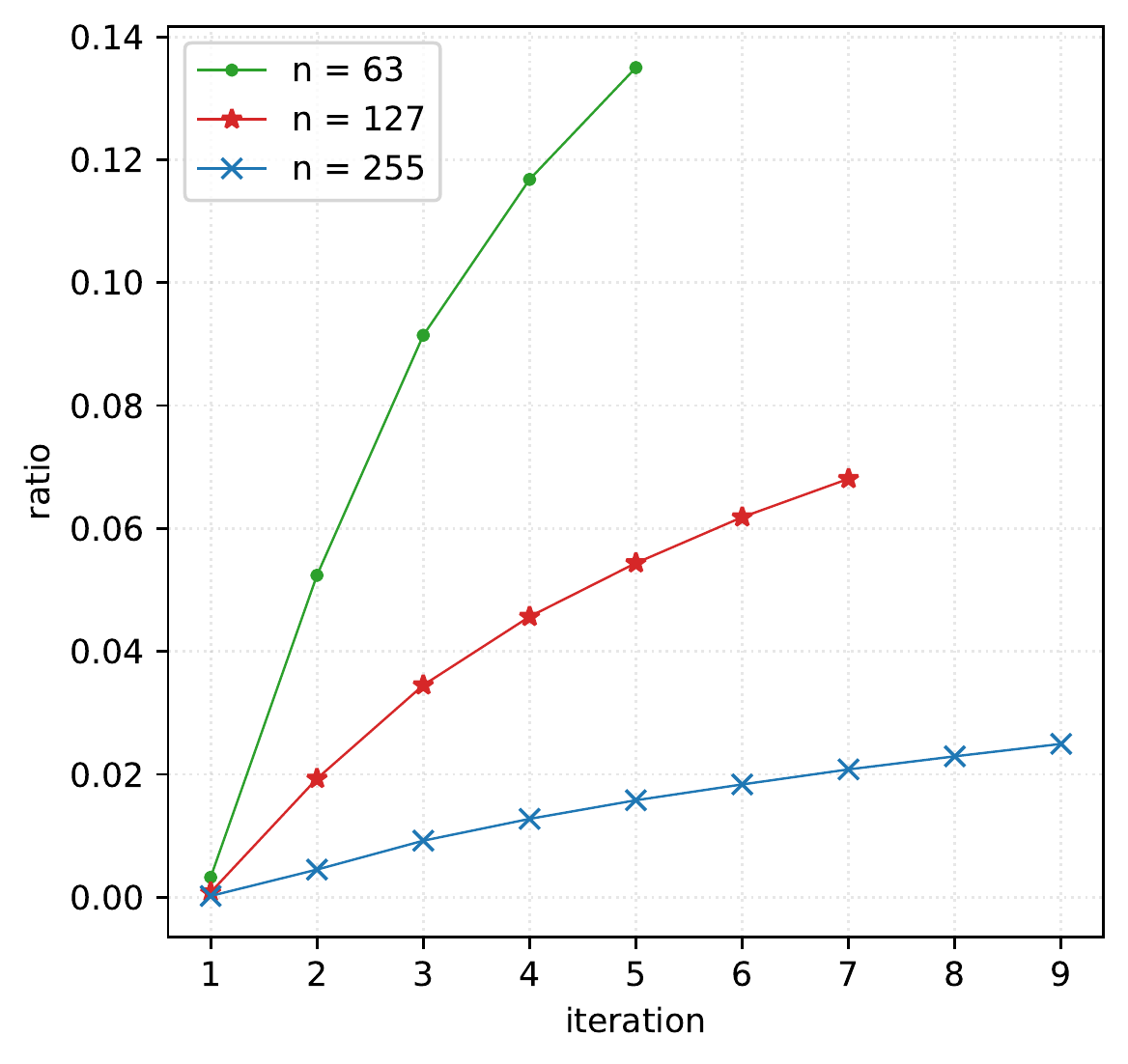}\label{fig:mLap_v-ratio_s}}
  		\caption{4-d multiple right-hand side Poisson problem using $\delta = \varepsilon = 10^{-5}$}
  		\label{fig:mLap}
  	\end{figure}
  	As we can see in Figure~\ref{fig:mLap_CH}, TT-GMRES converges in $5$ iterations for $n = 63$, in $7$ for $n = 127$ and in $9$ for $n = 255$. Figure~\ref{fig:mLap_v-rank} shows that the TT-rank of the last Krylov vector becomes quickly large, with maximum values ranging from $200$ to $300$. However looking at Figures~\ref{fig:mLap_v-ratio} and~\ref{fig:mLap_v-ratio_s}, the compression ratio for a single basis vector and for the entire basis remains extremely small, from $0.05$ to $0.2$ for the first one and from $0.02$ and $0.14$ for the entire basis, meaning that the TT approach is still effective from the memory point of view. As in the parametric operator case, we study the bounds expressed in Propositions~\ref{prop:eta_b} and~\ref{prop:eta_Ab_m}. In Figure~\ref{fig:mLap_SvsM_eta_b}, we see that the bound for $\eta_{\ten{b}}$ is always quite tight, around $1$ order of magnitude approximately. To use the result of Proposition~\ref{prop:eta_Ab_m}, we set $w$ equal to the number of iterations to converge and for every $\ell\in\{1,\dots, p\}$, we define the vector $\gamma_\ell\in\R^w$ such that
  	\[
  	\gamma_\ell(i) = \psi_\ell(\ten{t}_k)\quad\text{for every}\qquad k\in\{1,\dots, w\}.
  	\]
  	We define $\ell_m$ and $\ell_M$ as the indexes which realize the minimum and the maximum of $\gamma_\ell$ norm, i.e.,
  	\begin{equation}
  	\label{eq3:mrhsM}
  	\ell_m = \argmin_{\ell\in\{1,\dots,p\}} \norm{\gamma_\ell}\quad\text{and}\quad\ell_M = \argmax_{\ell\in\{1,\dots,p\}} \norm{\gamma_\ell} .
  	\end{equation}
  	In this specific case for each grid point step, the value of $\ell_m$ and $\ell_M$ is reported in Figure~\ref{fig:mLap_SvsM_eta_AM_b}. The same Figure~shows that the bound in this specific case is quite good, with approximately less of $1$ order of magnitude of difference, in the optimal and in the worst case. Moreover the three scaled ``all-in-one'' curves overlap from the second iteration, suggesting again that $\rho^*$ from Corollary~\ref{cor:eta_Ab} is a good approximation of the scaling factors.  
  	\begin{figure}[!htb]
  		\centering
  		\subfloat[Convergence history in $\eta_{\ten{b}}$ for $n = 63$]{\includegraphics[scale =0.3, width=0.31\linewidth, height=0.31\linewidth]{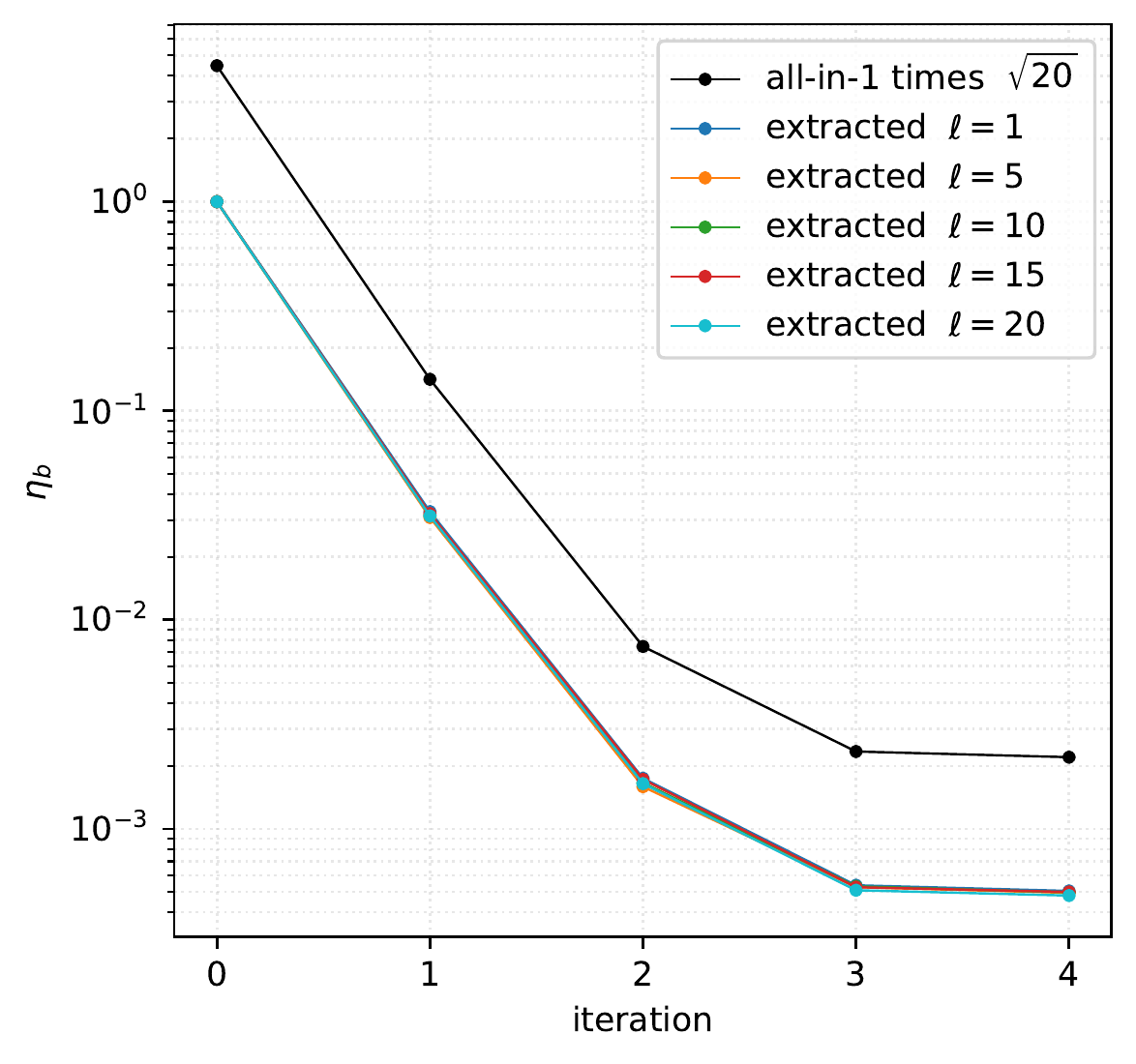}\label{fig:mLap_eta-b_64}}
  		\quad
  		\subfloat[Convergence history in $\eta_{\ten{b}}$ for $n = 127$]{\includegraphics[scale=0.3, width=0.31\linewidth, height=0.31\linewidth]{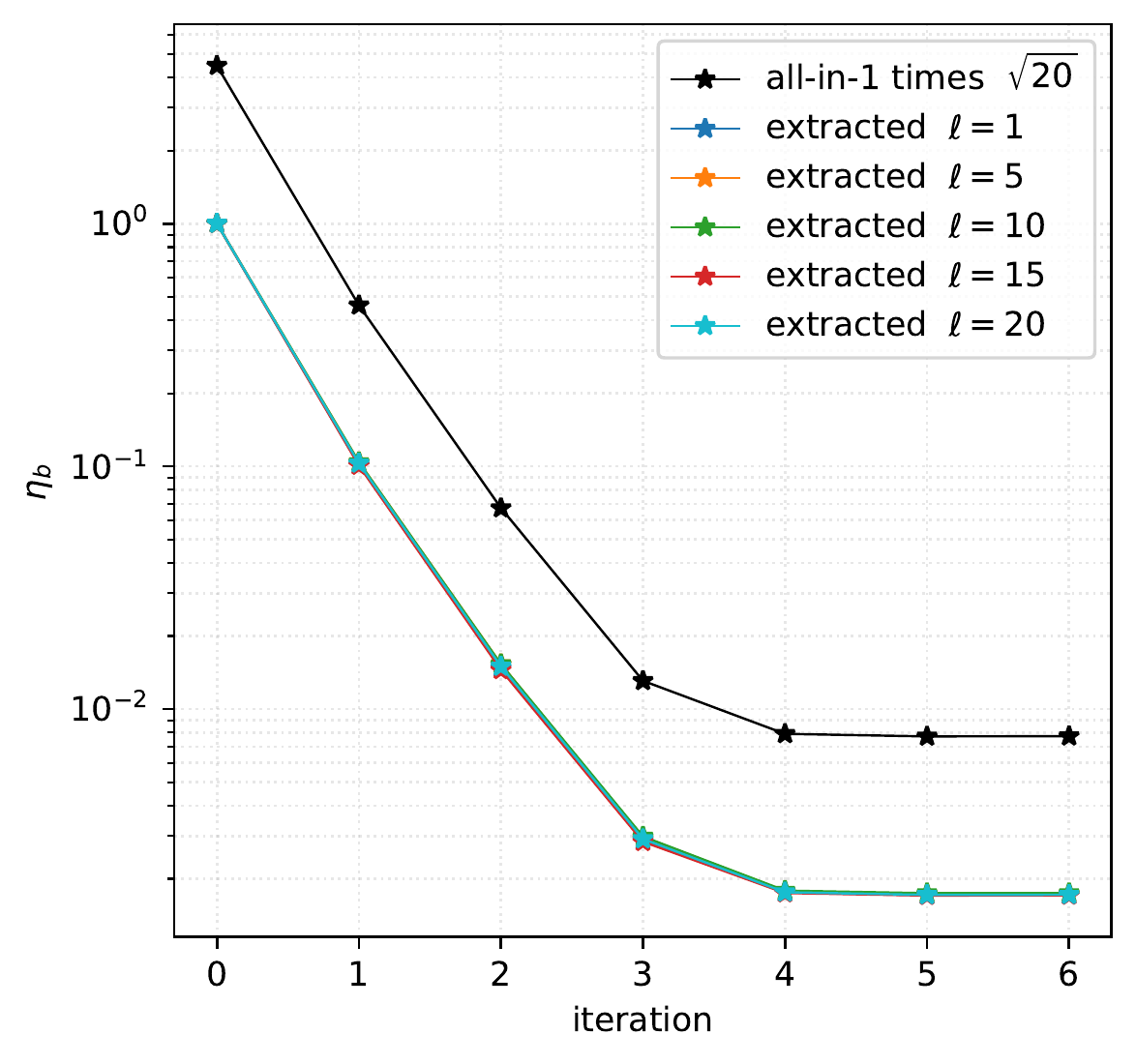}\label{fig:mLap_eta-b_CH_128}}
  		\quad
  		\subfloat[Convergence history in $\eta_{\ten{b}}$ for $n = 127$]{\includegraphics[scale=0.3, width=0.31\linewidth, height=0.31\linewidth]{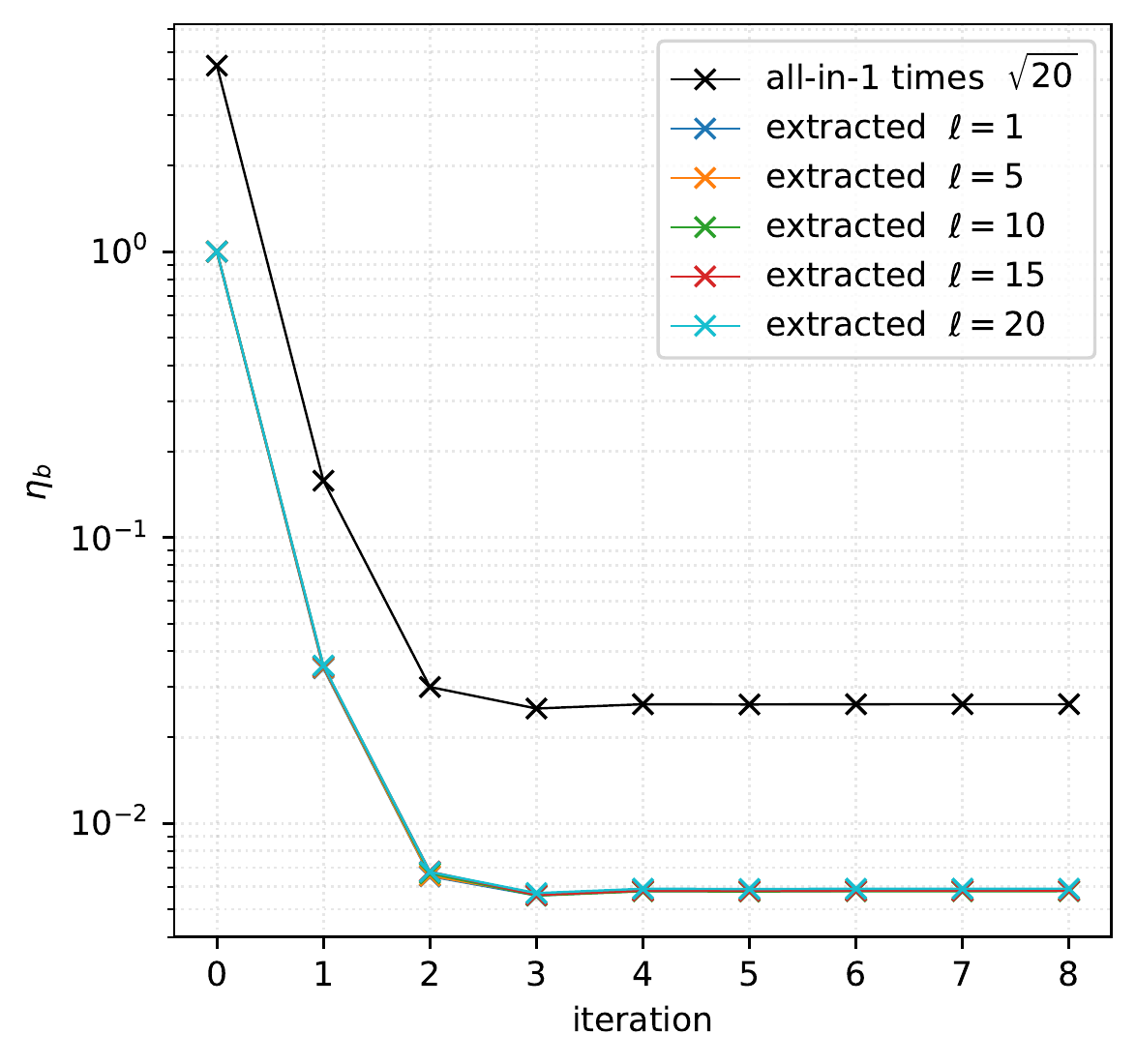}\label{fig:mLap_eta-b_CH_256}}
  		\caption{4-d  Poisson problem $\eta_{\ten{b}}$ bound  using $\delta = \varepsilon= 10^{-5}$} 
  		\label{fig:mLap_SvsM_eta_b}
  	\end{figure}

  	\begin{figure}[!htb]
  	  \centering
  	\subfloat[Convergence history in $\eta_{\ten{AM}, \ten{b}}$ for $n = 63$]{\includegraphics[scale =0.3, width=0.31\linewidth, height=0.31\linewidth]{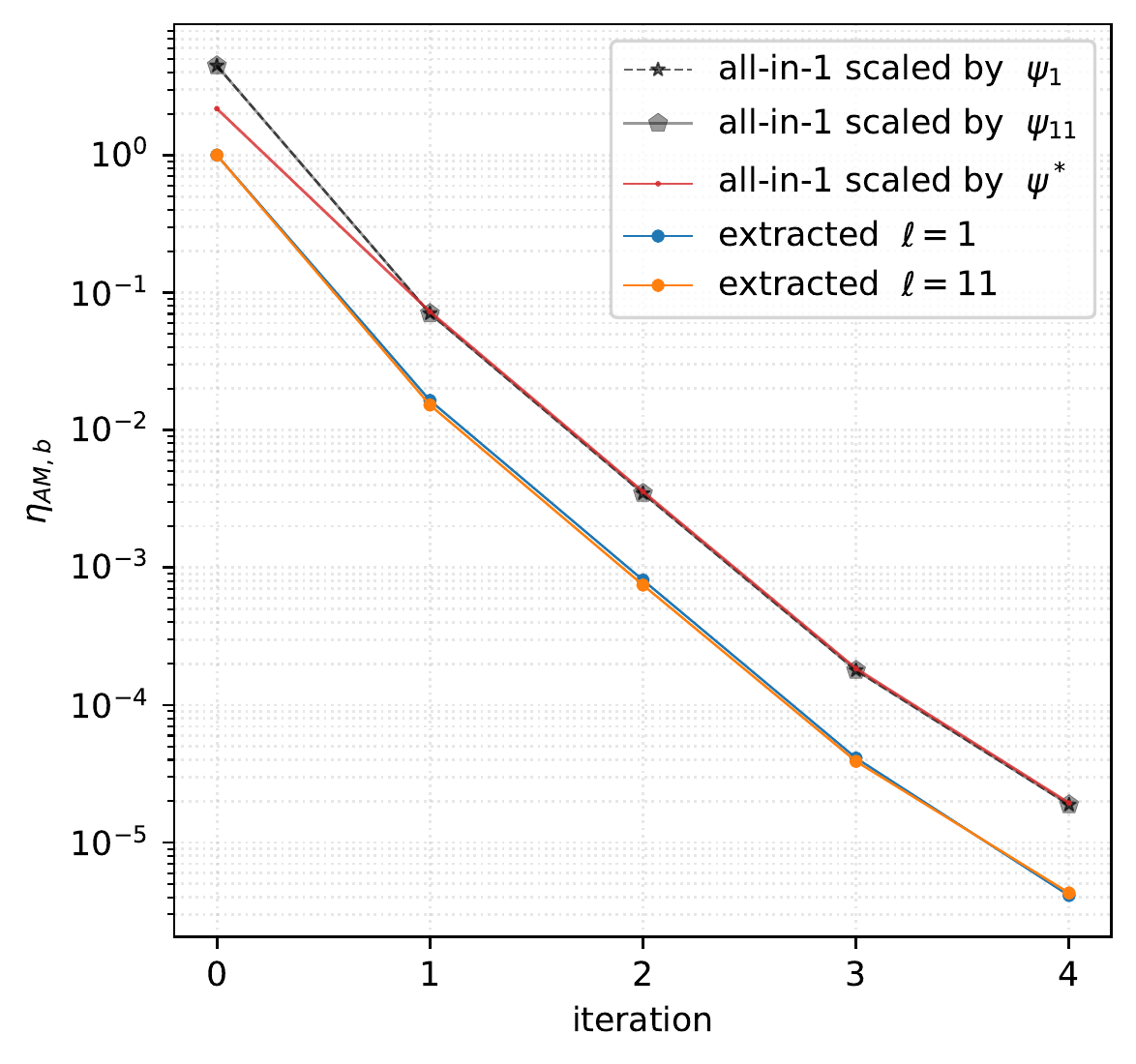}\label{fig:mLap_eta-AMb_CK_64}}
  	\quad
  	\subfloat[Convergence history in $\eta_{\ten{AM},\ten{b}}$ for $n = 127$]{\includegraphics[scale=0.3, width=0.31\linewidth, height=0.31\linewidth]{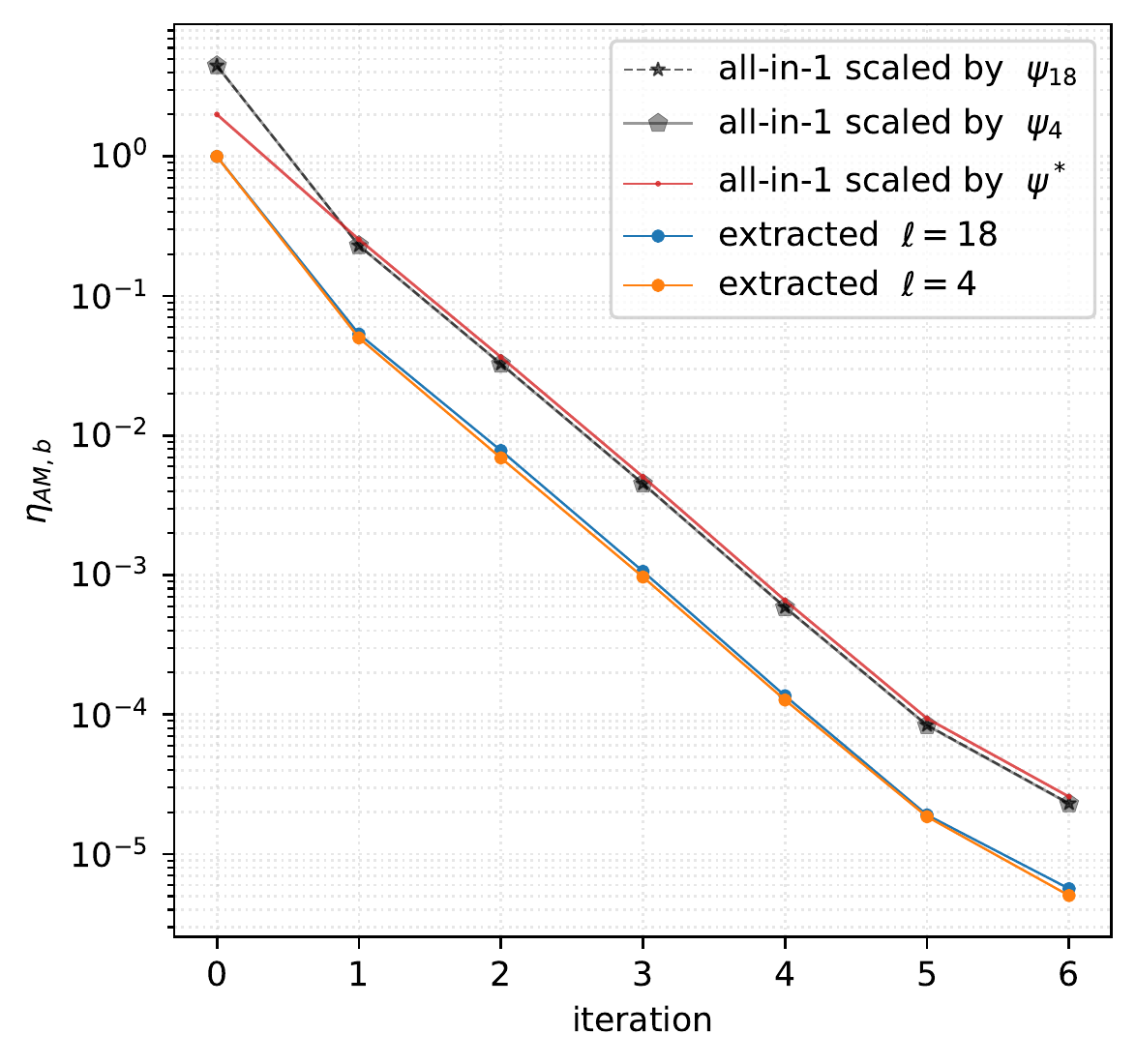}\label{fig:mLap_eta-AMb_CH_128}}
  	\quad
  	\subfloat[Convergence history in $\eta_{\ten{AM},\ten{b}}$ for $n = 127$]{\includegraphics[scale=0.3, width=0.31\linewidth, height=0.31\linewidth]{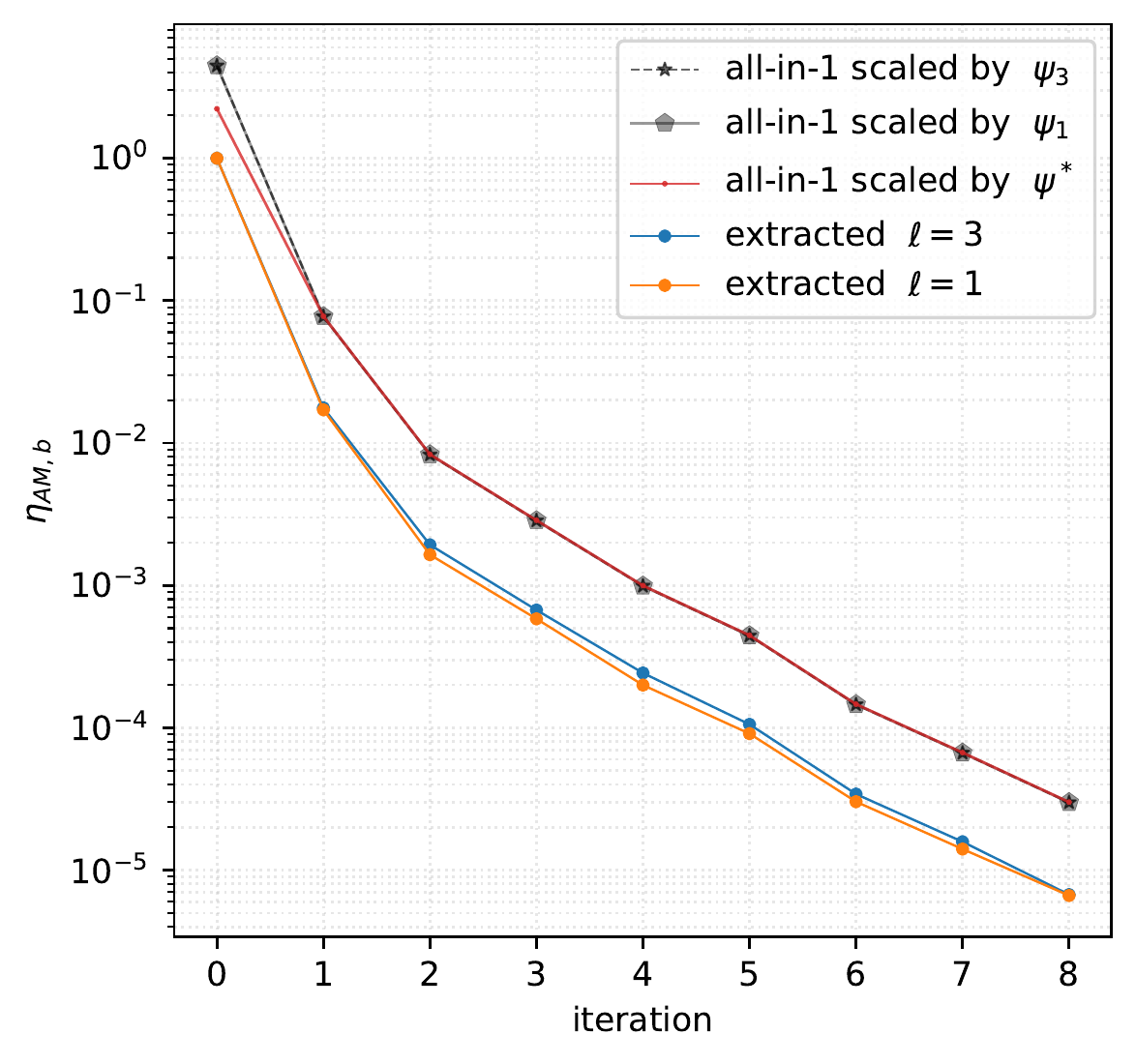}\label{fig:mLap_eta-AMb_CH_256}}
  	\caption{4-d multiple right-hand side Poisson problem $\eta_{\ten{AM},\ten{b}}$ bound using $\delta = \varepsilon= 10^{-5}$} 
  	\label{fig:mLap_SvsM_eta_AM_b}
  	\end{figure}
  	
  	\subsubsection{Convection-diffusion problem\label{sec4:mCD}} As previously, the aim of this subsection is to illustrate the solution of multiple convection-diffusion problem~\eqref{eqNE:CD}, with different right-hand sides. Let $\ten{A}_0$ be the discretization of~\eqref{eqNE:CD} operator over a Cartesian grid of $n$ points per mode for the domain $\Omega = [0, 1]^3$. Let $\ten{b}\in\R^{n\times n\times n}$ be the right-hand side discretization defined in Section~\ref{sssec:CD}.
 We define the individual linear system as
  	\[
  	\ten{A}_0 \ten{u}_\ell = \ten{b} + \ten{e}_\ell
  	\]
  	where $\ten{e}_\ell\in\R^{n\times n \times n}$ is a realization of the normal distribution $\mathcal{N}(0,1)$ for every $\ell\in\{1,\dots, p\}$. Since the aim is solving simultaneously the $p$ problems, as in Section~\ref{ssec:Ai1}, we define the ``all-in-one'' tensor linear operator $\ten{A}\in\R^{(p\times p)\times (n\times n)\times(n\times n)\times(n\times n)}$
  	\[
  	\ten{A} = \I_p \otimes (\ten{-\Delta}_3)
  	\]
  	while the ``all-in-one'' right-hand side is $\ten{c}\in\R^{p\times n\times n\times n}$ such that
  	\[
  	\ten{c}(\ell, i_1, i_2, i_3) = \ten{b}(i_1, i_2, i_3) + \ten{e}_\ell(i_1, i_2, i_3). 
  	\]
  	for every $i_k\in\{1,\dots, n_k\}$, $\ell\in\{1,\dots, p\}$ for $k\in\{1,\dots,3\}$. The problem is solved for $n\in\{63, 127\}$ and $p = 20$. As in all the previous cases of study, we use the preconditioner stated in~\eqref{eq:prec_d+1} with $q\in\{16, 32\}$ and we impose a small TT-rank to $\ten{e}_\ell$, so that the TT-rank of $\ten{c}$ ends up being $11$ at maximum.
  	
  	Figure~\ref{fig:mCD_CH} illustrates the convergence history in $5$ iterations for both the grid {dimensions}. If we compare this Figure~with Figure~\ref{fig:CD_CH}, we observe that the curves are very similar. Generally speaking, the number of iterations for GMRES to converge, neglecting the effect of the rounding, is equal to the number of eigenvectors which span the subspace where the right-hand side lives. This implies that if all the right-hand sides belong to the same linear subspace, the number of iterations necessary to converge is the same, implying that under this hypothesis solving for $1$ or $p$ right-hand sides requires the same number of iterations. This point is further discussed in 
\reportPaper{Appendix~\ref{app:eig}}{~\cite[Appendix~C]{coulaud2022}}. 
  	From the point of view of the memory consumption, the comparison of Figure~\ref{fig:CD_v-rank} and~\ref{fig:mCD_v-rank} shows that the solution for $20$ right-hand sides leads to TT-rank significantly larger, from $25$ to $30$ in the single right-hand side solution versus more than $200$ for the $20$ right-hand side ``all-in-one'' system. However if we had solved $20$ systems independently, summing all the TT-ranks, we could have reached a maximum of $500$ up to $700$. This becomes more interesting if we compare the compression ratios for the last Krylov vector, looking at Figure~\ref{fig:CD_v-ratio} and~\ref{fig:mCD_v-ratio}. We have a ratio from $0.02$ up to $0.12$ for a single right-hand side solution versus $0.1$ up to $0.17$ for the simultaneous one, which shows that these ratios are extremely closed, considering that in the second case we are solving in a higher dimension. A similar argument holds for the ratio of compression of the entire Krylov basis. In Figure~\ref{fig:mCD_v-ratio_s}, the ratio is between $0.06$ and $0.12$, while in Figure~\ref{fig:CD_v-ratio_s} it is between $0.01$ and $0.07$.
  	\begin{figure}[!htb]
  		\centering
  		\subfloat[Convergence history]{\includegraphics[scale =0.45, width=0.33\linewidth, height=0.33\linewidth]{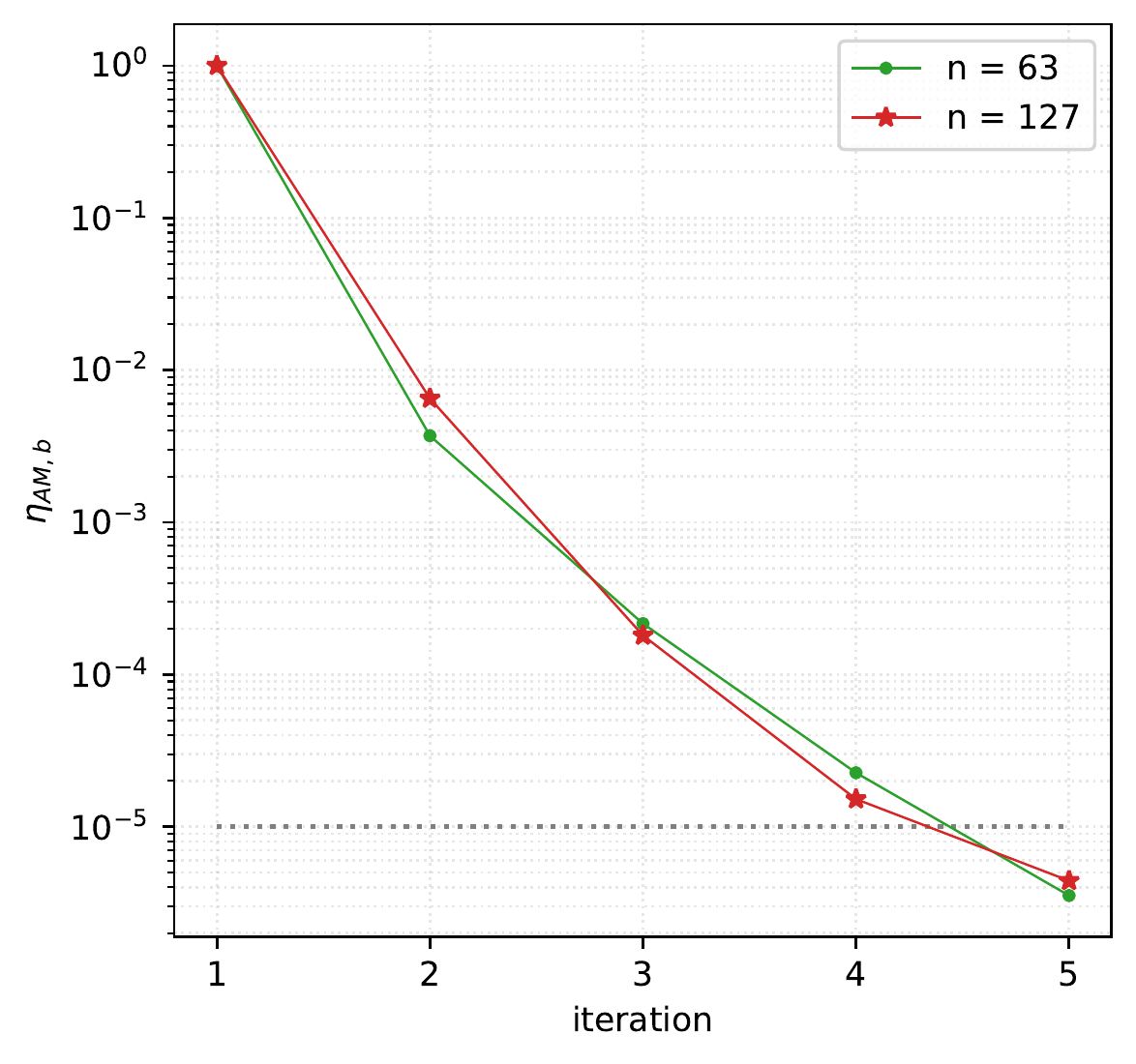}\label{fig:mCD_CH}}
  		\quad
  		\subfloat[Maximal TT-rank of the last Krylov vector ]{\includegraphics[scale=0.45,width=0.33\linewidth, height=0.33\linewidth]{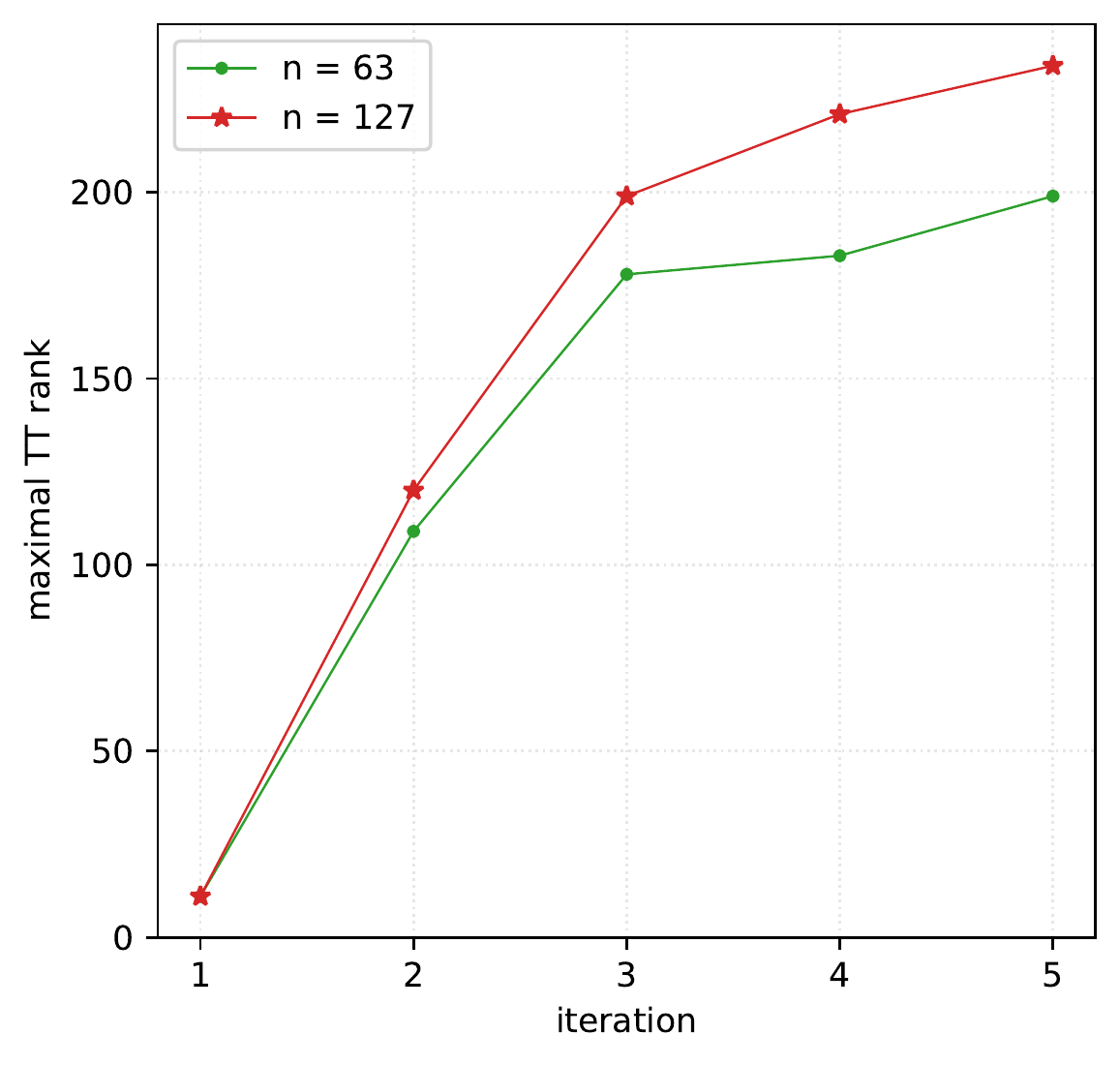}\label{fig:mCD_v-rank}}
  		\vskip\baselineskip
  		\subfloat[Compression ration for the last Kyrolv vector]{\includegraphics[scale=0.45, width=0.33\linewidth, height=0.33\linewidth]{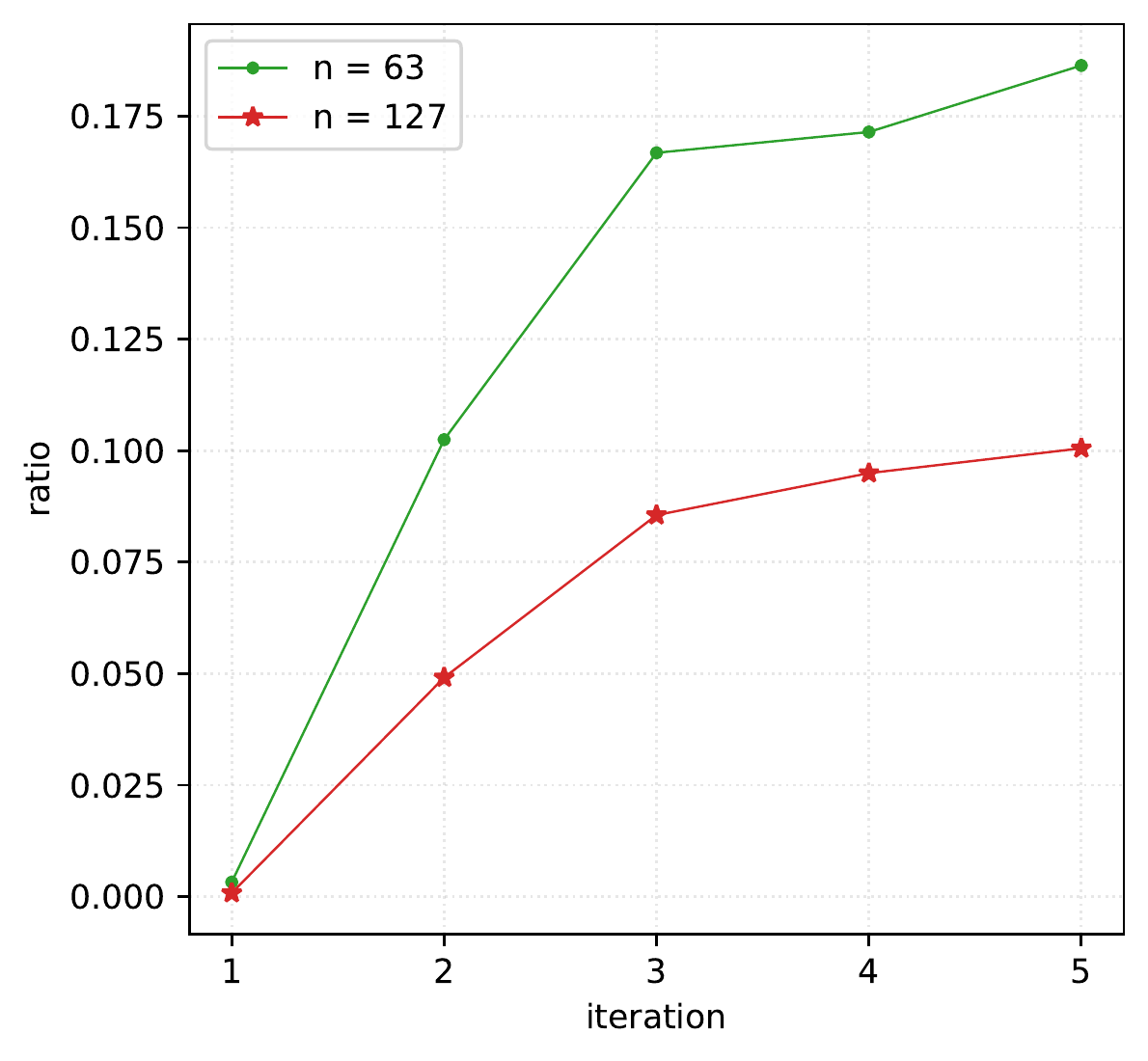}\label{fig:mCD_v-ratio}}
  		\quad
  		\subfloat[Compression ratio for the entire Krylov basis]{\includegraphics[scale=0.45, width=0.33\linewidth, height=0.33\linewidth]{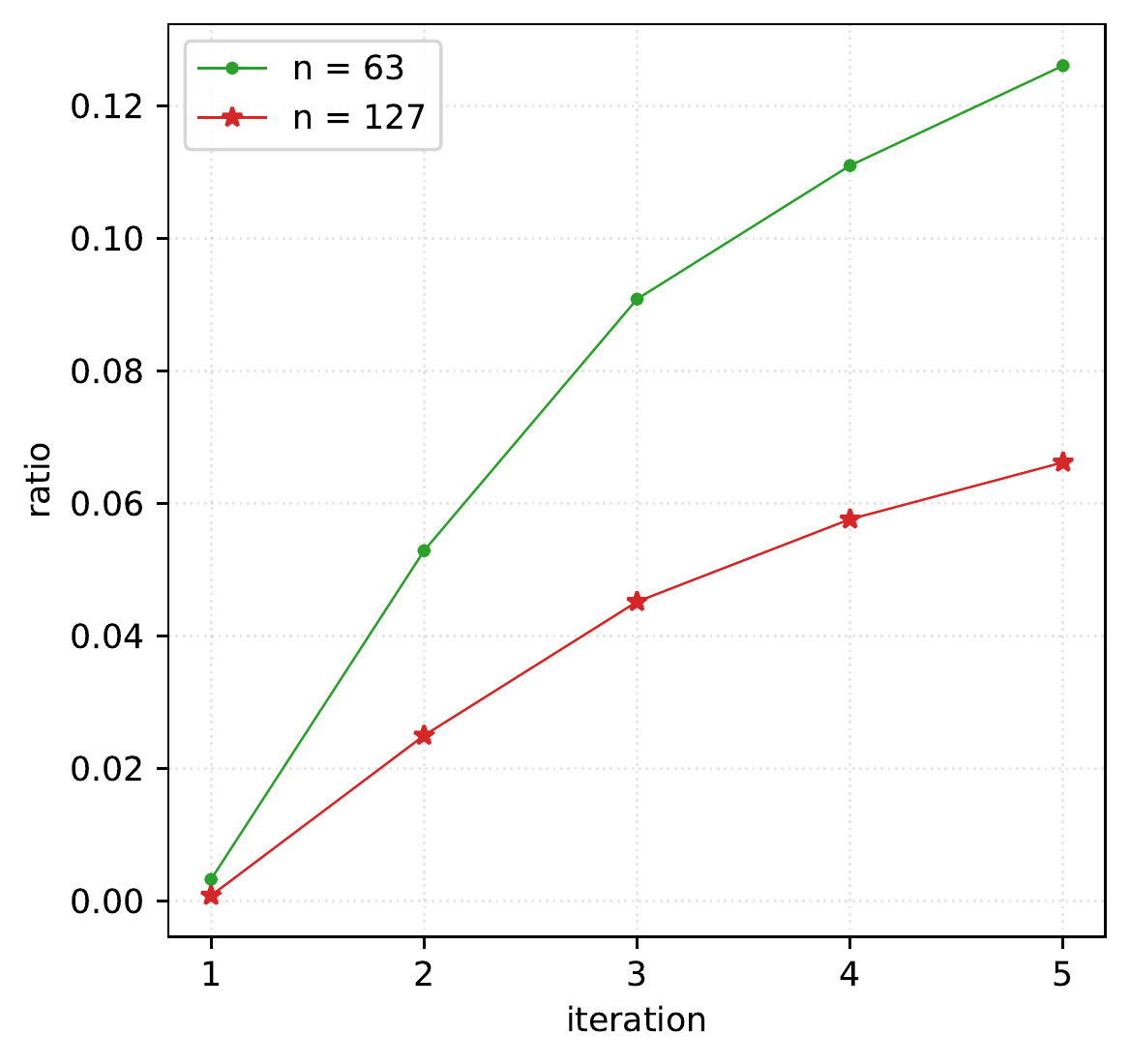}\label{fig:mCD_v-ratio_s}}
  		\caption{4-d multiple right-hand sides convection-diffusion problem  using $\delta = \varepsilon= 10^{-5}$}
  		\label{fig:mCD}
  	\end{figure}
  	
  	In Figure~\ref{fig:mCD_SvsM_eta_b}, we present the bound for $\eta_{\ten{b}}$ stated in Proposition~\ref{prop:eta_b}.
  	We see that it is quite tight during the first iterations and gets more loose at the end, setting at more than $1$ order of difference. As in the previous subsection, we compute $\ell_m$ and $\ell_M$ according to Equation~\eqref{eq3:mrhsM}, deciding which curves are plotted in Figure~\ref{fig:mCD_SvsM_eta_AM_b}. The resulting bound, displayed in Figure~\ref{fig:mCD_SvsM_eta_AM_b}, is quite tight, being of slightly less than $1$ order of magnitude approximately, with the three scaled curves overlapping from the second iteration. 
  	\begin{figure}[p]
  		\centering
  		\subfloat[Convergence history in $\eta_{\ten{b}}$ for $n = 63$]{\includegraphics[scale =0.45, width=0.33\linewidth, height=0.33\linewidth]{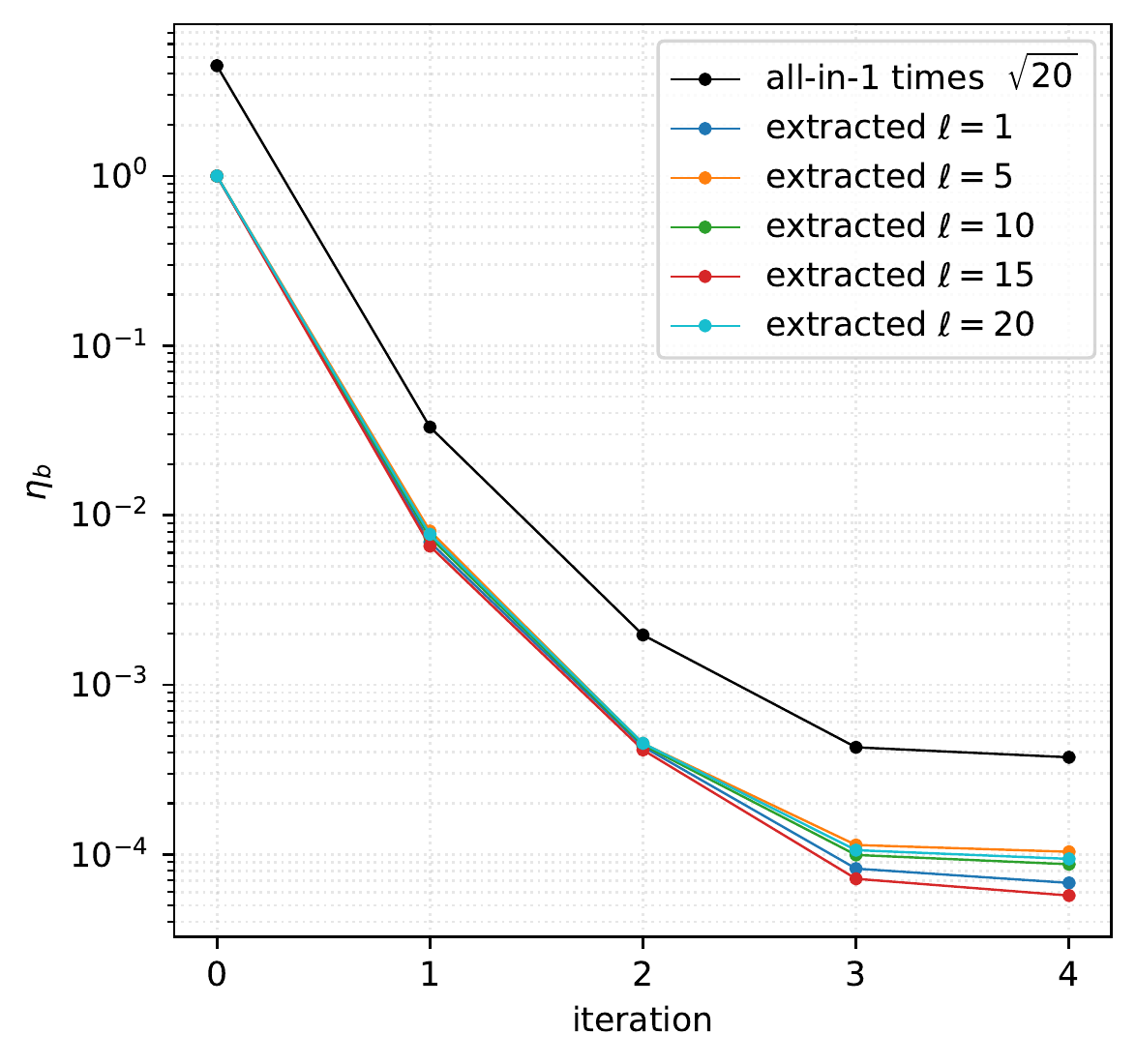}\label{fig:mCD_eta-b_64}}
  		\quad
  		\subfloat[Convergence history in $\eta_{\ten{b}}$ for $n = 127$]{\includegraphics[scale=0.45, width=0.33\linewidth, height=0.33\linewidth]{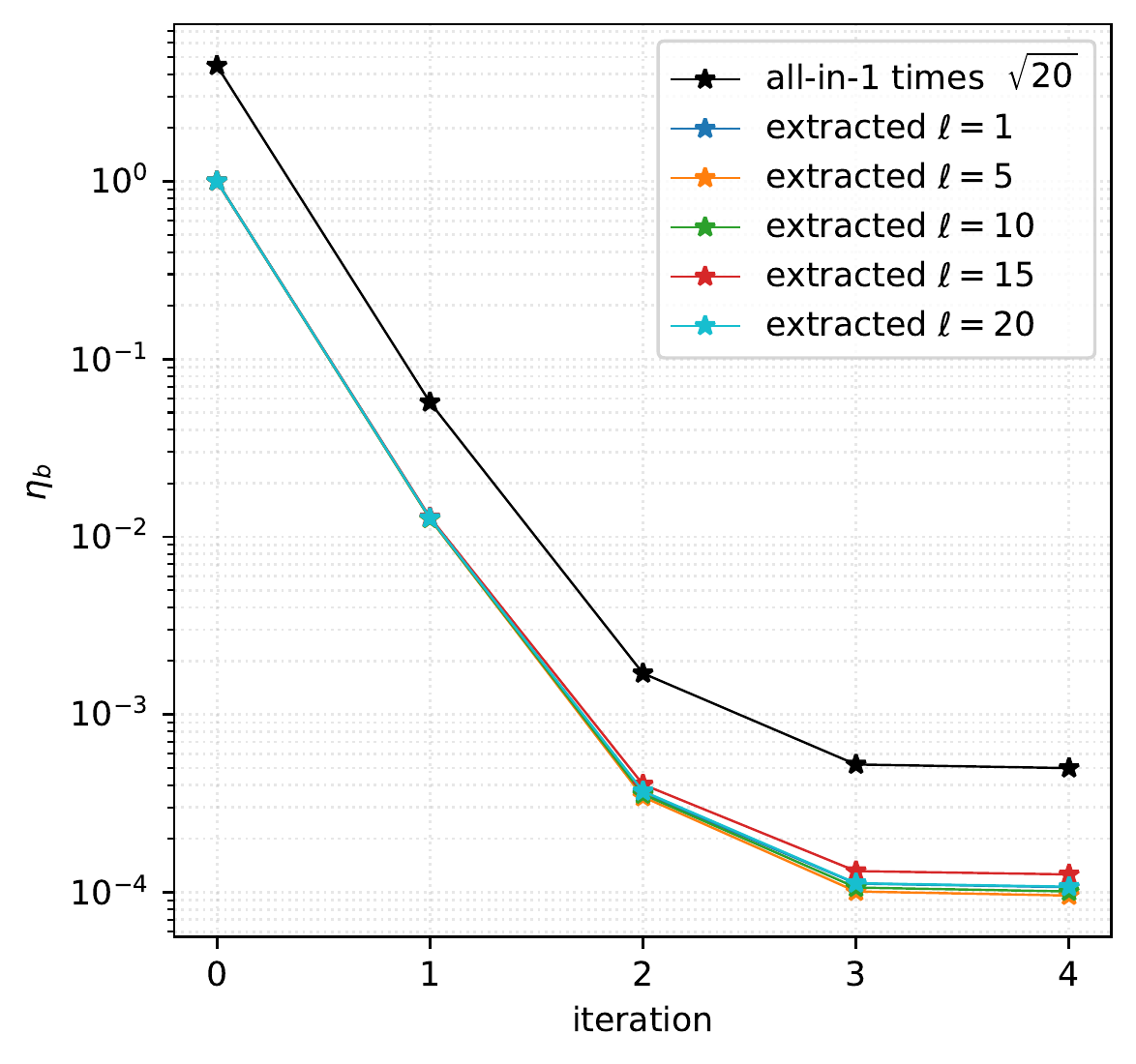}\label{fig:mCD_eta-b_CH_128}}
  		
  		\caption{4-d convection-diffusion problem $\eta_{\ten{b}}$ bound using $\delta = \varepsilon= 10^{-5}$} 
  		\label{fig:mCD_SvsM_eta_b}
  	\end{figure}
  	\begin{figure}[p]
  		\centering
  		\subfloat[Convergence history in $\eta_{\ten{AM}, \ten{b}}$ for $n =~63$]{\includegraphics[scale =0.45, width=0.33\linewidth, height=0.33\linewidth]{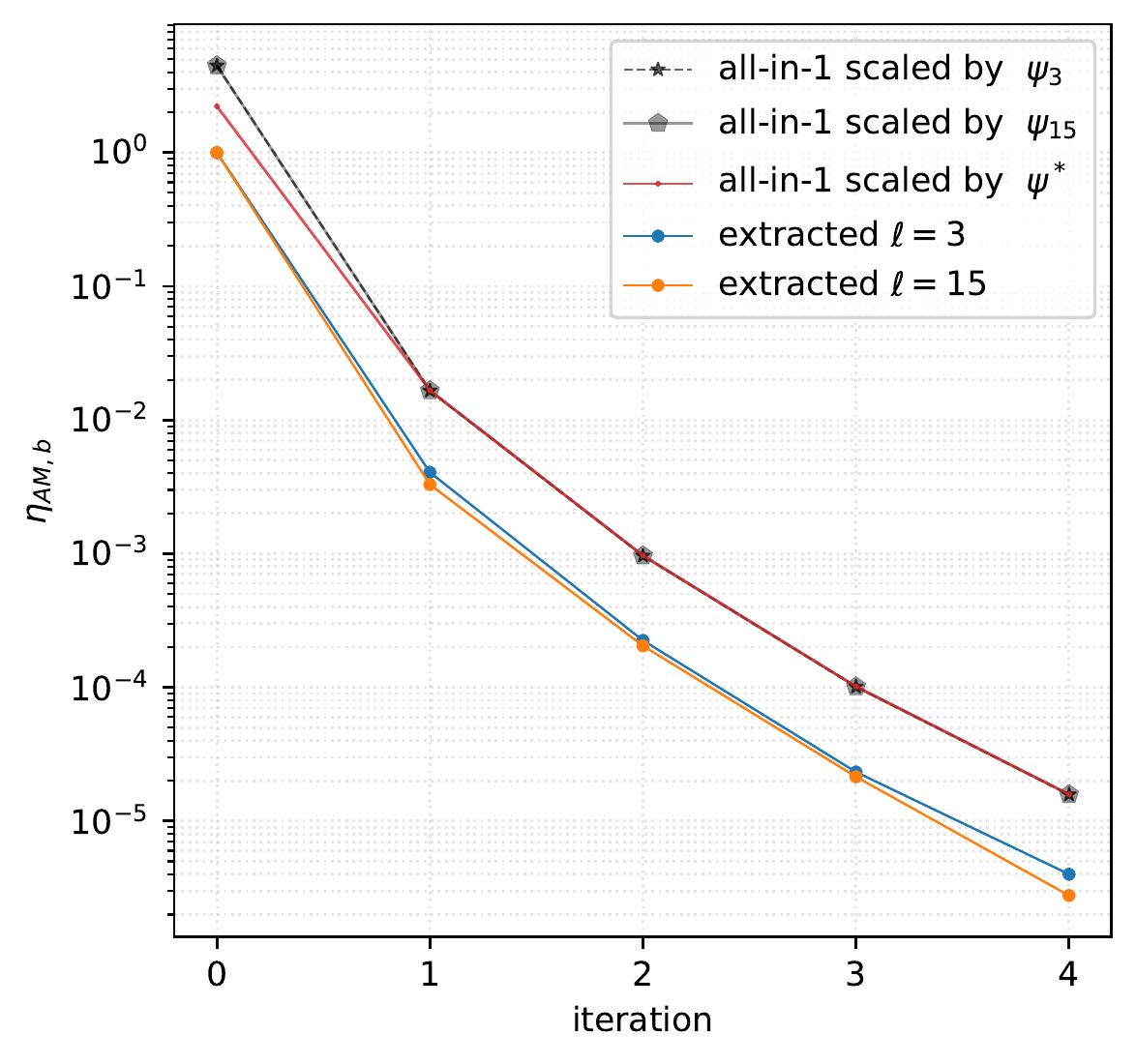}\label{fig:mCD_eta-AMb_CK_64}}
  		\quad
  		\subfloat[Convergence history in $\eta_{\ten{AM},\ten{b}}$ for $n =~127$]{\includegraphics[scale=0.45, width=0.33\linewidth, height=0.33\linewidth]{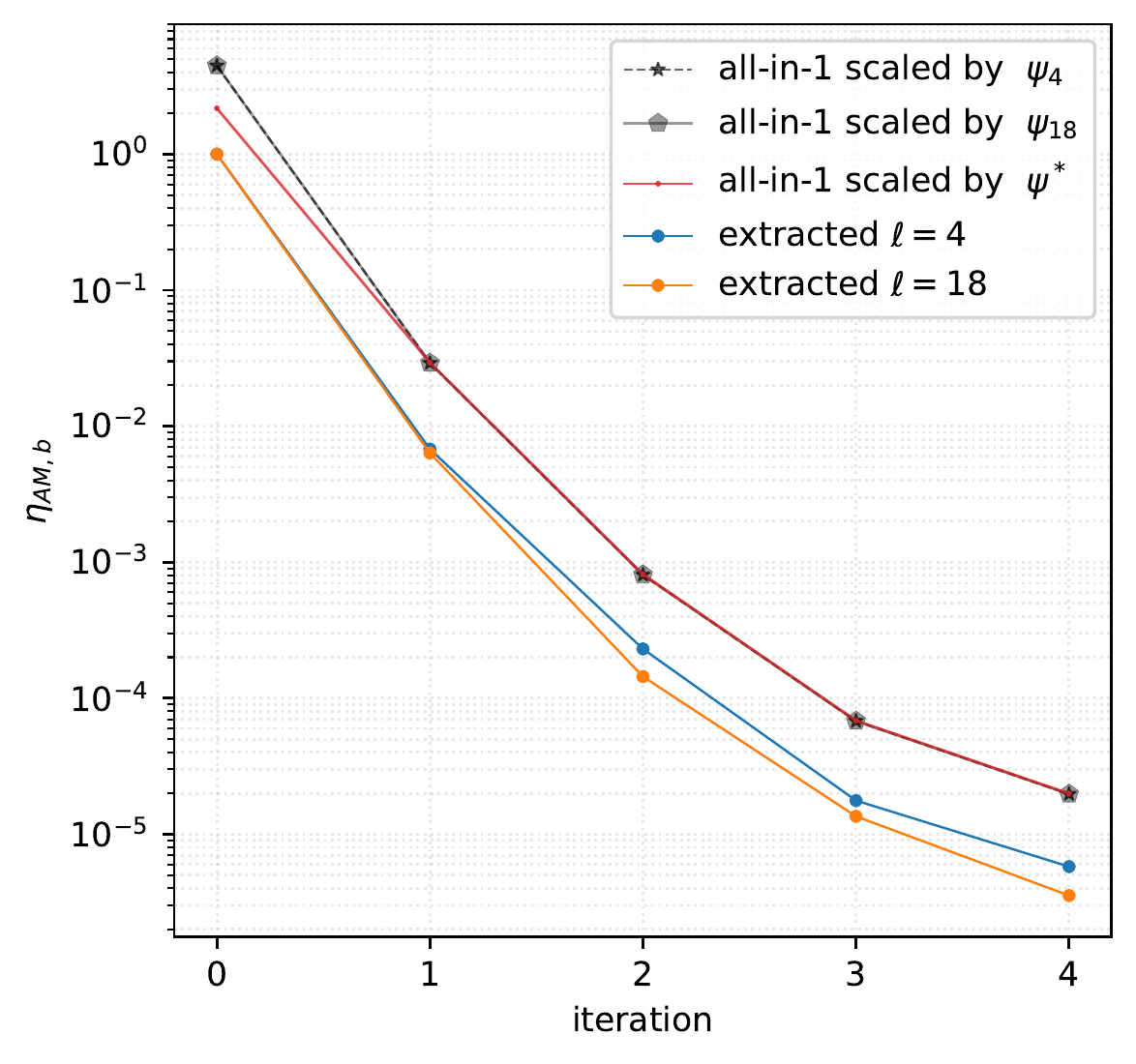}\label{fig:mCD_eta-AMb_CH_128}}
  		
  		\caption{4-d multiple right-hand side convection-diffusion problem $\eta_{\ten{AM},\ten{b}}$ bound using $\delta = \varepsilon= 10^{-5}$} 
  		\label{fig:mCD_SvsM_eta_AM_b}
  	\end{figure}
  
\reportPaper{}{We refer to Appendix~A (Appendix~B) of~\cite{coulaud2022} for further details on the
  the preconditioner choice and the convergence of problems with the same operator and different right-hand sides solved simultaneously respectively.}



\section{Concluding remarks}
%
	In this work we proposed a GMRES algorithm for solving high-{dimensional} liner systems expressed in TT-format and we investigated numerically its backward stability. The examples presented in Section~\ref{sec3} suggest that the backward stability properties observed in the matrix framework still hold true in the tensor context where the recompression (TT-rounding operation) is the dominating part of the computational round-off. Several of these examples enable us to evaluate the tightness of the  proposed backward  bounds, theoretically proved in Section~\ref{sec2}. The existence of these bounds together with the memory requirement illustrate the capabilities of the simultaneous approach when solving parametric tensor linear systems. In Section~\ref{ssb:Dolgov} we highlight the differences between our algorithm and its previously presented implementation~\cite{Dolgov2013}, stressing that our approach guarantees the backward stability property of the computed solution. The proposed TT-GMRES algorithm still carries some intrinsic drawbacks. The memory requirement increases with the number of iterations, making crucial the use of an efficient preconditioner. Therefore the development of effective preconditioner for multilinear operators is a challenging open  question. 

\section*{Acknowledgement}
Experiments presented in this paper were carried out using the PlaFRIM experimental testbed, supported by Inria, CNRS (LABRI and IMB), Université de Bordeaux, Bordeaux INP and Conseil Régional d’Aquitaine (see https://www.plafrim.fr). 

    \printbibliography
\clearpage
\appendix
\renewcommand{\thesubsection}{\thesection.\Alph{subsection}}
\renewcommand\thetheorem{\thesubsection.\arabic{theorem}}
\section*{Appendices}
\addcontentsline{toc}{section}{Appendices}
\renewcommand{\thesubsection}{\Alph{subsection}}
\subsection{Preconditioner parameter study} \label{app:prec}
In this first appendix the preconditioner firstly introduced in~\eqref{eqTT:5} is further investigated. In particular we focus on the effect on the convergence of the number of addends and on the compression accuracy chosen to compute it. 
As in Subsection~\ref{ssec2:TT}, let $\ten{M}\in\R^{(n\times n)\times \dots \times (n\times n)}$ be the  $d$-order {TT-matrix} that approximate the inverse of the discrete Laplacian $\ten{\Delta}_d$, cf.~\cite{Hackbusch2006I,Hackbusch2006II}, defined as
\[
\ten{M} = \sum_{k=-q}^{q} c_k \exp(-t_k\Delta_1)\otimes \cdots\otimes \exp(-t_k\Delta_1)
\]
where $c_k = \eta t_k$, $t_k = \exp(k\eta)$ and $\eta = \pi/q$. As already mentioned, since $\ten{M}$ is a sum of tensors, its TT-rank is greater or equal than $2q+1$. The magnitude of TT-ranks of $\ten{M}$  conditions the TT-ranks of the Krylov basis vectors and of the final solution. Therefore it is convenient to keep $\ten{M}$ TT-rank significantly small, either by reducing the number of addends, i.e., choosing a low value for $q$, or by compressing $\ten{M}$ with an accuracy $\tau$.  
With the help of a Poisson problem, we illustrate the trade off between number of addends and compression accuracy which leads to the optimal convergence.
We consider the Poisson problem, written as
	\[
		\begin{cases}
		-\Delta u = 1 \quad&\text{in}\quad \Omega = [0, 1]^3 \\
		\;\;\; u = 0 \quad&\text{in}\quad \partial \Omega
		\end{cases}
	\]
 so that the effect of the preconditioner is as much evident as possible. Let $-\ten{\Delta}_3$ be the discretization of the Laplacian operator over a grid of $n = 63$ points per mode. Similarly let $\ten{b}\in\R^{n\times n\times n}$ be a tensor with all the entries equal to $1$. Then, setting $\ten{A} = -\ten{\Delta}_3$, TT-GMRES solves the tensor linear system $\ten{A}\ten{u} = \ten{b}$, preconditioning it on the right as
 \[
 	\ten{A}\ten{M}_{q, \tau}\ten{t} = \ten{b}
 \]
 for $q\in\{2, 8, 16, 32, 64\}$ and $\tau\in\{10^{-2}, 10^{-8}\}$. The parameters of TT-GMRES are tolerance $\varepsilon = 10^{-16}$, rounding accuracy $\delta = 10^{-5}$, dimension of the Krylov space $m = 25$ and a maximum of $2$ restart. We set TT-GMRES tolerance equal to the machine precision so that the algorithm performs all the $50$ iterations. In Table~\ref{tab:1}, we report the maximal TT-rank of $\ten{M}_{q, \tau}$ and the $\norm{\ten{A}\ten{M}_{q, \tau}}_2$ rounded at the third digits for all the combinations of $q$ and $\tau$. Remark that fixed a value for $q$ the L2 norm of the preconditioned linear system is the same up to the third digits for both the values of $\tau$. This seems to suggest that the number of addends plays a key role in determining the quality of the preconditioner, while the rounding accuracy $\tau$ affects more significantly the TT-rank, removing kind of unnecessary information. Indeed for $\tau = 10^{-2}$ and $q \ge 8$, the maximal value of the TT-rank is always $5$, but depending for an increasing number of addends, the L2 norm gets closer to $1$. Similarly for $\tau = 10^{-8}$ and $q\ge 32$, the maximal TT-rank is $15$ and the rounded L2 norm is equal to $1$. 
 \begin{table}[H]
 	\centering
 	\scalebox{0.85}{
 	\begin{tabular}{ll ccccc ccc ccccc}
 		\toprule
 		&& \multicolumn{5}{c}{$\mathbf{\tau = 10^{-2}}$} &&&& \multicolumn{5}{c}{{$\mathbf{\tau = 10^{-8}}$}} \\
 		\cmidrule(lr){3-7} \cmidrule(lr){10-15}
 		\thead{$\mathit{q}$}   && \thead{$\mathbf{2}$} & \thead{$\mathbf{8}$} & \thead{$\mathbf{16}$} & \thead{$\mathbf{32}$} & \thead{$\mathbf{64}$} &&&& \thead{$\mathbf{2}$} & \thead{$\mathbf{8}$} & \thead{$\mathbf{16}$} & \thead{$\mathbf{32}$} & \thead{$\mathbf{64}$} \\
 		\cmidrule(lr){3-7} \cmidrule(lr){10-15}
 		Max TT-rank of $\ten{M}$&& $2$ & $5$ & $5$ & $5$ & $5$ &&&& $2$ & $7$ & $13$ & $15$ & $15$\\
 		L2 norm of $\ten{A}\ten{M}_{q, \tau}$&& $0.012$&
 		$0.276$&
 		$0.949$&
 		$1.00$&
 		$1.00$
 		 &&&& $0.012$&
 		 $0.276$&
 		 $0.949$&
 		 $1.00$&
 		 $1.00$\\ 
 		\bottomrule
 	\end{tabular}}
 	\caption{Preconditioner properties for grid step $n = 63$.}
 	\label{tab:1}
 \end{table}

Looking at the convergence history in Figures~\ref{fig:LapP_CH_1e-2} and~\ref{fig:LapP_CH_1e-8}, a value of $q \ge 16$ is already sufficient to reach in a very low number of iterations the bound $10^{-5}$, due to the TT-GMRES rounding value $\delta$. Figure~\ref{fig:LapP_CH_1e-8} shows clearly that keeping more information in the preconditioner, TT-GMRES may reach very low levels. However in Figure~\ref{fig:LapP_v-rank_1e-8} we observe the side effect of more information. The TT-rank of the last Krylov vector increases significantly for very accurate preconditioner. More precisely, comparing Figures~\ref{fig:LapP_v-rank_1e-2} and~\ref{fig:LapP_v-rank_1e-8}, the TT-rank of the last Krylov vector doubles if the preconditioner is more accurately rounded. Notice also that in Figure~\ref{fig:LapP_v-rank_1e-2}, the TT-rank for $q\in\{16, 32, 64\}$ is almost the same. For the solution viewpoint, the rounding accuracy chosen for the preconditioner has not a big impact on its TT-rank. Indeed, as plotted in Figure~\ref{fig:LapP_x-rank_1e-2} and~\ref{fig:LapP_x-rank_1e-8}, for both the values of $\tau$ and for all $q \ge 8$, the TT-rank of the solution is equal to $5$, while only for $q = 2$ it increases, meaning that only $5$ addends are not sufficient to speed up the discrete Laplacian convergence. 

\begin{figure}[!htb]
	\centering
	\subfloat[\footnotesize Convergence history for $\tau = 10^{-2}$ and rounding $\delta = 10^{-5}$]{\includegraphics[scale=0.45, width=0.4\linewidth, height=0.4\linewidth]{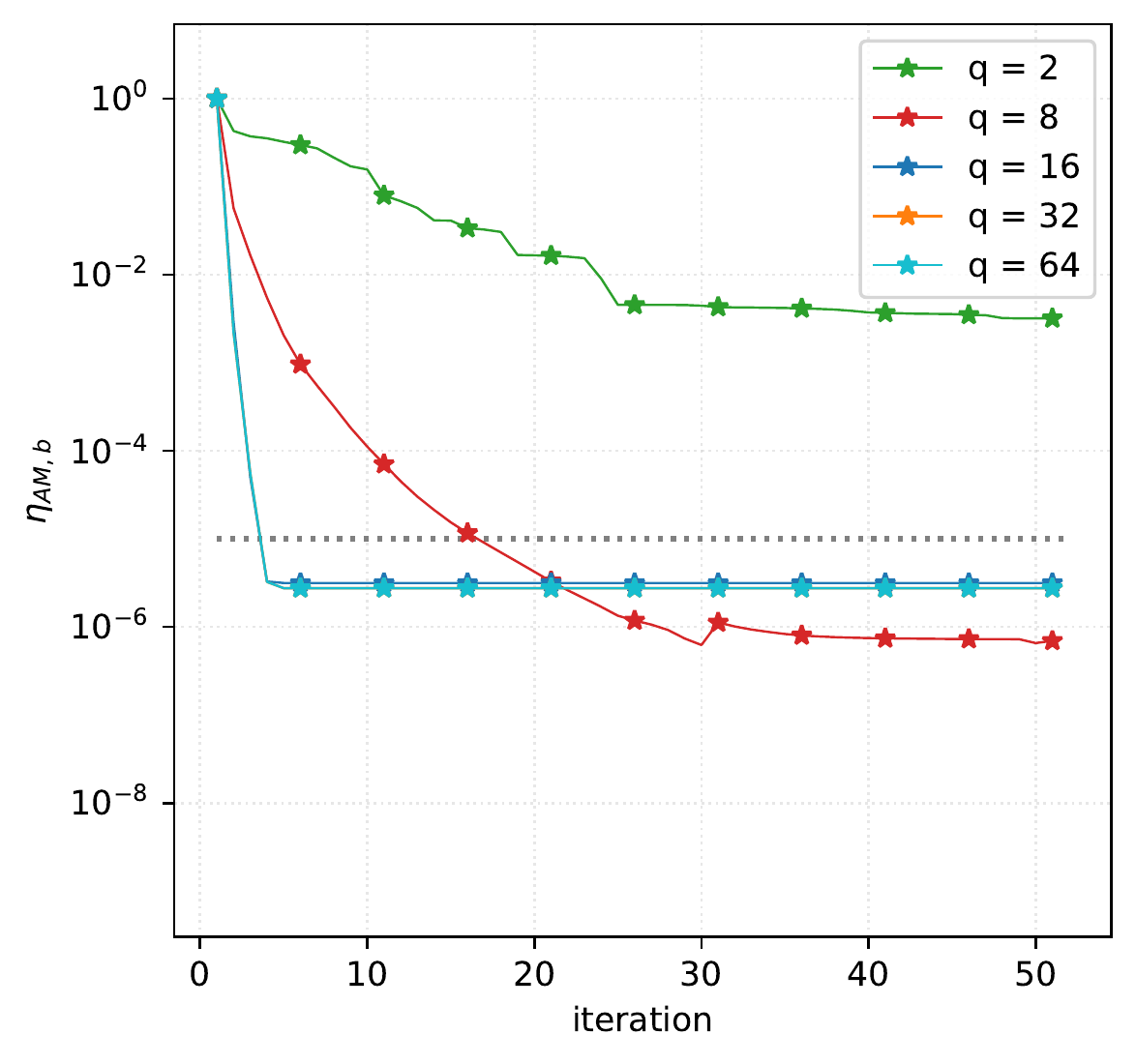}\label{fig:LapP_CH_1e-2}}
	\quad
	\subfloat[\footnotesize Convergence history for $\tau = 10^{-8}$ and rounding $\delta = 10^{-5}$]{\includegraphics[scale=0.45, width=0.4\linewidth, height=0.4\linewidth]{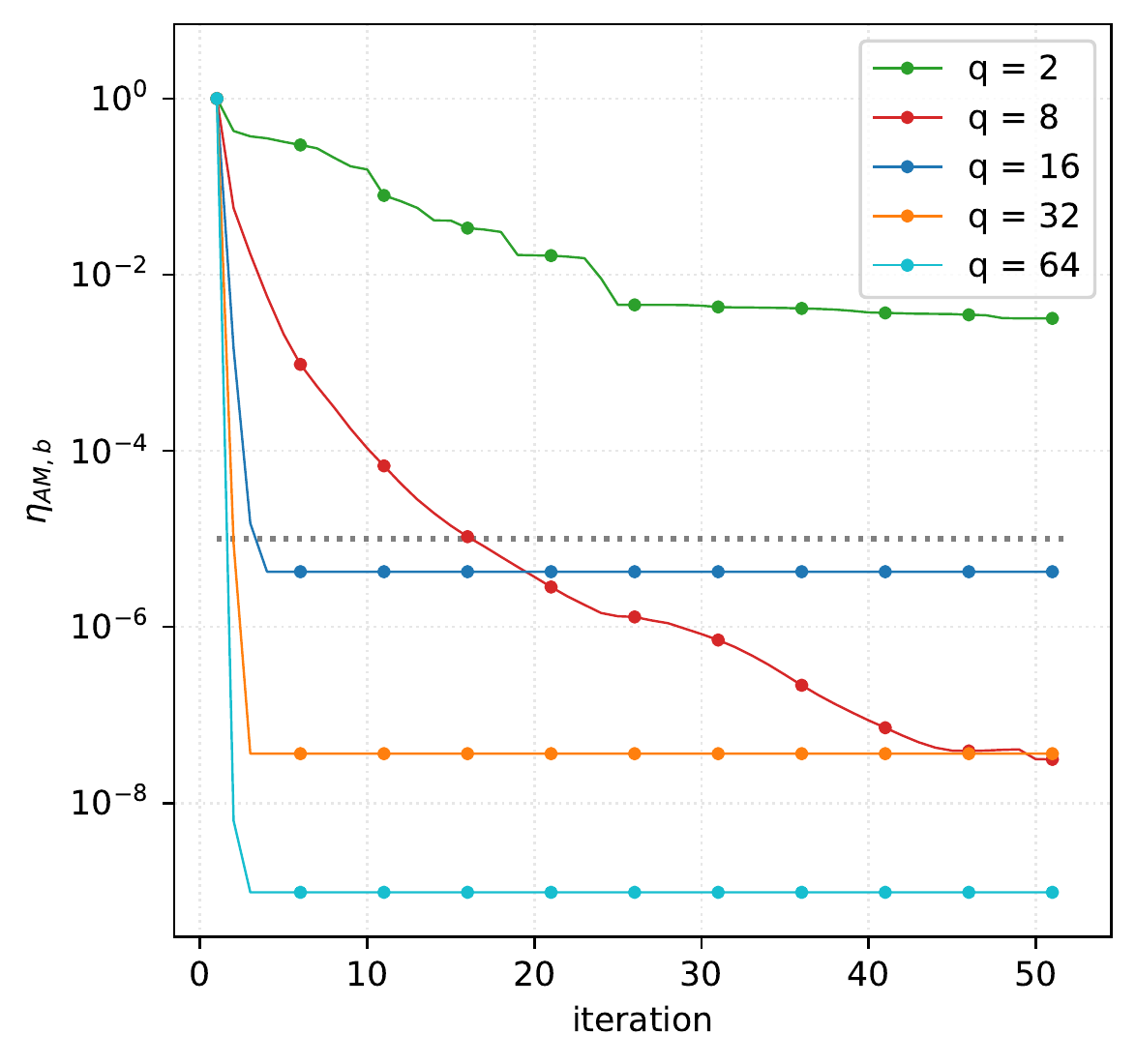}\label{fig:LapP_CH_1e-8}}
	\caption{3-d Poisson problem, comparing preconditioners}
\end{figure}

\begin{figure}[!htb]
		\centering
		\ContinuedFloat
		\subfloat[\footnotesize Maximal TT-rank of the last Krylov vector for $\tau = 10^{-2}$]{\includegraphics[scale=0.45, width=0.4\linewidth, height=0.4\linewidth]{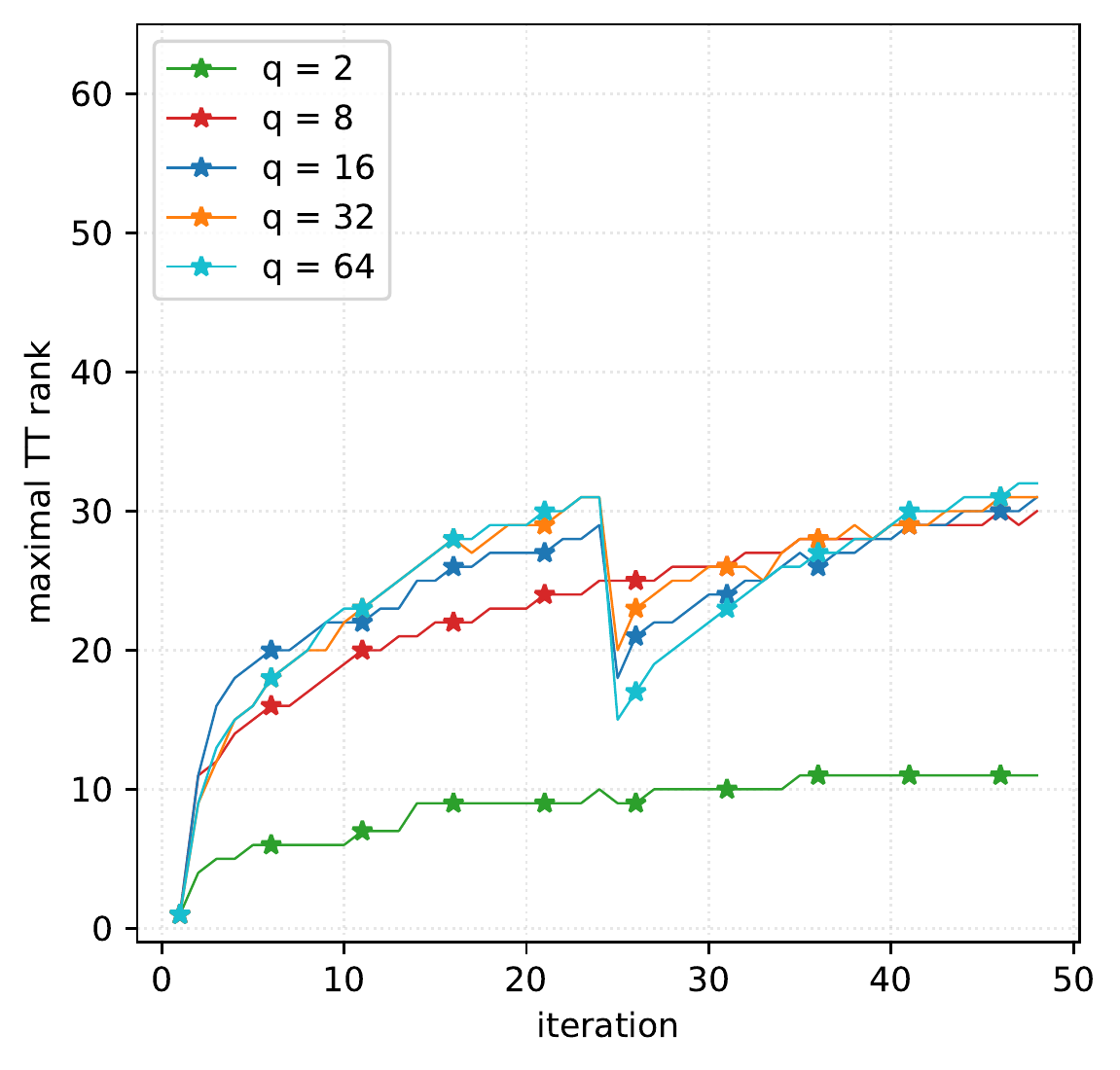}\label{fig:LapP_v-rank_1e-2}}
		\quad
		\subfloat[\footnotesize Maximal TT-rank of the last Krylov vector for $\tau = 10^{-8}$]{\includegraphics[scale=0.45, width=0.4\linewidth, height=0.4\linewidth]{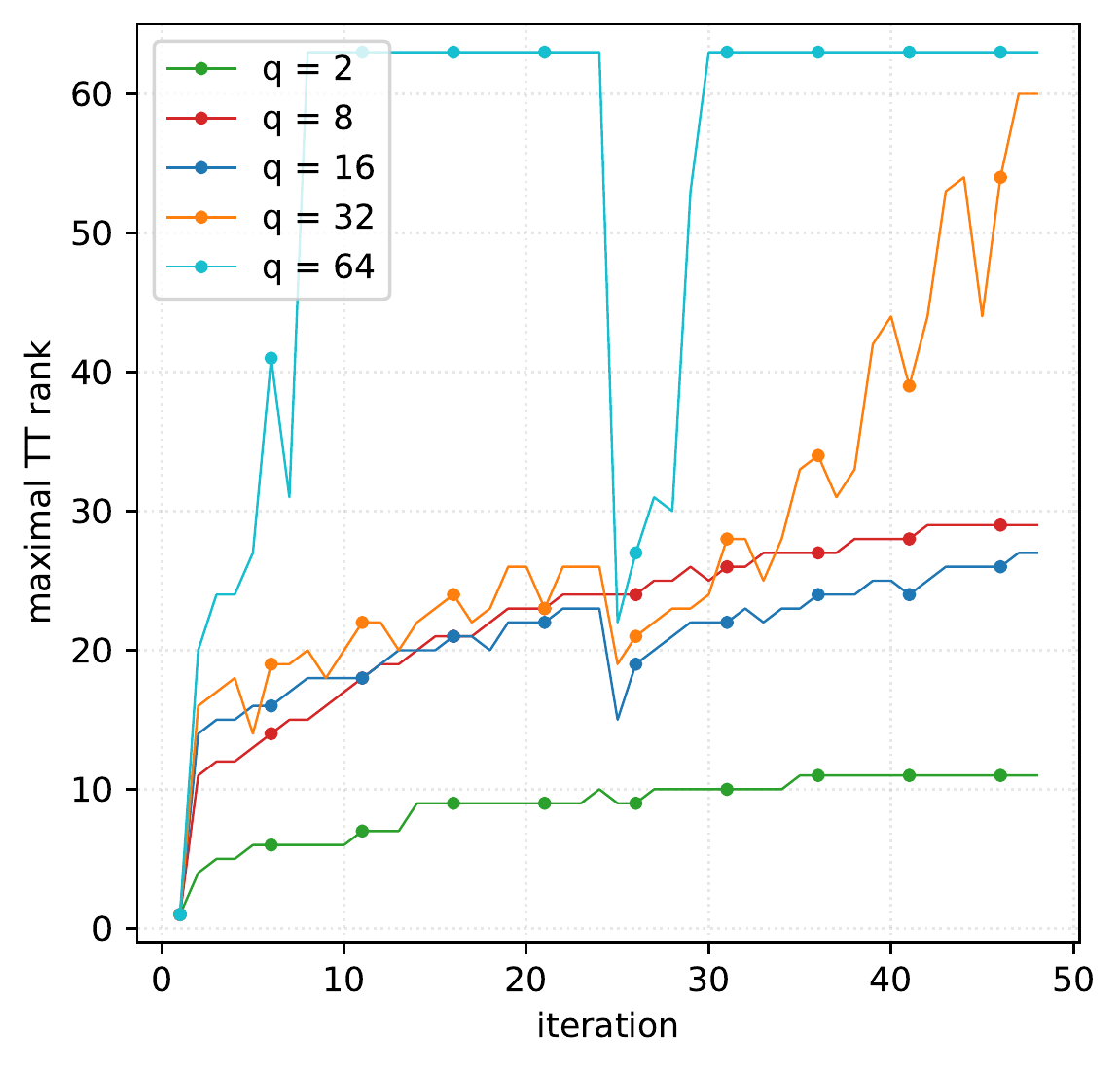}\label{fig:LapP_v-rank_1e-8}}
		\caption{3-d Poisson problem, comparing preconditioners}
\end{figure}

\begin{figure}[!htb]	
		\centering
		\ContinuedFloat
	\subfloat[\footnotesize{Maximal TT-rank of the iterative solution for $\tau = 10^{-2}$}]{\includegraphics[scale=0.45, width=0.4\linewidth, height=0.4\linewidth]{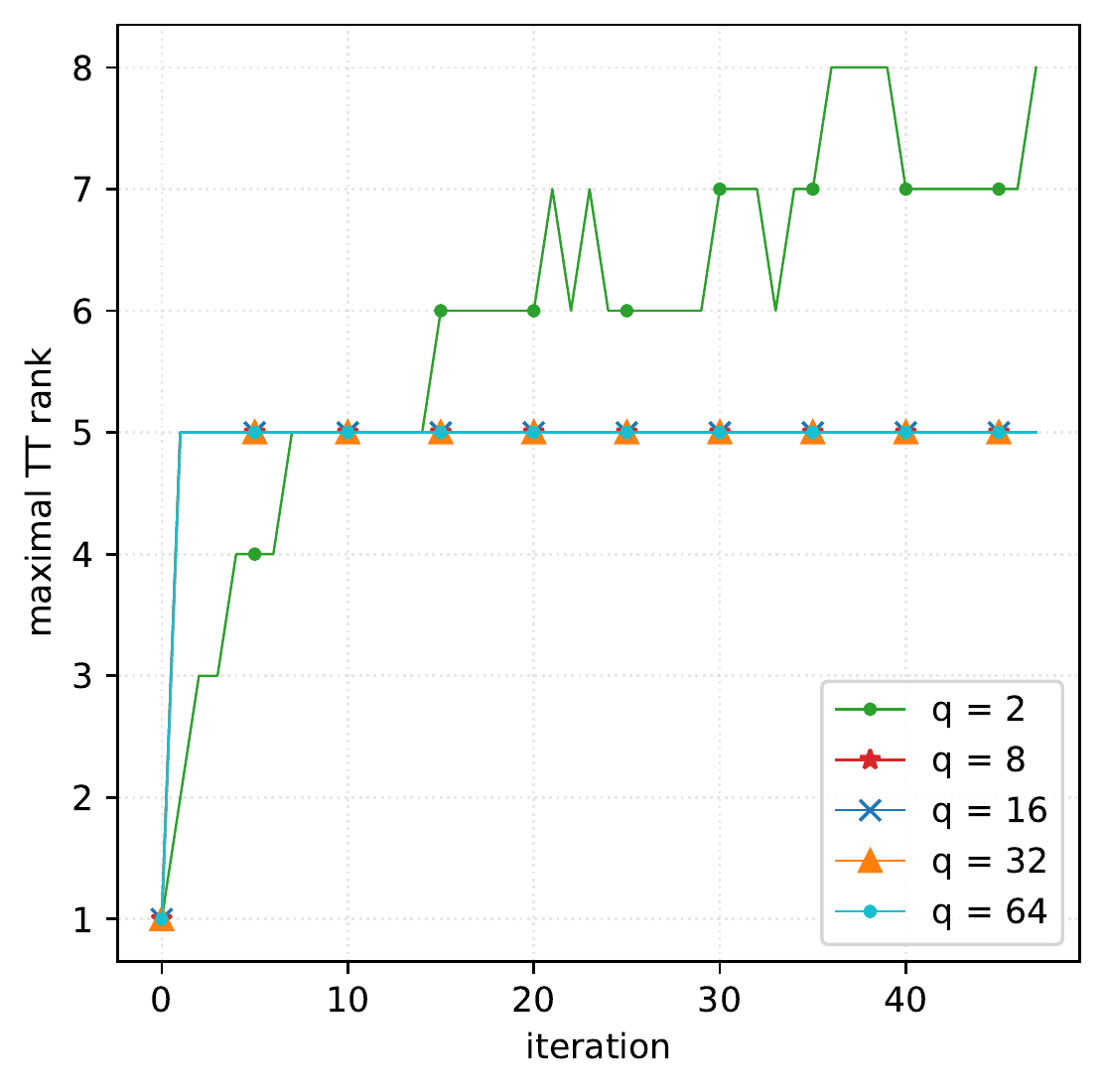}\label{fig:LapP_x-rank_1e-2}}
	\quad
	\subfloat[\footnotesize{Maximal TT-rank of the iterative solution for $\tau = 10^{-8}$}]{\includegraphics[scale=0.45, width=0.4\linewidth, height=0.4\linewidth]{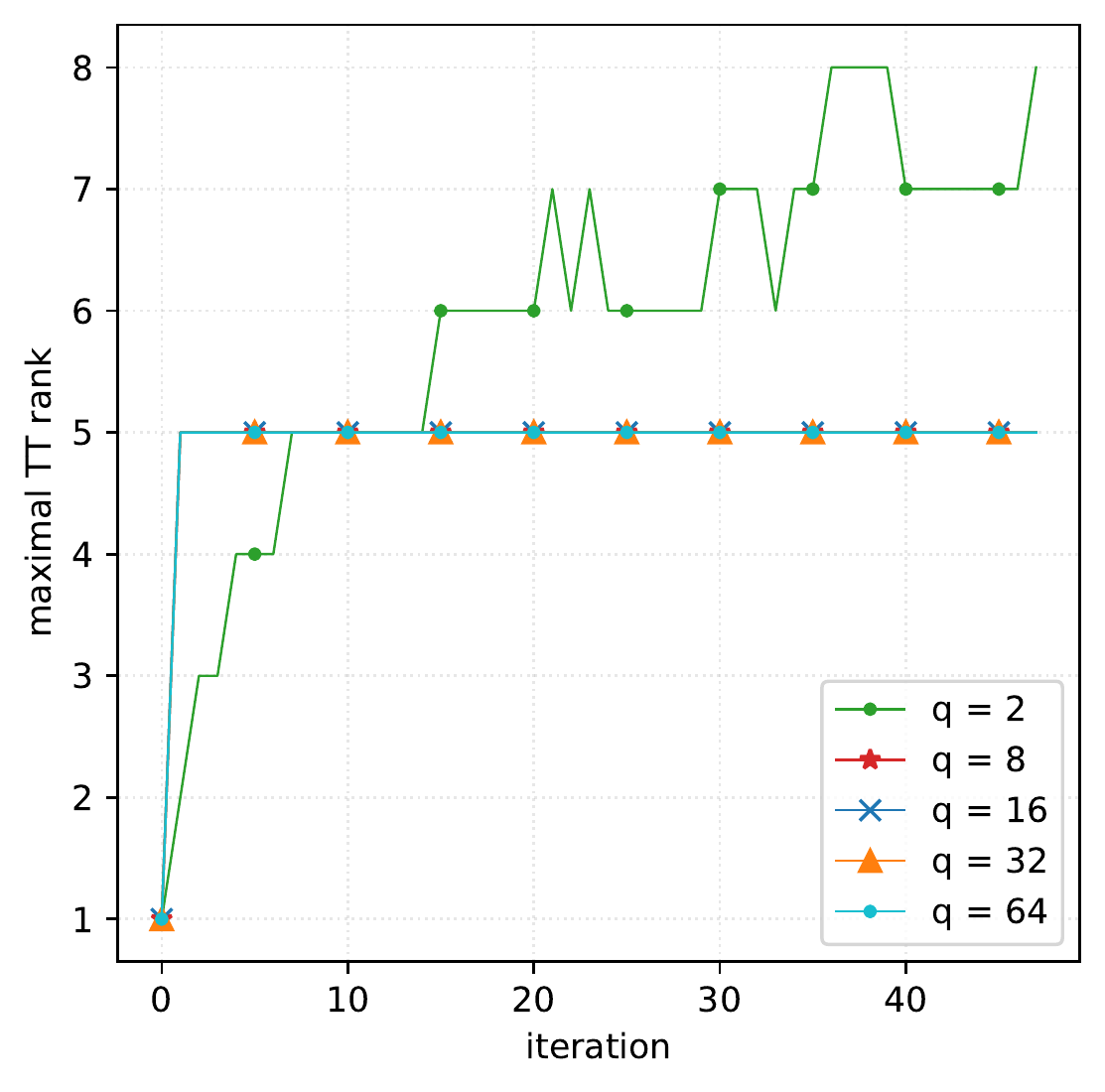}\label{fig:LapP_x-rank_1e-8}}
	\caption{3-d Poisson problem, comparing preconditioners}
	\label{fig:LapP}
\end{figure}
\clearpage
\subsection{Multiple right-hand sides: a focus on eigenvectors\label{app:eig}}
In this appendix we study further the convergence of a multiple right-hand side problem. Indeed comparing the convergence history of the convection-diffusion problem, see Subsection~\ref{sssec:CD} and of the multiple right-hand side convection-diffusion problem discussed in~\ref{sec4:mCD}, notice that the number of iterations necessary to converge is equal, $5$ in both cases. This appendix explains the causes of the phenomenon. 
\par As we already explained, given a (tensor) linear system $\ten{A}\ten{x} = \ten{b}$ and the null tensor as initial guess, at the $k$-th iteration GMRES minimizes the norm of the residual $\ten{r}_{k} = \ten{A}\ten{x} - \ten{b}$ on the Krylov space of dimension $k$ defined as
\[
	\mathcal{K}_k(\ten{A}, \ten{b}) =\text{span}\bigl\{\ten{b}, \ten{A}\ten{b}, \dots, \ten{A}^{k-1}\ten{b}\bigr\}
\]
where $\ten{A}^{h}$ is obtained from $h$ contractions over the indexes $(1,3,\dots, 2d-1, 2d+1)$ of tensor operator $\ten{A}$.
If $\ten{b}$ is equal to $\ten{e}_i$ an eigenvector of the tensor operator $\ten{A}$, then the Krylov space writes
\begin{equation}\nonumber
	\begin{split}
		\mathcal{K}_k(\ten{A}, \ten{b}) &=\text{span}\bigl\{\ten{b}, \ten{A}\ten{b}, \dots, \ten{A}^{k-1}\ten{b}\bigr\}\\
		&=\text{span}\bigl\{\ten{e}_i, \lambda_i\ten{e}_i, \dots, \lambda^{k-1}_i\ten{e}_i\bigr\}\\
		&=\text{span}\bigl\{\ten{e}_i\}
	\end{split}
\end{equation}
where $\lambda_i$ is the $i$-th eigenvalue of $\ten{A}$. The Krylov space dimension is equal to $1$, i.e., the number of eigenvector $\ten{e}_i$ necessary to express the right-hand side. Theoretically the number of iterations necessary to converge, i.e., the dimension of the Krylov space where the exact solution lives, is By linearity equal to the number of eigenvector necessary to express the right-hand side as their linear combination. In the problems presented in~\ref{sec4:mLap} and~\ref{sec4:mCD}, we add a random generated tensor to the chosen right-hand side. Comparing the results a single right-hand side and multiple ones for the convection-diffusion problem, we may conclude that the introduced error has not increased the number of eigenvectors, for the tolerance chosen. 
\par Let now consider a more peculiar problem with two right-hand side, living in subspaces generated by different eigenvector. More  in details one right-hand side belongs to the subspace generated by a single eigenvector, while the other to the subspace generated by $j$ different eigenvector. Thanks to our previous argument, theoretically the two systems converge independently with a different number of iterations, one for the first and $j$ for the second. When we solve the two systems together, we expect the ``all-in-one'' system to converge as the slowest one, i.e., as the slowest converging one.
 Let $\ten{e_1}, \dots, \ten{e}_{j+1}$ be the first $j+1$ different eigenvectors of the $3$-dimensional discrete Laplacian $-\ten{\Delta}_3$. We consider the two following linear systems
\begin{align}
	-\ten{\Delta}_3 \ten{y}_1 &= \ten{e}_1\label{eq:mEig1}\\
	-\ten{\Delta}_3 \ten{y}_2 &= \sum_{\ell = 2}^{j+1}\ten{e}_\ell\label{eq:mEig2}. 
\end{align}
As described in Subesection~\ref{ssec:Ai1}, we define the ``all-in-one'' linear system with the `diagonal' tensor operator $\ten{A}$ and the ``all-in-one'' right-hand side $\ten{b}$. TT-GMRES is used to solve this ``all-in-one'' system for $j=10$, for a grid step {dimension} equal to $n\in\{63, 127, 255\}$, without preconditioner, with tolerance $\varepsilon$ and rounding accuracy $\delta$ equal to $10^{-5}$, no restart and a maximum of $50$ iterations. Figure~\ref{fig:meig_CH} shows the convergence history of the problem with the three different grid {dimensions}. Since there is no preconditioner the convergence is kind of slow, if compared with Figure~\ref{fig:mLap_CH}. At the same time, the TT-ranks growth not too quickly, because both there is not a random generated error and there are just two right-hand side.  
\begin{figure}[!htb]
	\centering
	\subfloat[Convergence history]{\includegraphics[scale = 0.45, width=0.33\linewidth, height=0.33\linewidth]{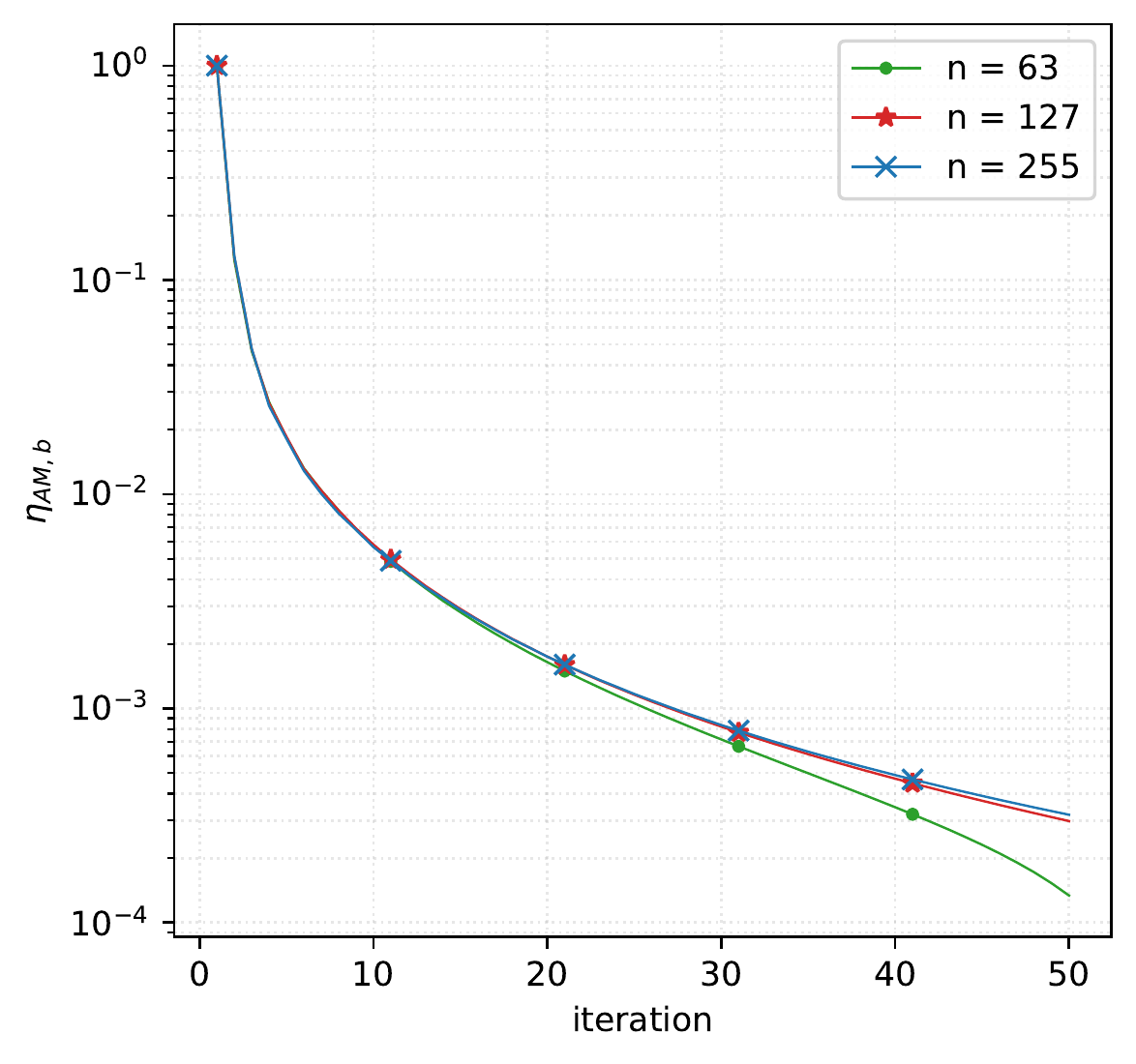}\label{fig:meig_CH}}
	\quad
	\subfloat[Maximal TT-rank of the last Krylov vector]{\includegraphics[scale=0.45, width=0.33\linewidth, height=0.33\linewidth]{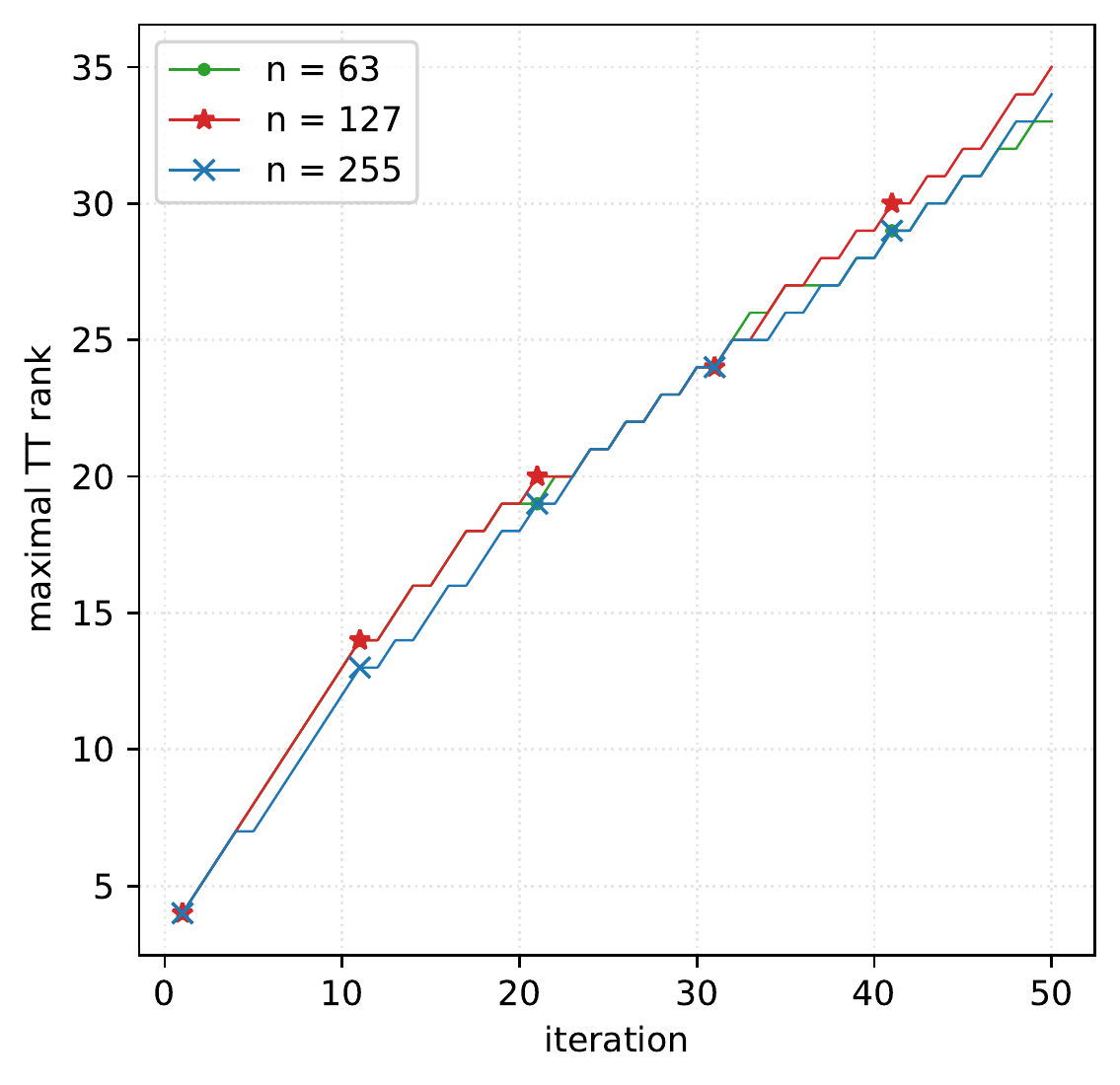}\label{fig:meig_v-rank}}
	\caption{4-d multiple right-hand side problem with eigenvector} 
	\label{fig:mEig}
\end{figure}
The residual curve associated with Equation~\ref{eq:mEig1} in blue is almost flatten for all the grid {dimensions}, suggesting that GMRES almost immediately minimized it completely. On the other side the orange curve of problem~\eqref{eq:mEig2} residual decreases slowly, meaning that minimizing its norm requires more steps.  Lastly in all the plots the ``all-in-one'' residual curve follows exactly the eigenvector sum problem residual, confirming that the general convergence of the ``all-in-one'' system is decided by the slowest converging system, that is our intuition, confirmed in for all the grid {dimensions} in Figure~\ref{fig:mEig_res}. As last remark the ``all-in-one'' doesn't converge in $10$ iterations because of the rounding effect, which slows down the convergence.

\begin{figure}[!htb]
	\centering
	\subfloat[Residual norm for $n = 63$]{\includegraphics[scale =0.3, width=0.31\linewidth, height=0.31\linewidth]{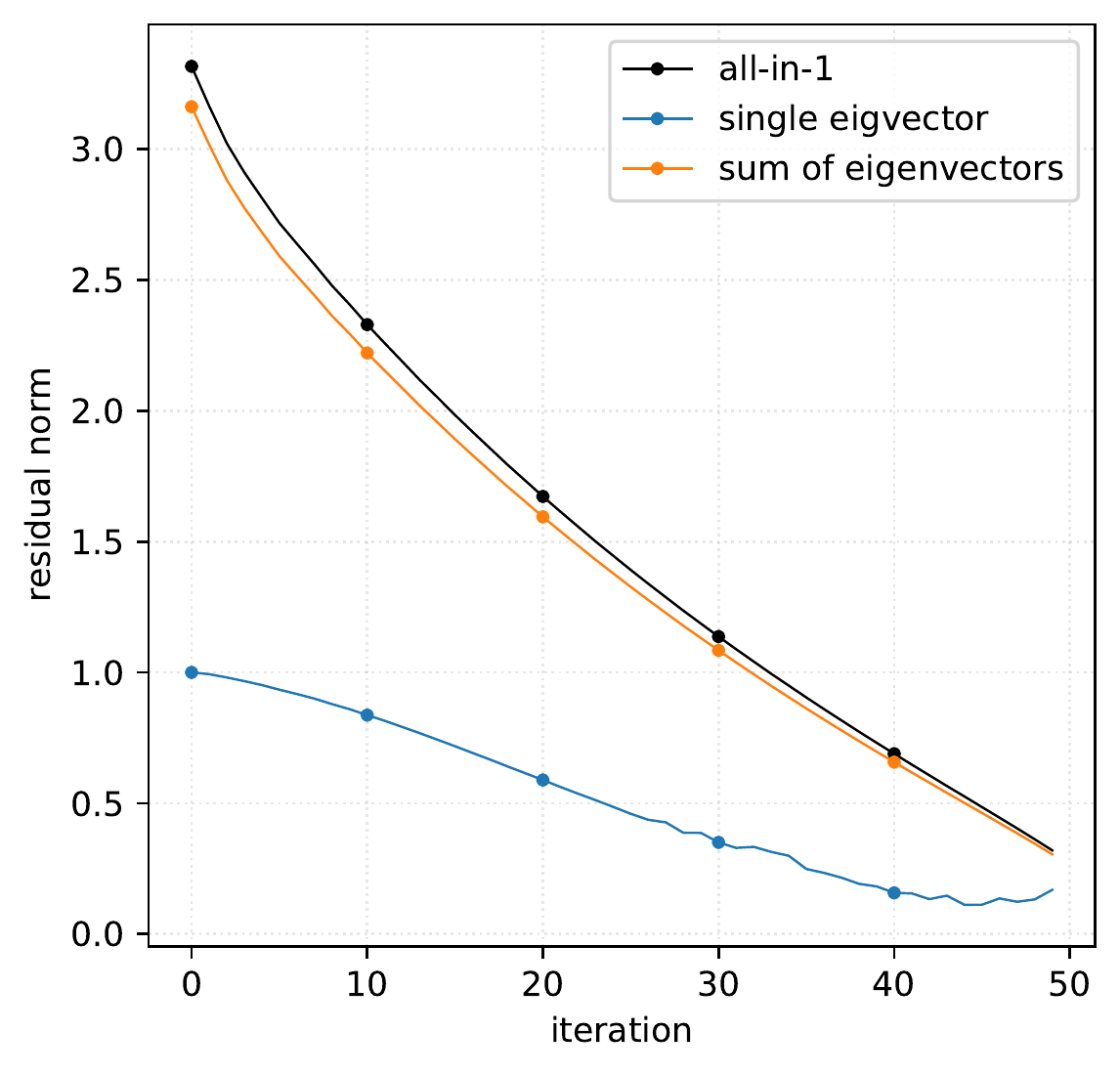}\label{fig:meig_res_64}}
	\quad
	\subfloat[Residual norm for $n = 127$]{\includegraphics[scale =0.3, width=0.31\linewidth, height=0.31\linewidth]{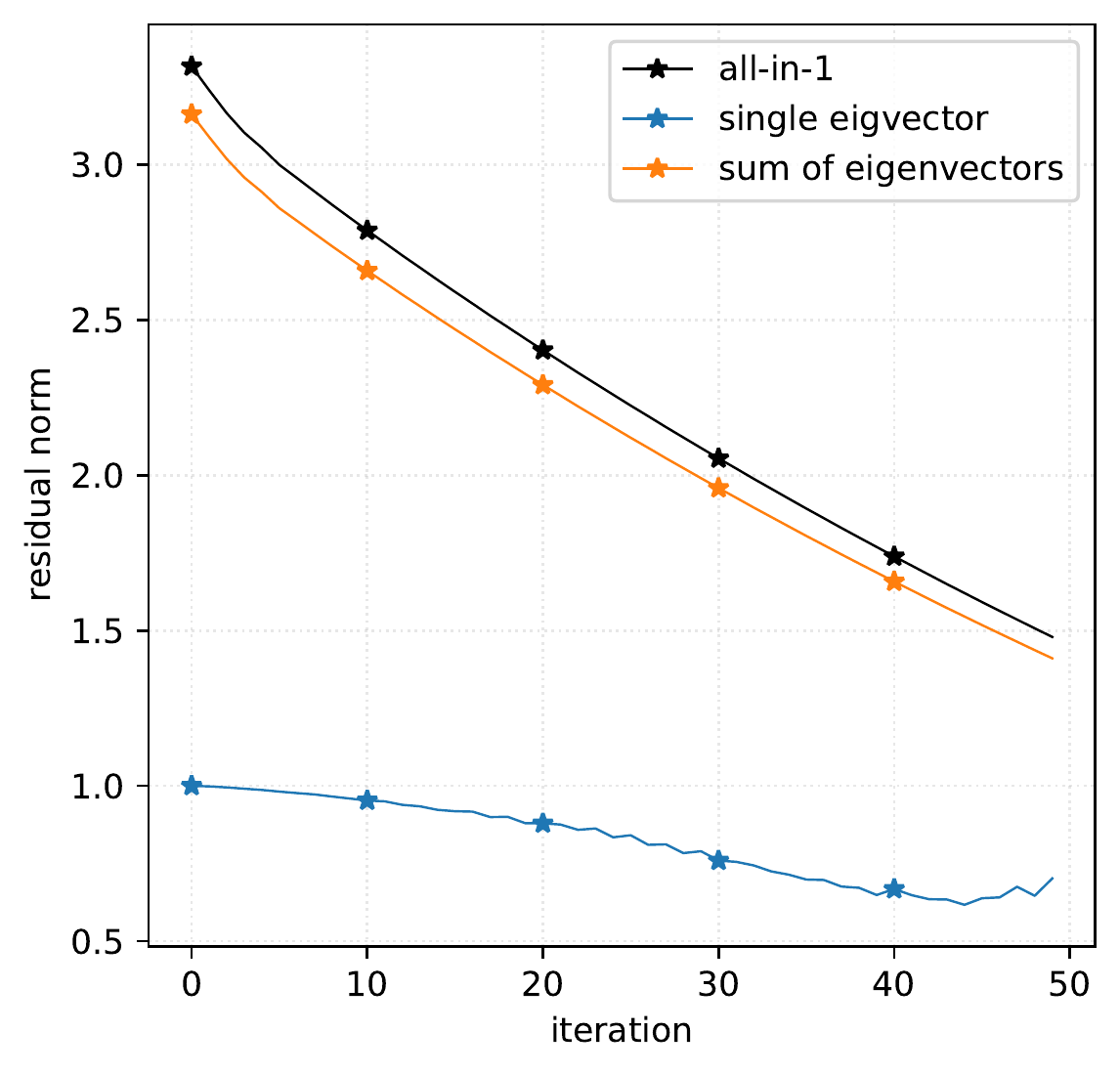}\label{fig:meig_res_128}}
	\quad
	\subfloat[Residual norm for $n = 255$]{\includegraphics[scale =0.3, width=0.31\linewidth, height=0.31\linewidth]{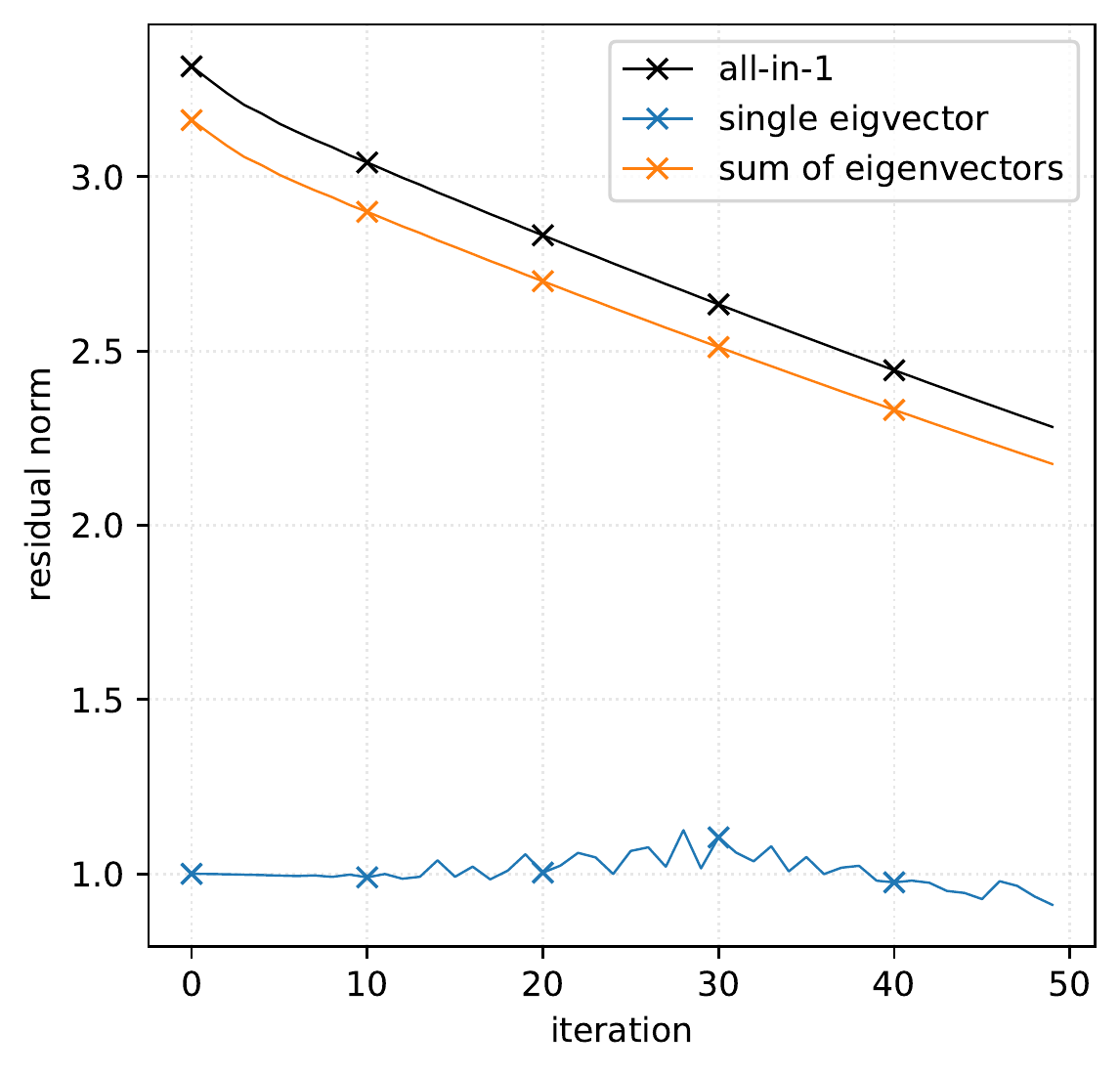}\label{fig:meig_res_256}}
	\caption{4-d multiple right-hand side problem, residual comparison} 
	\label{fig:mEig_res}
\end{figure}

\subsection{{Further details on the ``all-in-one'' system}}  \label{app:tt}
This appendix describes in details the construction in TT-format of the ``all-in-one'' system. As conclusion we provides a corollary to Proposition~\ref{prop:eta_Ab_m}. 

\par As previously stated, given a tensor $\ten{a}\in\R^{n_1\times \cdots \times n_d}$ in TT-format with TT-cores $\coreTT{a}_k\in\R^{r_{k-1}\times n_k\times r_k}$, we denote by $\ten{a}^{[k, i_k]}$ the $i_k$-th slice with respect to mode $k$, which in TT-format writes as
	\[
		\ten{a}^{[k, i_k]} = \coreTT{a}_1\cdots\coreTT{a}_{k-1}\matcoreTT{A}_k(i_k)\coreTT{a}_{k+1}\cdots \coreTT{a}_d
	\]
with $\matcoreTT{A}_k(i_k)\in\R^{r_{k-1}\times r_k}$. Since henceforth we will take slice only with respect to the first mode, instead of writing $\ten{a}^{[1, i_1]}$ for the $i_1$-th slice on the first mode we will simply write $\ten{a}^{[i_1]}$. Similarly $\ten{A}^{[\ell]}$ denotes the $(\ell, \ell)$-th slice of $\ten{A}\in\R^{(n_1\times n_1)\times \dots \times (n_d\times n_d)}$ with respect to the first two modes.

\par We start constructing the elements of the ``all-in-one'' system from the $p$ individual right-hand sides. Let $\ten{b}_\ell\in\R^{n_1\times \cdots\times n_d}$ be a TT-vector for every $\ell\in\{1,\dots, p\}$ with TT-cores $\coreTT{b}_{\ell,k}\in\R^{s_{\ell, k}\times n_k\times s_{\ell, k+1}}$ for every $k\in\{1,\dots,d\}$ with $s_1 = s_{d+1} = 1$, i.e.,
\begin{equation}
\label{def:b_ell}
\ten{b}_\ell = \coreTT{b}_{\ell,1}\cdots \coreTT{b}_{\ell,d}
\end{equation}
and its $(i_1, \dots, i_d)$ element writes
	\[
		\ten{b}_\ell (i_1,\dots, i_d) = \matcoreTT{B}_{\ell, 1}(i_1)\cdots \matcoreTT{B}_{\ell, d}(i_d)
	\]
with $\matcoreTT{B}_{\ell, k}(i_k)\in\R^{s_{\ell, k}\times s_{\ell, k+1}}$, $\matcoreTT{B}_{\ell, 1}(i_1)\in\R^{1\times s_{\ell,2}}$ and $\matcoreTT{B}_{\ell, d}(i_d)\in\R^{s_{\ell, d}\times 1}$. For simplicity we impose $s_k = s_{\ell, k}$ for every $\ell\in\{1,\dots,p\}$ and $k\in\{1,\dots, d\}$.

\begin{remark}
	This assumption on the TT-rank of $\ten{b}_\ell$ is not binding. Indeed setting $s_k = \max_{h\in\{1,\dots, p\}} s_{h,k}$, then each core tensor $\coreTT{b}_{h,k}$ of mode sizes $(s_{h, k}, n_h, s_{h, k+1})$ can be extended with zeros to $(s_{k}, n_h, s_{k+1})$ 
\end{remark}

	We want to construct a tensor $\ten{b}\in\R^{p\times n_1\times \cdots \times n_d}$ such that its $\ell$-th slice with respect to the first mode of $\ten{b}$ is $\ten{b}_\ell$, i.e., $\ten{b}^{[\ell]} = \ten{b}_\ell.$ As consequence, the $k$-th TT-core of $\ten{b}$ is $\coreTT{b}_k\in\R^{ps_{k}\times n_k\times ps_{k+1}}$ such that
		%
	\[
	\matcoreTT{B}_k(i_k) = \begin{bmatrix} \matcoreTT{B}_{1,k}(i_k)&&\\
	&\ddots&\\
	&& \matcoreTT{B}_{p,k}(i_k)
	\end{bmatrix}\;\in\R^{ps_k\times ps_{k+1}}
	\]
for $k\in\{1,\dots, d-1\}$, while $\coreTT{b}_{d}\in\R^{ps_{d-1}\times n_d\times 1}$ and $\coreTT{b}_0\in\R^{1\times p \times ps_1}$ are 
	\[
	\matcoreTT{B}_d(i_d) = \begin{bmatrix} \matcoreTT{B}_{1,d}(i_d)\\
	\vdots\\
	\matcoreTT{B}_{p,d}(i_d)
	\end{bmatrix}\; \in\R^{ps_d\times 1}\qquad\text{and}\qquad \matcoreTT{B}_0(\ell) = \begin{bmatrix} 0 \cdots 1 \cdots 0  \end{bmatrix}\; \in\R^{1\times ps_1}
	\]
with the $\ell$-th component of $\matcoreTT{B}_0(\ell)$ beign the only non-zero element. The TT-expression of $\ten{b}$ is
\begin{equation}
\label{def:b}
\ten{b} = \coreTT{b}_0\coreTT{b}_1\cdots\coreTT{b}_d.
\end{equation}

By construction we have that the $\ell$-th slice of $\ten{b}$ with respect to mode $1$ is 
	\[
	\begin{split}
	\ten{b}^{[\ell]} & =\matcoreTT{B}_0(\ell)\coreTT{b}_1\cdots \coreTT{b}_d \\
	&=\coreTT{b}_{\ell,1}\cdots \coreTT{b}_{\ell,d} \\
	&= \ten{b}_\ell.
	\end{split}
	\]

\par We illustrate now the construction of the ``all-in-one'' system tensor linear operator. Let $\ten{C},\ten{G}\in\R^{(n_1\times n_1)\times \cdots \times (n_d\times n_d)}$ be two TT-matrix with $k$-th TT-core $\coreTT{c}_k\in\R^{r_{k}\times n_k \times n_k \times r_{k+1}}$ and $\coreTT{g}_k\in\R^{q_{k}\times n_k \times n_k \times q_{k+1}}$ for $k\in\{1,\dots, d\}$ with $q_1 = r_1 = r_{d+1} = q_{d+1} = 1$, whose TT-expression is
	\begin{equation}
	\label{eq:C:CG}
		\ten{C} = \coreTT{c}_1\cdots 		\coreTT{c}_d\qquad\text{and}\qquad
		\ten{G} = \coreTT{g}_1\cdots \coreTT{g}_d.
	\end{equation}
Given a diagonal matrix $D = \text{diag}(\alpha_1,\dots, \alpha_p)$, we define $\ten{A}\in\R^{(p\times p)\times(n_1\times n_1)\times \cdots\times (n_d\times n_d)}$ as 
	\begin{equation}
	\label{def:A}
		\ten{A} = \mathbb{I}_p\otimes \ten{C} + D\otimes \ten{G}
	\end{equation}
Then the expression of $\coreTT{a}_k\in\R^{(r_k+q_k)\times n_k\times n_k \times (r_{k+1} + q_{k+1})}$ the $k$-th TT-core of $\ten{A}$ is
	\begin{equation}
	\label{eq:C:Acores}
		\matcoreTT{A}_k(i_k, j_k) = \begin{bmatrix}
				\matcoreTT{C}_k(i_k, j_k) & \0\\
				\0 & \matcoreTT{G}_k(i_k, j_k)
				\end{bmatrix}
		\qquad\text{and}\qquad 
		\matcoreTT{A}_d(i_d, j_d) =\begin{bmatrix}
				\matcoreTT{C}_d(i_d, j_d)\\
				\matcoreTT{G}_d(i_d, j_d)
				\end{bmatrix}
	\end{equation}

for every $i_k, j_k\in\{1,\dots, n_k\}$ and $k\in\{1,\dots, d\}$. The first TT-core $\coreTT{a}_0\in\R^{1\times p\times p \times 2}$ writes
	\begin{equation}
	\label{eq:C:Acores1}	
		\matcoreTT{A}_0(\ell, m) =  \delta_{\ell, m}a_{\ell}\qquad\text{with}\qquad a_\ell = 
		\begin{bmatrix}
		1 & \alpha_{\ell}
		\end{bmatrix}  
	\end{equation}
with $\delta_{\ell, m}$ the Kronecker delta, for $\ell, m\in\{1,\dots, p\}$.
The final TT-expression of $\ten{A}$ is
	\[
		\ten{A} = \coreTT{a}_0\coreTT{a}_1\cdots \coreTT{a}_d.
	\]
Remark now that $\ten{A}^{(\ell, m)}$ the $(\ell, m)$-th slice with respect to mode $1$ of $\ten{A}$ is 
	\[
		\ten{A}^{[\ell, m]} =\matcoreTT{A}_0(\ell, m)  \coreTT{a}_1\cdots\coreTT{a}_d= \delta_{\ell,m}a_\ell  \coreTT{a}_1\cdots\coreTT{a}_d.
	\]
If $\ell \ne m$, then $\ten{A}^{[\ell, m]} = \0$. On the other side, if $\ell$ and $m$ are equal, then
	\[
		\ten{A}^{[\ell, \ell]} = \mathbb{I}(\ell, \ell)a_\ell\,\coreTT{a}_1\cdots\coreTT{a}_d = \coreTT{c}_1\cdots \coreTT{c}_d + \alpha_{\ell}\, \coreTT{g}_1\cdots \coreTT{g}_d = \ten{C} + \alpha_{\ell}\, \ten{G}.
	\]

\par Let consider $\ten{A}$ and $\ten{b}$ as defined in Equation~\eqref{def:A} and~\eqref{def:b}, given $\ten{x}\in\R^{p\times n_1\times \dots \times n_d}$ and define the new vector
	\[
		\ten{r} = \ten{A}\ten{x} - \ten{b}.
	\]
We want to prove that $\ten{r}^{[\ell]}$ the $\ell$-the slice with respect to the first mode of $\ten{r}$ is equal to the difference of the $\ell$-th slices, i.e.,
	\begin{equation}
	\label{eq:C:ths}
		\ten{r}^{[\ell]} = \ten{A}^{[\ell, \ell]}\ten{x}^{[\ell]} - \ten{b}_\ell = (\ten{C} + \alpha_{\ell}\ten{G})\ten{x}^{[\ell]} - \ten{b}_\ell 
	\end{equation}
	since the $(\ell, \ell)$-th slice of $\ten{A}$ is $\ten{C} + \alpha_{\ell}\ten{G}$ for every $\ell\in\{1,\dots, p\}$. Remark that the $\ell$-th slice of $\ten{b}$ is $\ten{b}_\ell$ by construction. As consequence, the Equation~\eqref{eq:C:ths} is true if we show that the $\ell$-th slice of the contraction between $\ten{A}$ and $\ten{x}$ is equal to the contraction of their $\ell$-th slices, i.e., 
	\[
		(\ten{A}\ten{x})^{[\ell]} = \ten{A}^{[\ell, \ell]}\ten{x}^{[\ell]}.
	\]
 
	\begin{lemma}
		\label{lm:2}
		Given $\ten{A}$, $\ten{C}$, $\ten{G}$ as in Equations~\eqref{eq:C:CG} and~\eqref{def:A}, let
		$\ten{x}\in\R^{p\times n_1\times \dots \times n_d}$ be a $(d+1)$-order tensor. 
		Then the $\ell$-th slice of $\ten{A}\ten{x}$ is equal to the product of their $\ell$-th slices, i.e.
		\[
			(\ten{A}\ten{x})^{[\ell]} = \ten{A}^{[\ell, \ell]}\ten{x}^{[\ell]} = (\ten{C} + \alpha_{\ell}\ten{G})\ten{x}^{[\ell]}.
		\]
		
		Defined  and $\ten{w} = (\ten{C}+\alpha_{\ell} \ten{G})\ten{x}^{[\ell]}$, then $\ten{y}^{[\ell]}$ the $\ell$-th slice of $\ten{y}$ with respect to mode $1$ is equal to $\ten{w}$, i.e., 
		\[
			\ten{y}^{[\ell]} = \ten{w}.
		\]
	\end{lemma}
	\begin{proof}
		Let $\coreTT{x}_k\in\R^{t_{k}\times n_k\times t_{k+1}}$ be the $k$-th TT-core of $\ten{x}$ for $k\in\{1,\dots, d\}$ with $t_{d+1} = 1$ and $\ten{x}_0\in\R^{1\times p\times t_{1}}$, getting
		\begin{equation}
		\label{def:x}
			\ten{x} = \coreTT{x}_0\coreTT{x}_1\cdots \coreTT{x}_d
		\end{equation}
		Set $\ten{y} =\ten{A}\ten{x}$, then by the property of TT-contraction, we get
		\[
		\begin{split}
			\ten{y}(\ell, i_1,\dots, i_d) 
					&= \sum_{j_0, j_1,\dots, j_d = 1}^{p, n_1,\dots, n_d}\matcoreTT{A}_0(\ell,j_0)\matcoreTT{A}_1(i_1,j_1)\cdots \matcoreTT{A}_d(i_d, j_d)\matcoreTT{X}_0(j_0)\matcoreTT{X}_1(j_1)\cdots \matcoreTT{X}_d(j_d)\\
					&= \sum_{j_0, j_1,\dots, j_d = 1}^{p, n_1,\dots, n_d} \bigl(\matcoreTT{A}_0(\ell, j_0)\kron \matcoreTT{X}_0(j_0)\bigr)\bigl(\matcoreTT{A}_1(i_1, j_1)\otimes_{\textsc{K}}\matcoreTT{X}_1(j_1)\bigr)\cdots \bigl(\matcoreTT{A}_d(i_d, j_d)\otimes_{\textsc{K}}\matcoreTT{X}_d(j_d)\bigr)\\
					&= \matcoreTT{Y}_0(\ell)\matcoreTT{Y}_1(i_1)\cdots \matcoreTT{Y}_d(i_d)
		\end{split}
		\]
		where $\matcoreTT{Y}_k(i_k) \in\R^{r_{k}t_k\times r_{k+1}t_{k+1}}$ with $r_1 = 1$, $Y_{d}(i_d) \in\R^{r_d t_d\times 1}$ and $\matcoreTT{Y}_0(\ell) \in\R^{1\times t_1}$ defined as
		\begin{equation}
		\label{eq:C:Ycores}
			\begin{split}
				\matcoreTT{Y}_0(i_0) & = \sum_{j_0 = 1}^{p}\matcoreTT{A}_0(\ell, j_0)\kron \matcoreTT{X}_0(j_0)\\
				\matcoreTT{Y}_k(i_k) &= \sum_{j_k = 1}^{n_k}\matcoreTT{A}_k(i_k, j_k)\otimes_{\textsc{K}}\matcoreTT{X}_k(j_k)\qquad\text{for}\qquad k\in\{2,\dots, d\}.
			\end{split}
		\end{equation}
		Remark now that in the expression $\matcoreTT{Y}_0(\ell)$, the quantity $\matcoreTT{A}_0(\ell,j_0)$ is actually a vector of $2$ elements  times $\delta_{\ell,j_0}$, so we replace the Kronecker product with a simple scalar-matrix product, writing
		\begin{equation}
			\begin{split}
			\label{eq:C:Ycore1}
				\matcoreTT{Y}_0(\ell) 	&= \sum_{j_0 = 1}^{p} \delta_{\ell,j_0}a_\ell\kron \matcoreTT{X}_0(j_0) \\
							&= \matcoreTT{A}_0(\ell,\ell)\kron \matcoreTT{X}_0(\ell).
			\end{split}
		\end{equation}
		
		Let $\ten{x}^{[\ell]}$ be the $\ell$-th slice with respect to the first mode of $\ten{x}$, whose TT-expression is
		\begin{equation*}
			\ten{x}^{[\ell]} = \matcoreTT{X}_0(\ell)\coreTT{x}_1\cdots \coreTT{x}_d.
		\end{equation*}
		We define $\ten{x}^{[\ell]}$ TT-cores to get the clean expression as 
		\begin{align*}
		\coreTT{x}_{\ell,1} &= \matcoreTT{X}_0(\ell)\ten{x}_1 \;\in\R^{1\times n_1\times t_2}\\
		\coreTT{x}_{\ell,k} &= \coreTT{x}_{k}\in\R^{t_k\times n_k\times t_{k+1}}\qquad\text{for}\qquad k\in\{2,\dots, d\}
		\end{align*}
		getting
		\[
			\ten{x}^{[\ell]} = \coreTT{x}_{\ell,1}\cdots\coreTT{x}_{\ell,d}.
		\]
		To compute $\ten{w} = (\ten{C}+\alpha_{\ell}\ten{G})\ten{x}^{[\ell]}$, we need to clarify the structure of the $k$-th TT-core of $\ten{H} = (\ten{C} + \alpha_{\ell}\ten{G})$, given by the TT-sum rule. Therefore the $k$-th TT-core is $\coreTT{h}_k\in\R^{(r_k + q_k)\times n_k\times n_k \times (r_{k+1} + q_{k+1})}$ such that
		\begin{equation}
			\label{eq:C:CGcores}
			\matcoreTT{H}_k(i_k, j_k) = \begin{bmatrix}
			\matcoreTT{C}_k(i_k, j_k) & \0 \\
			\0 & \matcoreTT{G}_k(i_k, j_k)
			\end{bmatrix}\qquad\text{and}\qquad \matcoreTT{H}_d(i_d, j_d) = \begin{bmatrix}
			\matcoreTT{C}_d(i_d, j_d)\\
			\matcoreTT{G}_d(i_d, j_d)
			\end{bmatrix}
		\end{equation}
		for $i_k, j_k\in\{1,\dots, n_k\}$ and $k\in\{2,\dots, d\}$ with $r_{d+1} + q_{d+1} = 1$. The first TT-core $\coreTT{h}_{1}\in\R^{1\times n_1\times n_1\times (r_2+q_2)}$ is
		\[
			\matcoreTT{H}_1(i_1, j_1) = \begin{bmatrix}
			\matcoreTT{C}_1(i_1, j_1)\\
			\alpha_{\ell}\matcoreTT{G}_1(i_1, j_i)
			\end{bmatrix}
		\]
		for $i_1, j_1\in\{1,\dots, n_1\}$. 
		\par Compute the $(i_1,\dots, i_d)$ element of $\ten{w} = \ten{H}\ten{x}^{[\ell]} =  (\ten{C}+\alpha_{\ell}\ten{G})\ten{x}^{[\ell]}$ as
		\[
		\begin{split}
			\ten{w}(i_1,\dots, i_d) 
			&= \sum_{j_1,\dots, j_d}^{n_1\dots n_d}\matcoreTT{H}_1(i_1, i_1)\cdots \matcoreTT{H}_d(i_d, j_d)\matcoreTT{X}_{\ell,1}(j_1)\cdots \matcoreTT{X}_{\ell,d}(j_d)\\
			&= \sum_{j_1,\dots, j_d}^{n_1,\dots, n_d} \bigl(\matcoreTT{H}_1(i_1, j_1))\kron \matcoreTT{X}_{\ell,1}(j_1)\bigr)\cdots \bigl(\matcoreTT{H}_d(i_d, j_d)\kron\matcoreTT{X}_{\ell,d}(j_d)\bigr)\\
			&= \matcoreTT{W}_1(i_1)\matcoreTT{W}_2(i_2)\cdots \matcoreTT{W}_d(i_d)
		\end{split}
		\]
		where $\matcoreTT{W}_k(i_k)\in\R^{r_kt_k\times r_{k+1}t_{k+1}}$, $W_d\in\R^{r_dt_d \times 1}$ and $\matcoreTT{W}_1 \in\R^{1\times r_2t_2}$ are defined as
		\begin{align}
		\matcoreTT{W}_1(i_1) &= \sum_{j_1 = 1}^{n_1}\matcoreTT{H}_1(i_1, j_1)\kron \matcoreTT{X}_0(\ell)\matcoreTT{X}_1(j_1)\\
		\matcoreTT{W}_k(i_k) &= \sum_{j_k=1}^{n_k}\matcoreTT{W}_k(i_k, j_k)\kron\matcoreTT{X}_k(j_k)\qquad\text{for}\qquad k\in\{2,\dots, d\}.
		\end{align}
		
		Comparing Equation~\eqref{eq:C:CGcores} and~\eqref{eq:C:Acores}, we get $\coreTT{h}_k= \coreTT{a}_k$ for $k\in\{2,\dots, d\}$. As consequence $\matcoreTT{W}_k(i_k)$ writes as
		
		\[
			\matcoreTT{W}_k(i_k) = \sum_{j_k=1}^{n_k}\matcoreTT{H}_k(i_k, j_k)\kron\matcoreTT{X}_k(j_k) = \sum_{j_k=1}^{n_k}\matcoreTT{A}_k(i_k, j_k)\kron\matcoreTT{X}_k(j_k) = \matcoreTT{Y}_k(i_k).		
		\] 
		If we show that $\matcoreTT{Y}_0(\ell)Y_1(i_1)$ is equal to $\matcoreTT{W}_1(i_1)$, then the thesis holds true. Remark now that $\matcoreTT{H}_1(i_1, j_1)$ can be expressed equivalently as
		\[
			\matcoreTT{H}_1(i_1, j_1) = \begin{bmatrix}
			1 & \alpha_{\ell}
			\end{bmatrix}	\begin{bmatrix}
			\matcoreTT{C}_1(i_1, j_1) & \0 \\
			\0 & \matcoreTT{G}_1(i_1, j_1) 
			\end{bmatrix} 
			 = \matcoreTT{A}_0(\ell, \ell)\matcoreTT{A}_1(i_1,j_1).	
		\]
		from Equation~\eqref{eq:C:Acores} and~\eqref{eq:C:Acores1}. Thanks to this last equation, we have that
		\[
		\begin{split}
		\Bigl(\matcoreTT{H}_1(i_1, j_1)\Bigr)\kron\Bigl(\matcoreTT{X}_0(\ell)\matcoreTT{X}_1(j_1)\Bigr) &=
		\Bigl(\matcoreTT{A}_0(\ell, \ell) \matcoreTT{A}_1(i_1, j_1)\Bigr)\kron\Bigl(\matcoreTT{X}_0(\ell)\matcoreTT{X}_1(j_1)\Bigr)\\ &=\Bigl(\matcoreTT{A}_0(\ell, \ell)\kron \matcoreTT{X}_0(\ell)\Bigr)\Bigl(\matcoreTT{A}_1(i_1, j_1))\kron \matcoreTT{X}_1(j_1)\Bigr)
		\end{split}
		\]
		by the mixed-product property of the Kronecker product. Summing over index $j_1$ the previous equation leads to
		\[
		\begin{split}
		\matcoreTT{W}_1(i_1) &= \sum_{j_1=1}^{n_1}\Bigl(\matcoreTT{H}_1(i_1, j_1)\Bigr)\kron\Bigl(\matcoreTT{X}_0(\ell)\matcoreTT{X}_1(j_1)\Bigr) \\
		&=  \sum_{j_1=1}^{n_1}\Bigl(\matcoreTT{A}_0(\ell, \ell)\kron \matcoreTT{X}_0(\ell)\Bigr)\Bigl(\matcoreTT{A}_1(i_1, j_1))\kron \matcoreTT{X}_1(j_1)\Bigr)\\
		&= \matcoreTT{Y}_0(\ell)\matcoreTT{Y}_1(i_1),
		\end{split}
		\]
		from Equations~\eqref{eq:C:Ycores} and~\eqref{eq:C:Ycore1}
		i.e., the thesis.
	\end{proof}

 By the result of Lemma~\ref{lm:2}, we get that the $\ell$-th slice of $\ten{A}\ten{x}$ writes
	\[
		(\ten{A}\ten{x})^{[\ell]} = \ten{C}\ten{x}^{[\ell]}.
	\]
	Therefore the $\ell$-th slice of $\ten{r}$ is 
	\[
		\ten{r}^{[\ell]} = \ten{C}\ten{x}^{[\ell]} - \ten{b}^{[\ell]}.
	\]
	
	As conclusive result of this construction, we want to show that
	
	\[
		||\ten{r}||^2 = \sum_{\ell = 1}^{p}||\ten{r}^{[\ell]}||^{2}.
	\]
	
	\begin{lemma}
		\label{lm:1}
		Given $\ten{s}\in\R^{n_0\times n_1\times \dots \times n_d}$ and its $i_0$-th slice with respect to the first mode $\ten{s}^{(i_0)}$ then 
		\[
		||\ten{s}||^2 = \sum_{i_0 = 1}^{n_0}||\ten{s}^{[i_0]}||^2.
		\]
	\end{lemma}
	\begin{proof}
	Let $\ten{s}\in\R^{n_0\times n_1\times \cdots \times n_{d}}$ be a  $(d+1)$-order tensor expressed in TT-format with TT-cores 
	$\coreTT{{s}}_{k}\in\R^{r_{k}\times n_k\times r_{k+1}}$ with $r_0 = r_{d+1} = 1$, such that
		\[
			\ten{s} = \coreTT{s}_0\cdots \coreTT{s}_{d}.
		\]
		
		Let define the TT-expression of $\ten{s}^{(i_0)}\in\R^{n_1\times \cdots \times n_{d}}$ is
		\begin{equation*}
			\ten{s}^{[i_0]} = \matcoreTT{S}_0(i_0)\coreTT{s}_1\cdots \coreTT{s}_{d}.
		\end{equation*}

		To have a correct TT-representation of $\ten{s}^{[i_0]}$ we define its TT-cores $\coreTT{s}_{i_0,k}\in\R^{r_{k-1}\times n_k \times r_k}$ as follows

		\begin{align*}
			\coreTT{s}_{i_0,1} &= \matcoreTT{S}_0(i_0)\coreTT{s}_1 \;\in\R^{1\times n_1\times r_2}\\
			\coreTT{s}_{i_0, k} &= \coreTT{s}_{k}\;\in\R^{r_{k}\times n_k\times 	r_{k+1}}\qquad\text{for}\qquad k\in\{2,\dots, d\}.
		\end{align*}
		Compute now the norm of $\ten{s}$ as
		\begin{equation*}
			\begin{split}
			\label{eq:4}
			||\ten{s}||^2 &= \sum_{i_0, j_1,\dots, j_d = 1}^{n_0,n_1\dots, n_d}\matcoreTT{S}_0(i_0)\matcoreTT{S}_1(j_1)\cdots \matcoreTT{S}_d(j_d)\matcoreTT{S}_0(i_0)\matcoreTT{S}_1(j_1)\cdots \matcoreTT{S}_d(j_d)\\
			&=\sum_{i_0,j_1,\dots, j_d = 1}^{n_0,n_1,\dots, n_d}\Bigl(\matcoreTT{S}_0(i_0)\kron \matcoreTT{S}_0(i_0)\Bigr)\Bigl(\matcoreTT{S}_1(j_1)\kron \matcoreTT{S}_1(j_1)\Bigr)\cdots 	\Bigl(\matcoreTT{S}_d(j_d)\kron \matcoreTT{S}_d(j_d)\Bigr).
			\end{split}
		\end{equation*}
		Applying the mixed-product property of the Kronecker product to the first two matrix product, this last equation writes as
		\begin{equation*}
			\begin{split}
			||\ten{s}||^2 &=\sum_{i_0,j_1,\dots, j_d = 1}^{n_0,n_1,\dots, n_d}\Bigl(\matcoreTT{S}_0(i_0)\kron \matcoreTT{S}_0(i_0)\Bigr)\Bigl(\matcoreTT{S}_1(j_1)\kron \matcoreTT{S}_1(j_1)\Bigr)\cdots 	\Bigl(\matcoreTT{S}_d(j_d)\kron \matcoreTT{S}_d(j_d)\Bigr)\\
			&= \sum_{i_0=1}^{n_0}\sum_{j_1, \dots, j_d = 1}^{n_1,\dots, n_d}\Bigl(\matcoreTT{S}_0(i_0)\matcoreTT{S}_1(j_1)\kron \matcoreTT{S}_0(i_0)\matcoreTT{S}_1(j_1)\Bigr)\cdots \Bigl(\matcoreTT{S}_d(j_d)\kron \matcoreTT{S}_d(j_d)\Bigr)\\
			&= \sum_{i_0=1}^{n_0}\sum_{j_1, \dots, j_d = 1}^{n_1,\dots, n_d}\Bigl(\matcoreTT{S}_{i_0,1}(j_1)\kron \matcoreTT{S}_{i_0,1}(j_1)\Bigr)\cdots \Bigl(\matcoreTT{S}_{i_0,d}(j_d)\kron \matcoreTT{S}_{i_0,d}(j_d)\Bigr)\\
			&= \sum_{i_0=1}^{n_0}||\ten{s}^{[i_0]}||^2.
			\end{split}
		\end{equation*}
		i.e., the thesis.
	\end{proof}
	
	By the result of Lemma~\ref{lm:1}, we have
	\[
		||\ten{r}||^2 = \sum_{\ell = 1}^{p}||\ten{r}^{[\ell]}||^{2}.
	\]

	Once the ``all-in-one'' system has been completely described in its construction in TT-format, we present a further result related to Proposition~\ref{prop:eta_Ab}.
	\begin{corollary}
	\label{cor:eta_Ab:II}
	Under the hypothesis of Corollary~\ref{cor:eta_Ab}, if there exists a $k^\dagger\in\N$ such that $\norm{\ten{x}^{[\ell]}_k}~\le~ \norm{\ten{A}^{-1}}_2\sqrt{p}$ for every $k\ge k^\dag$, then
	\begin{equation}
	\label{eqCP2:T2}
	\eta_{\ten{A},\ten{v}}(\ten{x}_k)\,\rho^\dag \ge \eta_{\ten{A}_\ell, \ten{v}_\ell}(\ten{x}^{[\ell]}_k)\qquad\text{where}\qquad \rho^\dag = \frac{\sqrt{p}}{2-\nu}(1+\kappa_2(\ten{A}))
	\end{equation}
	with $\kappa_2(\ten{A}) = ||\ten{A}||_2||\ten{A}^{-1}||_2$ for every $\ell\in\{1,\dots, p\}$ and for every $k\in\N$ such that $k \ge k^{\ddagger}$ where $k^{\ddagger} = \max \{k^{**}, k^\dag\}$ with $k^{**}$ given in Corollary~\ref{cor:eta_Ab}.
\end{corollary}

\end{document}